\documentclass[11pt]{article}

\usepackage{setspace}
\usepackage[margin=1.25in]{geometry}

\usepackage[english]{babel}

\usepackage{amsmath}
\usepackage{amsfonts}
\usepackage{amssymb}
\usepackage{amsthm}
\usepackage{graphicx}
\usepackage{tikz-cd}
\usepackage{mathrsfs}
\usepackage{xfrac}
\usepackage{url}
\usepackage{color}
\usepackage[colorlinks=true, citecolor=blue]{hyperref}
\usepackage{enumitem}
\usepackage[usestackEOL]{stackengine}
\usepackage{accents}
\usepackage{bbm}
\usepackage{authblk}

\numberwithin{equation}{section}
\def\a{\alpha}

\def\be{\begin{equation}}
\def\ee{\end{equation}}
\def\ba#1\ea{\begin{align}#1\end{align}}
\def\no{\nonumber\\ }
\def\ra{\rangle}
\def\la{\langle}

\def\s{\sigma}
\def\vol{\vartheta}

\def\F{\scr F}

\newcommand{\ca}[1]{\mathcal{#1}}
\newcommand{\bb}[1]{\mathbb{#1}}
\newcommand{\fr}[1]{\mathfrak{#1}}
\newcommand{\scr}[1]{\mathscr{#1}}

\newcommand{\wt}[1]{\widetilde{#1}}

\newcommand{\wh}[1]{\widehat{#1}}
\newcommand{\ol}[1]{\overline{#1}}
\newcommand{\un}[1]{\underline{#1}}

\newcommand{\iprod}{\mathbin{\lrcorner}}

\newtheorem{theorem}{Theorem}
\newtheorem{lemma}{Lemma}
\newtheorem*{assumption}{Assumptions}
\newtheorem{corollary}{Corollary}

\graphicspath{{Figures/}}

\title{{\huge Quantizing the guiding center: \\do quantization and coarse graining commute?
\vspace{1cm}}}
\date{August 25, 2026}
\author[1]{Rodrigo Andrade e Silva\thanks{andradeesilvarodrigo@gmail.com}}
\author[2]{Ted Jacobson\thanks{jacobson@umd.edu}}
\affil[1]{\it \small Perimeter Institute for Theoretical Physics, Waterloo, ON, Canada}
\affil[2]{\it \small Maryland Center for Fundamental Physics, University of Maryland, College Park, MD, 20742 USA}

\begin{document}
\begin{titlepage}
\maketitle
\thispagestyle{empty}

\begin{abstract}
We develop a nonperturbative, group-theoretic quantization of the effective guiding center theory of a charged particle drifting in a magnetic field in two spatial dimensions, and compare the resulting quantum theory with a corresponding coarse graining of the underlying microscopic theory. In the classical effective theory, the small gyro motion is not resolved, while the motion of the center of the gyro orbit remains observable. The reduced phase space is the physical space itself, so quantization leads to noncommuting spatial coordinates, and the effective theory loses access to the metric structure of physical space, retaining only its area structure. 
By formulating a prescription to match between the microscopic and effective quantum theories, we find that the predictions of the quantized effective theory are generally consistent with those of the microscopic theory. However, for closed isomagnetic contours the effective theory predicts a quantization of ``radius'', a spatial discreteness  absent from the microscopic theory. 
This illustrates that quantization of an effective theory may yield spurious nonperturbative predictions, even if that quantum theory shows no internal signs of breakdown.
\end{abstract}

\end{titlepage}

\tableofcontents


\newpage
\section{Introduction}
\label{sec:intro}

If the underlying structure of a system is governed by quantum mechanics, then 
an autonomously evolving collection of effective degrees of freedom of the system
is expected to be governed by quantum mechanics as well, 
since it is built from the quantum variables and inherits
commutation relations. 
Indeed this is seen to be the case in few body systems,
in condensed matter, in nuclear physics, and in the standard model of
particle physics. It has been essential to the 
development of physics, 
since at a given stage we are always dealing with an effective 
theory of some sort. 

In the cases of (nonrelativistic) particle mechanics and electromagnetism,
where the classical limit of the effective theory was 
known prior to its quantum version, quantization has been a magic wand 
that uplifts an established classical theory to a valid effective quantum one.
Although quantization is an inherently ambiguous recipe, 
the symmetry structure of the classical theory and its Hamiltonian
has proved sufficient to guide physicists to a quantization that describes Nature.
Whether this pattern will persist in the case of general relativity is 
of course not known. It is widely expected that quantized general relativity is incomplete in the UV, 
and that the true underlying degrees of freedom and their dynamics look very different
from those of general relativity.
While it is hard to imagine that 
perturbatively quantized general relativity is not
a correct description of nature at some level of 
approximation, what is less clear is whether 
non-perturbative effects arising from the quantized 
effective theory, such as those due to the global 
structure of the phase space or of chosen quantizing observables, would accurately reflect 
properties of the underlying theory. 

The  more general question is to what extent does a non-perturbative quantization of an effective classical theory give predictions compatible with the underlying microscopic quantum theory in the suitable regime. In particular, could non-perturbative effects stemming from global aspects of the phase space of the effective theory furnish useful insights into the microscopic quantum theory, or could they lead to spurious predictions that do not correspond to properties of the microscopic theory? Moreover, since quantization often relies on background structures to resolve ambiguities, what happens when the background structures of the microscopic and effective theories are distinct? Notice that the effective theory may ``see'' \emph{more} background structures (produced by ``frozen'' microscopic degrees of freedom) 
or \emph{fewer} (if the effective theory becomes insensitive to aspects of the microscopic background structure). In the challenging case of general relativity, it is not currently known which of these
possibilities in fact occurs, nor how the nature and predictions of the 
resulting quantum theory might depend on them.

The questions above can be phrased in terms of the following diagram:
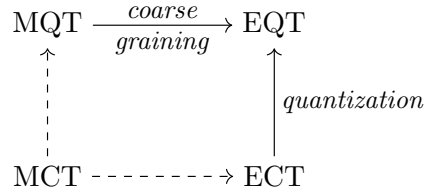
\begin{center}
\begin{tikzpicture}[commutative diagrams/every diagram]
\node (QT) at (-1.5,0) {{MQT}};
\node (EQT) at (1.5,0) {{EQT}};
\node (CT) at (-1.5,-2) {{MCT}};
\node (ECT) at (1.5,-2) {{ECT}};

\path[commutative diagrams/.cd, every arrow, every label]
(CT)  edge[dashed] (ECT)
(CT)  edge[dashed] (QT)
(QT)  edge node[above] {\small \emph{coarse}}
           node[below] {\small \emph{graining}} (EQT)
(ECT) edge node[right] {\small \emph{quantization}} (EQT);
\end{tikzpicture}
\end{center}
The relevant arrows, in general, are the ``quantization'' of the {\it Effective Classical Theory} ({ECT}) and the ``coarse graining'' of the {\it Microscopic Quantum Theory} ({MQT}), and the questions all relate to whether they end on the same {\it Effective Quantum Theory} ({EQT}). In some cases, there may exist a {\it Microscopic Classical Theory} ({MCT}) which produces both the ECT and the MQT, via coarse graining and quantization, respectively. In that case, the questions can be more sharply phrased in terms of the commutativity properties of the diagram.

In this paper we explore, in a peculiar example, 
the relation between the non-perturbative quantization of a classical 
effective theory and the effective description of the underlying microscopic
quantum theory.
An effective classical theory 
typically has a kinetic energy term in the Lagrangian, which leads to 
canonically conjugate configuration and momentum pairs, and
second order equations of motion. 
Here we examine an effective
classical theory that has no kinetic energy, and that satisfies a
first order dynamical equation determining the velocity in terms of
the position. The theory is the guiding center approximation for 
the motion of a charged particle confined to a two-dimensional plane, 
in a magnetic field normal to the plane (cf.~Sec.~\ref{SecGuidingApprox}). 
In this approximation, the gyro motion of the 
charge is unresolved, the adiabatically invariant magnetic moment of the
gyro motion is treated as exactly conserved, 
and one follows only the motion of the center of the gyro orbit,
i.e., the so-called ``guiding center'' 
motion \cite{alfven1950cosmical, northrop1963adiabatic}. 
In this two-dimensional case, that motion is purely
the ``drift velocity'', which is determined algebraically 
by the gradient in the magnetic field strength.
The canonical momenta are entirely constrained, so the
reduced phase space is the spatial plane itself, whose coordinates 
have nonvanishing Dirac brackets, determined by the magnetic field strength (cf.~Sec.~\ref{SecGuidingPS}).
After quantization, this implies that the particle lives on a noncommutative plane.

The method of canonical quantization can be
applied to this guiding center  system, although quantization 
is more ambiguous than usual since no spatial metric 
appears in the Hamiltonian so 
there is no Euclidean symmetry group to preserve and realize 
unitarily  (cf.~Sec.~\ref{SecQuantization}). 
In the guiding center
theory, the only background structures appearing in the action are the 
vector potential, characterized by its gauge-invariant field strength two-form $F$, and 
the magnitude $B$ of the magnetic field, which we assume is nowhere vanishing. 
For different magnetic fields, what remains the same is 
the area 2-form $\vartheta:= B^{-1}{F}$, so we
regard $\vartheta$ as a background structure in the effective theory, along with $F$ 
(which is proportional to the symplectic form on the phase space). 
We constrain the quantization ambiguity 
by requiring that the Poisson algebra to be quantized includes a charge that 
symplectically generates
a $\vartheta$-preserving flow on the phase space. 
In practice, this implies that we quantize
a function of the magnetic field strength. 

For a constant magnetic field, we choose any canonically conjugate pair of spatial 
coordinates, and find that the effective theory reproduces the density of states
of a single Landau level (cf. Sec.~\ref{constantBeff}).
In the case of open isomagnetic contours with $dB$  nowhere zero,
we quantize $\log B$, which
leads to a Heisenberg algebra (cf. Sec.~\ref{openBlines}). The Heisenberg equation of motion matches the
classical one.
In the case of closed isomagnetic contours we instead use
the magnetic flux through the contours
as a coordinate, which leads to a centrally-extended Euclidean algebra in 2-dimensions (cf. Sec.~\ref{closedBlines}). Here we find that 
the flux is quantized, which can be interpreted as a quantization of the ``radial position'' of the particle, a feature not present in the microscopic theory. 
The spatial discreteness implied by the 
``flux gap" is never larger than the gyro scale of the microscopic theory, 
so it presents no contradiction above the guiding center resolution 
length scale.
Nonetheless, it demonstrates
that an effective theory may not necessarily reveal its own limitations, so 
a model of the world based on an effective theory
may be misleading about the nature of the microscopic domain.

Finally, we explore to which degree the quantization-coarse graining diagram above commutes in this setting (cf.~Sec.~\ref{SecMatch}). Having quantized the classical guiding center theory, and given the well-known quantization of the classical microscopic theory of charged particles in a magnetic field, we have a  well-controlled situation where all corners of the diagram are known and sufficiently simple. To close the diagram, we thus study how to obtain the quantized guiding center theory from the microscopic quantum theory, finding a consistent relation between these two levels of description in the suitable regime. In the case of a uniform magnetic field, we propose a correspondence between a fixed Landau level and the Hilbert space of the effective theory at fixed magnetic moment $\mu$ (cf. Sec.~\ref{MatchConstantB}). From this, in particular, we suggest an explanation for why the effective quantum theory predicts a discretization of radius in the case of closed isomagnetic contours, despite space being continuous in the microscopic description.
Then, most importantly, we consider the case of non-uniform magnetic field, extending the matching prescription by identifying a candidate subspace of the microscopic Hilbert space that can be associated with the state space of the effective theory (cf. Sec.~\ref{nonuniformBmatch}). 
We analyze rigorous conditions for adiabatic stability of the matching prescription, 
meaning that a state initially in this subspace remains approximately within the subspace under time evolution. We show by a rough estimation, for a simple class of magnetic fields, that the matching prescription is dynamically consistent for long time evolution in the regime of the guiding center approximation.
We also propose a prescription for mapping operators in the microscopic theory to operators in the effective theory, based on time coarse-graining and projection to the matched subspace.
Then, we perturbatively analyze the drift velocity 
to first order in the magnetic field gradient,
showing agreement with the prediction of the effective quantum theory 
(cf. Sec.~\ref{FirstOrderPertAnalysis}). 

We end the main text  with a summary of the paper, a discussion of implications of the results, and some suggestions for further investigations (cf.~Sec.~\ref{SecDiscussion}).  
Five appendices are included. In App.~\ref{AppClassicalGCconditions} the derivation of the classical guiding center approximation
is reviewed. In App.~\ref{AppTorus} the spatial plane is periodically identified, forming a non-commutative
torus, and a quantization of this phase space is derived. 
In App.~\ref{AppClosedBAlt} we explore an alternative quantization for the case of closed isomagnetic contours based on the group $\text{PSL}(2, \bb R)$.
In App.~\ref{AppLLmatrix} we derive some useful formulas
for matrix elements of functions of position between Landau level states. In App.~\ref{AppAdiabIneq} we present
a lengthy derivation establishing rigorous bounds characterizing the validity of the adiabatic 
approximation in the microscopic quantum theory, derive explicit estimates for these bounds in terms of the microscopic Hamiltonian and magnetic field, and use them to quantify the long-time adiabatic stability of the matching prescription.

\section{Guiding center approximation}
\label{SecGuidingApprox}

At the classical level, the motion of a nonrelativistic 
charged particle subjected to a magnetic field 
is described by the action 
\be\label{microS}
S = \int\!dt \left( \tfrac{1}{2}m\, g_{IJ}\,\dot x^I\dot x^J + q\, A_I\,\dot x^I\right).
\ee
Here $m$ and $q$ are the mass and electric charge of the particle, $x^I$ ($I=1,2,3$) denotes its position in three dimensional space, over-dot denotes derivative with respect to $t$, 
$g_{IJ}(x)$ is the spatial metric, and $A_I(x)$ is the magnetic potential 1-form.
Throughout the paper, a sum over repeated pairs of Latin indices is implicit.
As written the action is valid for any coordinate system $x^I$.
If the metric is flat, and we choose Cartesian coordinates so that $g_{IJ}=\delta_{IJ}$, 
the equations of motion take the form of the Lorentz force law,
\be\label{exactEOM}
m\ddot x^I = qF_{IJ} \dot x^J\,,
\ee
where
\be
F_{IJ} = \partial_I A_J-\partial_J A_I
\ee
is the (coordinate components of the) magnetic flux  2-form. 
The magnetic field vector is $B^I=\frac12 \epsilon^{IJK}F_{JK}$,
where $\epsilon^{IJK}$ is the alternating symbol in three dimensions (i.e., it is odd under index permutations and $\epsilon^{123}=1$).
Here $\partial_I$ denotes $\partial/\partial x^I$, and the Cartesian coordinate
indices can be raised and lowered without consequence.

If the magnetic field vector is constant, the particle spirals around a magnetic field 
line with constant angular velocity, and constant linear velocity along the field line direction. The former component
of the motion is called the gyro orbit.
If the magnetic field  is not constant, the motion is more complicated, but an approximation
becomes available if the gyro radius is small compared to the length scale over which the
field changes. The particle then executes gyro motion around the local magnetic field direction,
while the center of the gyro orbit primarily follows the magnetic field lines, with 
a slow transverse drift velocity caused by derivatives of the 
field line strength or direction. 
We first describe how this works in the three-dimensional case,
and then specialize to two dimensional motion. 

In the so-called {\it guiding center approximation} 
one writes an equation of motion for the center of the gyro orbit, without resolving
the gyro motion. An effective Lagrangian for the guiding center motion can be obtained
by decomposing the kinetic energy into the contributions from velocity 
components parallel and perpendicular 
to the magnetic field at each point. The parallel contribution is kept unchanged,
while the perpendicular contribution is replaced by the magnetic dipole energy
$\mu B$, and treated as a potential energy. The magnetic moment $\mu$ associated with the 
gyro motion is proportional to the angular momentum about the field line, and is an 
adiabatic invariant, so it is approximately conserved provided that the magnetic field is
nearly constant over the distance the particle travels during a gyro orbit. It is treated as a constant in the guiding center approximation. The  drift velocity contribution to the kinetic energy is subleading and therefore neglected. It turns out that, precisely because
the Lagrangian lacks a kinetic energy term for the transverse velocity, the Euler-Lagrange equations algebraically determine that velocity in terms of the other quantities.
This approach to deriving the guiding center approximation is reviewed 
in \cite{Jacobson:2024aap}.

In this paper we restrict to the case of 
motion in two spatial dimensions. This case 
is also governed by the action \eqref{microS},
but with the coordinate indices ranging over only $i=1,2$.
The magnetic field is then a (pseudo) scalar $B(x)$,
related to the flux 2-form via
\be\label{F}
F_{ij} = B \, \epsilon_{ij}\,,
\ee
where $\epsilon_{ij}$ is the area 2-form, normalized via $\epsilon_{ij}\epsilon_{kl}g^{ik}g^{jl}=2$.
If $B$ is constant, the orbits are circular, with angular frequency
\be
\Omega_B := q B/m\,,
\ee
called the {\it gyro frequency}. This is the same for all orbits,
but the {\it gyro radius}, i.e, the radius
of the circular orbit, depends on the angular kinetic energy. 
For non-constant $B(x)$, 
the guiding center approximation for two-dimensional motion 
describes only drift velocity, since there is no direction `along the field lines'.
As reviewed in App.~\ref{AppClassicalGCconditions}, the conditions for validity of this approximation are
\be
\rho^i\partial_i B\ll B,\qquad \rho^i\rho^j\partial_i\partial_jB\ll B\,,
\ee
$\rho^i$ is the displacement from the center of the gyro orbit, and absolute values are
implicit. The two-dimensional case of guiding center motion was 
discussed in \cite{Witten:1978rg}, and quantum aspects of this system have been discussed 
in \cite{chan2016quantum, Chan:2017dex,
Klauder:1996rj}. 

\section{Guiding center phase space}
\label{SecGuidingPS}

To canonically quantize the two-dimensional 
motion in the guiding center approximation, 
we must first characterize the corresponding phase space.
To this end we formulate the dynamics in terms of 
an action principle, and derive the symplectic structure from that. 
The action can be found by coarse-graining the microscopic action \eqref{microS},
as described in \cite{Jacobson:2024aap}. 
In two dimensions this reduces to 
\be\label{macroS}
S = \int\!dt \left(q\, A_i\,\dot x^i - \mu B\right),
\ee
where $x^i$ now represents the guiding center of the gyro motion, 
$\mu$ is the magnitude of the adiabatically conserved magnetic moment 
which is treated as a fixed constant in the macroscopic theory, 
and $B$ is the magnitude of 
the magnetic field.\footnote{The energy of a magnetic dipole in a magnetic field 
is in general $-\vec \mu\cdot\vec B$. 
For gyro motion $\vec\mu$ is antiparallel to 
$\vec B$, for either sign of the charge $q$, so $\vec \mu\cdot\vec B<0$ always,
hence the magnetic dipole energy is $|\vec\mu||\vec B|$.} From here on we will
restrict to $B>0$, since the guiding center approximation is never justified if $B=0$, 
so we need not distinguish $B$ from its magnitude. Moreover, 
to simplify the equations we may as well 
adopt units with $q=1$ at this point. (In order to later compare with the
microscopic theory in which $\mu$ is not fixed, we do not set $\mu=1$.) 

Unlike the microscopic action \eqref{microS}, the guiding center action
\eqref{macroS} is independent of the spatial metric, so Cartesian coordinates
have no preferred status in the theory. 
The one residue of the metric in the theory is the relation \eqref{F} 
between $B$ and $F_{ij}$. In the microscopic theory, we normalized
the area element $\epsilon_{ij}$ using the spatial metric. 
In the macroscopic
theory, the only background structures appearing in the action are the 
vector potential, characterized by its gauge-invariant field strength $F_{ij}$, and 
the magnitude $B$ of the magnetic field. 
\emph{In all that follows we shall assume that $F$ is a smooth nondegenerate 2-form and that $B$ is a smooth positive function.}
Given these two objects
we may  identify the area
2-form, 
now denoted by $\vol$, via \eqref{F},
\be
\vol := \frac{F}{B} \,.
\ee

To pass to the Hamiltonian formulation we start by evaluating the conjugate momenta,
\be
p_i := \frac{\partial L}{\partial \dot x^i} = A_i(x)\,.
\ee
Since this defines the momenta as functions of the coordinates, it actually
corresponds to a pair of primary constraints
\be
C_i := p_i -  A_i = 0\,.
\ee
The theory of constrained Hamiltonian systems was first 
worked out by Dirac and Bergmann 
(see \cite{Brown:2022iez} for a concise summary and references).
According to the Dirac-Bergmann algorithm,
the equations of motion are equivalent to 
Hamilton's equations with 
the ``primary Hamiltonian'' 
\be
H = p_i\dot x^i - L + \lambda^iC_i = \mu B + \lambda^iC_i
\ee
where $\lambda^i$ is a Lagrange multiplier 
whose equation of motion enforces the constraints.
The condition that the constraints are preserved in time
imposes the relation
\be\label{lambdaeq}
\lambda^jF_{ij}= \mu \partial_i B 
\ee%
which determines the Lagrange multiplier,
and the equation of motion for $x^i$ is 
\be
\dot x^i = \frac{\partial H}{\partial p_i} = \lambda^i.
\ee
We thus have
\be\label{drift}
\dot x^i = {}-\frac{\mu}{B} \vol^{ij} \partial_j B \,,
\ee
where $\vol^{ik}\vol_{jk}:={\delta^i}_j$.
It is also useful to express the equation of motion
directly in terms of the magnetic flux 2-form 
as 
\be\label{eomF}
\dot x \iprod F = {}- \mu \,dB\,,
\ee
where $\iprod$ denotes contraction on the
first slot of the 2-form $F$ which is written here using index-free notation.
Note that \eqref{eomF} is a fully tensorial equation,
in the sense that it involves no restriction whatsoever on the coordinate system.

Instead of working with the primary Hamiltonian and the
phase space coordinatized by $(x^i,p_j)$, one may pass to 
the reduced phase space in which the constraints are 
imposed ab initio. One restricts to the constraint surface
$p_i =  A_i(x)$, which is coordinatized by $x^i$ alone, 
and the original symplectic form $dp_i\wedge dx^i$ is replaced by 
\be\label{omega}
\omega = d(A_i)\wedge dx^i = F\,.
\ee
The corresponding Poisson bracket is called the {\it Dirac bracket},
which we denote by $\{\,,\,\}$.
The Dirac bracket
between functions $f(x)$ and $g(x)$ is given by
\be
\{ f, g\} = - \frac{1}{B} \vol^{ij} \frac{\partial f}{\partial x^i} \frac{\partial g}{\partial x^j}\,.
\ee
In particular, the Dirac bracket of the spatial coordinates is
\be\label{calg}
\{ x^i, x^j\} = - \frac{1}{B} \vol^{ij}\,,
\ee
so the position 
coordinates do not Poisson commute with each other.
In this reduced phase space formalism, the Hamiltonian is just the magnetic 
dipole energy, 
\be\label{H}
H = \mu B\;.
\ee
In this reduced phase space, the equation of motion in the form
\eqref{eomF} corresponds to Hamilton's equations in differential 
forms notation, $\dot x\iprod\omega = -dH$. 

\section{Quantization of the effective theory}
\label{SecQuantization}

We adopt the viewpoint that canonical quantization amounts to identifying a
Poisson-algebra of observables that is complete (so that 
any observable can be expressed as a function of them),
and constructing a unitary irreducible representation 
of the corresponding commutator-algebra of operators. 
In order to fully respect the structures in the classical theory,
the observables should be globally defined on the phase space, 
Casimir invariants of the classical and quantum algebras should match,
and the Poisson algebra should  exponentiate to a group action on the phase space.\footnote{The last point is central to Isham's group-theoretic approach to quantization \cite{isham1984topological, isham1989canonical}. 
The motivation is that, in quantum mechanics,
a finite dimensional Lie algebra of self-adjoint observables (under suitable regularity assumptions) always exponentiates to a group acting unitarily on the Hilbert space, while 
a Poisson algebra of classical observables 
does not always exponentiate into a group of symplectomorphisms on the phase space.
If this exponentiability property is not demanded for the classical algebra, one of the issues generally encountered is that the spectrum of the quantized observables will include values with no classical counterpart.
For example, quantizing the canonical algebra $\{x, p\} = 1 \mapsto [\wh x, \wh p] = i\hbar$ for the phase space of a particle on the positive half-line, $T^*\bb R^+$, would lead to the standard Heisenberg representation where the spectrum of $\wh x$ is the whole line, $\bb R$; this is directly linked to the fact that, classically, the flow of $p$ would try to move $x$ across the origin, and thus would not yield a group action on the half-line.}
We call such an algebra a ``quantizing algebra'', and the elements  of it 
the ``quantizing observables''.

Note that we are not allowed to just pick 
generic ``area-adapted'' coordinates $x^1$ and $x^2$ (i.e., satisfying $\vol = dx^1 \wedge dx^2$) as the quantizing algebra because, as seen from \eqref{calg}
\be\label{calg2}
\{ x^1, x^2\} = - \frac{1}{B}\,,
\ee
they generally do not form an algebra in the first place, even if we  include $B^{-1}$ as a third element. 
Also, we are not generally allowed to pick 
canonically-conjugate
coordinates $\gamma^1$ and $\gamma^2$ (i.e., satisfying $\{\gamma^1, \gamma^2\} = 1$) because, even though 
the Darboux theorem 
guarantees that 
such coordinates exist locally
on the phase space, 
they may not exist globally,
and even if they do exist
globally they may not 
range  from $-\infty$ to $+\infty$.\footnote{The reason for the demand on the range is that $\gamma^2$ generates an unrestricted translation of $\gamma^1$, so unless $\gamma^1$ (and $\gamma^2$) have range $\bb R$, the algebra will not exponentiate to a group action on the phase space.}
Thus, we must consider another way to select a suitable set of quantizing 
observables.

Experience shows that quantization is more likely to produce a 
 theory that accords with nature when the quantizing algebra
is constructed using the structures present in the classical 
theory and respecting the symmetries among them.
From the viewpoint of the effective theory,
we do not have access to the metric structure; only the magnetic scalar field $B$ and the magnetic flux 2-form $F$ are available.
But experimentation varying the magnetic field
would reveal
that $F$ and $B$ are not
independent, and that in fact the area form $\vol = F/B$ 
is a fixed, fundamental structure of physical space.
We thus take the area form as the background structure
to be respected as far as possible in the quantization.

Let $Q$ be one of the candidate quantizing
observables,
and $X$ be its symplectic flow, defined by $dQ = -X\iprod\omega$. 
The symplectic flow preserves $\omega = F$,
so if we demand that it also preserve the physical area $\vol = F/B$ then
it must preserve $B$, so $X\iprod dB = 0$. 
Since the phase space is two-dimensional
we have $\omega\wedge dB=0$, 
hence\footnote{We use here the fact that 
$\iprod$ is an antiderivation on the exterior algebra.}
\be
0 = X\iprod(\omega\wedge dB) = -dQ\wedge dB\,.
\ee
So, essentially,  $Q$ must be a function of $B$.\footnote{More precisely, $Q$ is a function of $B$ in any region where $dB \ne 0$; and vice versa where $dQ \ne 0$.} Of course, we cannot have a complete set of observables that are all functions of $B$, but we propose that at least one of the canonical observables should be a function of $B$. This is analogous to the usual quantization of a particle on a Euclidean space, where among all sets of conjugate coordinates $q$ and $p$, one chooses those where $p$ generates isometries of the physical space.

The criterion above does not seem to allow us to develop a universal quantization of the effective theory, as we must make 
some assumptions
about the topology of the {\it isomagnetic contours} (i.e., curves, or regions, of constant $B$) and the range of $B$.
We consider three
cases: constant $B$, and open and closed isomagnetic contours (assumed to be curves foliating the space), that we quantize differently.
In the constant case, $Q$ would be a constant and thus not furnish a useful coordinate, so 
we resort to any pair of area-adapted coordinates (which, in this case, do form an algebra together with $1$).
In the open case we use $Q =\log B$ as a coordinate, which requires that each value of $B$ corresponds to a unique contour.\footnote{It 
may be possible to piece together the quantization on patches of space 
foliated by contours on which the $B$ varies monotonically, but we will not address this in this paper.}
In the closed case we use $Q(x) = \int_x F$, where the integral is the magnetic flux enclosed by the isomagnetic contour passing through $x$.

Despite the non-universality of the quantization, we can still make some general statements. For instance, one expects that for any realization of  area-adapted coordinates, $x^1$ and $x^2$, and a given operator-ordering for $B$, \eqref{calg2} would be represented in the quantum theory up to corrections of order $\hbar^2$,
\be
[ \wh x^1, \wh x^2] = - i\hbar\wh{B^{-1}} + \ca O(\hbar^2)\,,
\ee
revealing that space is non-commutative in the guiding center theory.
This leads to an uncertainty relation
\be\label{spaceNC}
\Delta x^1 \Delta x^2 \ge \frac{\hbar}{2}|\la \wh {B^{-1}}\ra| + \ca O(\hbar^2)\,.
\ee
This phenomenon can be understood from the microscopic perspective: as $\Delta x^1$ becomes smaller, $\Delta v^1$ becomes larger due to the standard Heisenberg uncertainty relation, $\Delta v^1 \ge \frac{\hbar}{2m\Delta x^1}$; 
and the classical gyro motion implies that there is a corresponding increase in the spread of $\rho$ by $\Delta\rho \sim \Delta v^1/\Omega$; so the uncertainty of $x^2$ should increase by $\Delta x^2 \sim \Delta \rho \ge \frac{\hbar}{2m \Omega \Delta x^1} = \frac{\hbar}{2B \Delta x^1}$, in agreement with \eqref{spaceNC}.

\subsection{Constant magnetic field}
\label{constantBeff}

As mentioned above, in this case a pair of area-adapted coordinates defines a Heisenberg 
algebra.
Thus, given any pair of such coordinates, $x^1$ and $x^2$, globally defined and ranging from $-\infty$ to $+\infty$, the brackets \eqref{calg2} could be quantized as
\be\label{constBCC}
[ \wh x^1, \wh x^2 ] = - \frac{i\hbar}{B}\,,
\ee
The only unitary irreducible representation of the Heisenberg algebra, according to the Stone-von Neumann theorem, is the familiar one on $L^2(\bb R)$,\footnote{The theorem
also assumes that the representation exponentiates to a strongly continuous, 
unitary irreducible representation of the Heisenberg group \cite{ReedSimon1980}.\label{stonevonneumann}}
carried by complex-valued square-integrable wavefunctions of either variable.
In the $x^1$-realization,
\ba\label{xyrep}
\wh x^1 \psi(x^1) &= x^1 \psi(x^1) \\
\wh x^2 \psi(x^1) &= \frac{i\hbar}{B} \frac{\partial\psi}{\partial x^1}(x^1)\,.\label{x2rep}
\ea

There is, yet, a fair deal of ambiguity due to different choices of area-adapted coordinates. 
Any pair of area-adapted coordinates, say $(x^1,x^2)$ and $(y^1,y^2)$, are related by an area-preserving diffeomorphism, which is a symplectic symmetry. 
While the quantizations of different canonical pairs will produce unitarily equivalent 
representations, due to Stone-von Neumann theorem, in physics one is also interested in how other operators are represented in the Hilbert space. Lacking any other structure to guide the choice of operator orderings, these other operators will be subject
to operator ordering 
ambiguities that can only refer to the choice of the canonical pair, leading to a coordinate-dependent quantization.
We contrast this with the case encountered in the standard quantization of a particle living on a Euclidean space, where the flat metric structure of the background space selects a preferred class of (linearly-related) canonical coordinates --- namely, those where the position variables are Cartesian coordinates (or, equivalently, where the momentum variables generate isometries of the space). In that case, the spatial
symmetry provides a physical structure that suggests
choices of operator orderings in the quantization of other observables.
In the present case, however,  an area structure alone
on physical space does not suffice
to select a preferred class of (linearly-related) canonical coordinates.

Having quantized a particular 
pair of area-adapted coordinates $(x^1,x^2)$, we can estimate the density 
of states. The argument is parallel to the one used for the density of states of a given Landau level in \cite{Landau:1991wop}.  Consider a rectangular region with coordinate lengths $a_1$ and $a_2$ in the $x^1$ and $x^2$ directions. (Since the coordinate operators do not commute, this is meaningful only for rectangles with $a_1 a_2\gg \hbar/B$.)  
If we impose periodic boundary conditions in the $x_1$ direction, the allowed Fourier wavenumbers for $\psi(x_1)$  are $k_n=2\pi n/a_1$, with integer $n$, whose spacing is $2\pi/a_1$. According to \eqref{x2rep}, these wavenumbers are equal to 
eigenvalues of $B\wh x^2/\hbar$, which in the rectangular region range from $0$ to 
$Ba_2/\hbar$. The number of wavenumbers that fit in this interval is 
$(Ba_2/\hbar)/(2\pi/a_1) = Ba_1a_2/2\pi\hbar$, so
the density of states is 
\be\label{densitystatesLandau1}
\text{\it density of states} = \frac{B}{2\pi\hbar}\,,
\ee
in agreement with the standard density of states for a Landau level.
The magnetic length $\ell_B=\sqrt{\hbar/B}$ sets the
non-commutativity scale in \eqref{constBCC} and thus determines the density 
of states \eqref{densitystatesLandau1}. 
In the microscopic theory it also sets the scale of a 
circular lowest Landau level wavefunction; hence, even though the 
gyro motion is not resolved in the effective theory, the length 
scale of the smallest quantum gyro motion governs the 
effective theory. 
We will explore in detail the match between the 
quantum effective theory and the microscopic theory in Sec.~\ref{SecMatch}.

A more precise approach to defining the density of states is to eliminate 
the boundaries before quantization, 
via periodic identification of $x^1=0$ with $x^1=a_1$, 
and $x^2=0$ with $x^2=a_2$.
This results in a phase space with toroidal topology, which 
can be quantized 
by taking advantage of the fact that the torus is a 
quotient of the plane by a discrete group ($\bb Z^2$) of symmetries.
We spell this out in 
App.~\ref{AppTorus}, showing how it leads to the symplectic flux quantization 
condition and to the density of states.

The dynamics for the case of a constant magnetic field 
is particularly simple, since the Hamiltonian is a constant.
To probe the nature and fidelity of this effective theory  
in a less trivial setting, we consider in 
 \cite{GCHallHawking} the effect of adding
an external potential, and apply it to the quantum Hall effect and to 
Stone's model of Hawking radiation of edge modes in a quantum Hall system \cite{Stone:2012cx}.
We find that the predictions of the effective quantization 
match precisely with the main results of the microscopic quantum theory.

\subsection{Open isomagnetic contours}
\label{openBlines}

If $B$ ranges monotonically in the open interval $(B_1,B_2)$, with $B_2>B_1>0$,
a natural coordinate to quantize is $\log [(B - B_1)/(B_2-B)]$, 
which ranges from $-\infty$ to $\infty$. 
The equations of motion \eqref{eomF} are invariant under a constant scaling 
$B\rightarrow \lambda B$, and  
this coordinate is invariant under
that scaling. In the case $B_1=0$ and $B_2\rightarrow\infty$,
this is equivalent to $\log B$ minus the infinite constant $\log B_2$.
For simplicity in the following we consider only the range $B\in (0,\infty)$,
and we take $\log (B/b)$ as one of the coordinates to be quantized. 
The constant $b$, with dimensions of magnetic field, is included
to render the argument of the logarithm dimensionless.
A constant scaling of $B$ corresponds to a constant shift of the coordinate 
$\log B/b$. As the value of $b$ is arbitrary and affects nothing,
we shall henceforth suppress it.

The conjugate variable $\chi$ is defined such that 
\be\label{symplectic}
\omega = d\log B \wedge d\chi\,.
\ee
To see that such a $\chi$ exists, let $y$ be any coordinate independent of
$B$. Since the phase space is two dimensional, there exists 
a function $f(B,y)$ such that $\omega= f(B,y)\, d\log B\wedge dy$.
In terms of the new coordinate 
$\chi(B,y) = \int_0^y f(B,y') dy'$ we have
$\omega=d\log B\wedge d\chi$, since $dB\wedge dB=0$. 
The coordinate $\chi$ is the flow parameter
along the phase space flow generated by $\log B$,
i.e.,  $\partial_\chi\iprod \omega = {} -  d\log B$. 
As we wish to treat $\log B$ and $\chi$ as conjugate variables amenable to canonical quantization, we assume that $\chi$ is globally defined on the phase space, with range from $-\infty$ to $+\infty$.\footnote{An equivalent condition referring directly to the magnetic field is the following. Let $\ca U(B_1, B_2)$ be the region between the isomagnetic contours at $B_1$ and $B_2$, and consider any cut traversal to the isomagnetic contours dividing this region into two halves, $\ca U_+(B_1, B_2)$ and $\ca U_-(B_1, B_2)$. Then $\chi$ is unbounded if and only if $\int_{\ca U_+(B_1, B_2)}\omega$ and $\int_{\ca U_-(B_1, B_2)}\omega$ both diverge for all $B_1\ne B_2$.}
The flow generated by
the Hamiltonian $\mu B$ is 
$(\partial_t\chi) \partial_\chi$, 
with 
\be\label{cchidot}
\partial_t\chi = \mu B\,.
\ee
Our choice of variables is thus similar to action-angle variables. 
The solution, $\chi = \chi_0 + \mu Bt$, does not obviously look like it 
corresponds to the drift velocity which is proportional to the derivative 
of $B$. But recall that \eqref{eomF} holds in any coordinate system. 
Since we have chosen $\log B$ as one of the coordinates, the 
only derivative of $B$ that appears is $\partial B/\partial(\log B) = B$. 

Upon canonical quantization,
Poisson bracket relation
\be
\{\chi, \log B\} = 1
\ee
leads to the commutation relation
\be
[\wh\chi, \wh{\log B}] = i\hbar
\ee
where $\wh\chi$ and $\wh{\log B}$ are self-adjoint operators.
To analyze the dynamics, we adopt the Heisenberg picture.
Since $\wh{\log B}$ is a 
function of the Hamiltonian $\wh{H} = \mu\wh{B}$, it is constant 
in time. The equation of motion for $\wh\chi$ is
\be\label{qchidot}
\partial_t\wh\chi =\frac{1}{i\hbar}\, [\wh\chi,\mu e^{\wh{\log B}}] = \mu \wh B\,, 
\ee
whose solution is 
\be\label{chisol}
\wh\chi = \wh\chi_0 +\mu \wh B t\,, 
\ee
where $\wh\chi_0$ is a constant operator. This is identical in form to the
classical equation of motion and its solution. In a state in which the dispersion
of $\wh B$ is small, the quantum theory thus reproduces something close to the 
classical drift motion. The eigenstates of the Hamiltonian $\wh{\log B}$ 
are localized on a single isomagnetic line, and are non-normalizable.

In the Schr\"odinger picture,
the quantum theory is described by square-integrable, 
complex-valued wavefunctions on $\chi$ (or $\log B$), 
with the canonical variables acting on $\psi(\chi)$ as
\ba
\wh\chi\, \psi &= \chi \psi \no
\wh{\log B}\, \psi &= -i\hbar \frac{\partial \psi}{\partial \chi}\,.
\ea
The Hamiltonian, being a simple function of $\log B$, admits a straightforward quantization,
\be
\wh H\psi = \mu e^{\wh{\log B}} \psi = \mu e^{-i\hbar \frac{\partial}{\partial \chi}}\psi\,.
\ee
The energy eigenstates are $\sim e^{ik\chi}$, with constant $k$. They are 
non-normalizable, and have energy eigenvalue $\mu e^{\hbar k}$.
The domain of the Hamiltonian consists of functions whose Fourier transform
decays faster than $e^{-\hbar k}$. 
The  action of $H$ on such a $\psi(\chi)$ is 
\be
\wh H \psi(\chi) = \mu \int dk\, \wt{\psi}(k) e^{ik(\chi -i\hbar)}\,,
\ee
where $\wt{\psi}$ is the Fourier transform of $\psi$. If $\psi$ is analytic, $\wh H \psi(\chi) = \mu \psi\left(\chi - i\hbar\right)$.

To gain some insight into this unusual action, we
consider its effect on  a coherent state with spread $\sigma \in \bb R^+$ and center $\zeta \in \bb C$, $\phi_{\sigma,\zeta} \sim e^{-(\chi - \zeta)^2/2\sigma^2}$. 
The Hamiltonian maps this to another coherent state, 
$ \mu \phi_{\sigma, \zeta + i\hbar}$,
whose center is shifted  in the direction transverse to the isomagnetic lines, corresponding to a change in
the expectation value of $\log B$ by an amount $ \hbar^2/\sigma^2$. 
The time evolution of the initial coherent state is thus given by 
\be
\phi(\chi, t) = e^{-i\wh H t}\phi(\chi,0)= \sum_{n=0}^{\infty}
\frac{1}{n!} \left( \frac{-i\mu t}{\hbar}\right)^n  \phi_{\sigma, \zeta + in\hbar}\,.
\ee%
This superposition of states, shifted by multiples of $\hbar$ in the 
direction transverse to the isomagnetic lines, appears in tension 
with the solution \eqref{chisol} 
to the Heisenberg equation of motion for $\chi$, 
which indicates a constant velocity shift {\it along} those lines.
Evidently, the superposition 
must in fact combine to produce propagation along those lines.

\subsection{Closed isomagnetic contours}
\label{closedBlines}

Now we consider the case where all contours of constant $B$ are closed, providing a foliation of the phase space by  (topological) circles. As before, we assume that $F$ and $B>0$ are smooth, and for technical reasons we will also assume that the total magnetic flux through the plane is infinite and that a certain regularity condition (explained below) holds at the center of the foliation.

\subsubsection{Poisson algebra}
\label{subsecPoissonclosed}
In this case, while $\log B$ might provide the role of a ``radial'' coordinate, the conjugate variable
$\chi$ is not continuous.
Accordingly, $\chi$ is not a suitable observable. Together with $\log B$, we could try to quantize the continuous variables $\sin(2\pi \chi/\chi_0(B))$ and $\cos(2\pi \chi/\chi_0(B))$, where $\chi_0(B)$ is the range (or period) of $\chi$ along the isomagnetic contour at $B$.\footnote{The range can be
determined using \eqref{symplectic2}: 
$d\log B\wedge d \chi=\frac{1}{2\pi}d\phi \wedge d\theta =\frac{1}{2\pi}\frac{d\phi}{d\log B}d\log B\wedge d\theta =  d\log B\wedge d \big(\frac{d\phi}{d\log B}\frac{\theta}{2\pi}\big)$, 
so $\chi_0(B) = \frac{d\phi}{d\log B}$.}
However, the resulting candidate algebra would not close on account of the 
$B$-dependence of $\chi_0(B)$.
In particular, 
the bracket 
\be
\big\{\sin\left( {2\pi\chi}/{\chi_0(B)}\right), \log B\big\} = ({2\pi}/{\chi_0(B)}) \cos\left( {2\pi\chi}/{\chi_0(B)}\right)
\ee
is not, in general, equal to a linear combination of these three observables 
with constant coefficients.

We thus seek another coordinate $Q$, satisfying $dQ\wedge dB = 0$, 
whose canonical conjugate has the same period on every contour. 
It turns out that $Q(x)$ must be a linear function of the magnetic flux 
through the compact region ${\cal R}(x)$ enclosed by the $B$ contour that
passes through $x$, 
\be\label{phi}
\phi(x) :=  \int_{\ca R(x)}  F\,.
\ee
Note that $\phi$ vanishes at the 
origin and, since $B>0$, increases monotonically away from it.
Therefore $\phi$ serves as a label for the contours.
As shown in the next paragraph, there exists a conjugate coordinate
$\theta$ (unique up to a $\phi$-dependent shift) such that
\be\label{symplectic2}
\omega = F = \frac{1}{2\pi}d\phi \wedge d\theta\,,
\ee
everywhere except at $\phi = 0$,
with the range of $\theta$ equal to $2\pi$ on all 
of the closed $B$ contours.
Although $\theta$ is discontinuous (or multi-valued), 
the 1-form $d\theta$ is well defined except at $\phi=0$. 
Since any other suitable coordinate $Q$ satisfying $dQ\wedge dB = 0$ will be some monotonic function $g$ of $\phi$, we have
\be
\omega = \frac{1}{2\pi}dg(\phi) \wedge d\left(\theta/g'(\phi)\right).
\ee
The period of the variable conjugate to $g(\phi)$ on the contour passing through $x$ is $2\pi/g'(\phi(x))$, which is independent of $\phi$ if and only if $g(\phi) = a\phi + b$ for constants $a\ne0$ and $b$, as we claimed earlier. We choose to take simply $Q = \phi$.

To see that a coordinate $\theta$ conjugate to $\phi$ and 
of fixed period exists, we can make use of the assumption 
that 
$\phi$ is smooth everywhere (including at the origin) and that it has a nondegenerate minimum at the origin (i.e., its Hessian is positive definite).
(These assumptions are unnecessary for establishing this particular result, but they make the argument a bit simpler
and will be needed for showing that the algebra we later construct exponentiates to a group action on the phase space---for a more general argument, see footnote~\ref{footnoteFalpha}.)
It then follows from the 
Morse-Darboux lemma \cite{colindeverdiere1979lemme} 
that there exist (smooth) local coordinates $u$ and $v$ 
in a neighborhood of the origin such that
\be\label{Fuv}
F = du\wedge dv
\ee
and 
\be
\phi = \a (u^2 + v^2)
\ee
for some smooth function $\a$ satisfying $\a(0)=0$ and $\a'(0)>0$.
Since $\phi$ is by definition the flux of $F$ through the region $\ca R(u,v)$ bounded by isomagnetic contours, and $\phi$ is constant along curves of constant $u^2+v^2$, it follows that $\ca R(u,v)$ is the disc of coordinate radius $u^2+v^2$.
Thus 
\be\label{phiuv}
\phi(u,v) = \int_{\ca R(u,v)} du'dv' = \pi(u^2 + v^2)\,,
\ee
so $\a$ is just multiplication by $\pi$.
Introducing the angular variable 
$\theta := \text{arg}(u+iv)$
it follows that, in this neighborhood of the origin,
\be\label{Fptt}
F = \frac{1}{2\pi}d\phi \wedge d\theta\,.
\ee
To establish that this form holds in the entire plane
we need to show that the angle coordinate $\theta$ can be 
extended to outside the neighborhood of the origin while 
preserving \eqref{Fptt} and the $2\pi$ periodicity of $\theta$.

To prove that $\theta$ can be extended in this manner,
we start with an arbitrary extension, so that the extended $\theta$
is constant along ``radial'' curves that join smoothly to the 
constant $\theta$ curves in the interior. Then $(\phi,\theta)$ provide
good coordinates everywhere except at the origin, so there exists a smooth function
$f(\phi,\theta)$ such that
outside the neighborhood we have 
\be\label{Ff}
F = f(\phi,\theta)\, d\phi\wedge d\theta
= \frac{1}{2\pi}d\phi\wedge d\tilde\theta = \frac{1}{2\pi}d(\phi\, d\tilde\theta)\,,
\ee
where 
$\tilde\theta(\phi,\theta) = \int_0^\theta f(\phi,\theta') d\theta'$.
Integrating the first and last expressions for $F$ 
in \eqref{Ff} over the annulus bounded on the inside
of the neighborhood by the $\phi_i$ contour and on the outside by any $\phi$ contour, yields the equation $\phi - \phi_i = \phi \oint d\tilde\theta/2\pi - \phi_i$, which implies that the period of the coordinate $\tilde\theta$ is 
also $2\pi$. Dropping the tilde, this establishes the validity of 
\eqref{symplectic2}.\footnote{To see strictly that a coordinate $\theta$ exists and has period $2\pi$ on each isomagnetic circle, it suffices to assume that $F$ and $B$ are smooth. It is clear from \eqref{Ff} that a coordinate $\wt\theta$ conjugated to $\phi/2\pi$ exists everywhere away from the origin, so it only needs to be shown that it has period $2\pi$ on each contour. Consider a deformation of $F$ by a smooth function $\alpha(\phi)$, $F_\alpha := \alpha F$, with $\alpha = 0$ in the interval $[0,\varepsilon/2]$ and $\alpha=1$ for $\phi \ge \varepsilon$, interpolating smoothly in the interval $[\varepsilon/2,\varepsilon]$. We have $2\pi F_\alpha = \alpha(\phi) d\phi \wedge d\wt\theta = d\left(\int_0^\phi\!d\phi'\alpha(\phi') \,d\wt\theta\right)$, which is regular everywhere. Stokes' theorem gives
\[
2\pi \int_{\ca R(x)}\! F_\alpha = \int_0^{\phi(x)}\!d\phi'\alpha(\phi') \int_{\partial{\cal R}(x)} d\wt\theta\,.
\]
Now taking the limit $\varepsilon\to 0^+$ and using that, 
for a smooth $F$,
$\phi$ is continuous (and $0$ at the origin), 
the left-hand side approaches $2\pi\phi(x)$, while the right-hand side approaches $\phi(x) \int_{\partial{\cal R}(x)} d\wt\theta$.
It follows that $\oint d\wt\theta=2\pi$ on all closed isomagnetic contours. \label{footnoteFalpha}}

The functions $\sin\theta$ and $\cos\theta$ are everywhere defined 
except at $\phi=0$, 
and together with $\phi$ they do form a closed algebra,  
\ba
&\{\sin\theta,\, \phi\} = 2\pi \cos\theta \no
&\{\cos\theta,\, \phi\} = - 2\pi \sin\theta \no
&\{\sin\theta,\, \cos\theta\} = 0\,. \label{closedBcharges0}
\ea
This is $\fr e(2)$, the Euclidean algebra in 2-dimensions.
There is, however, an issue: this algebra does not exponentiate to a group action on the phase space. This is apparent from the fact that $\sin\theta$ and $\cos\theta$ are ill-defined at the origin, $\phi = 0$, and so are their symplectic flows at that point.\footnote{If one were to ignore its failure to act globally on the phase space, and insist in quantizing this algebra,
one would necessarily obtain a representation (assumed unitary and irreducible) 
where the spectrum of $\phi$ is unbounded from below, which is inconsistent with the classical property $\phi\ge0$.}

A simple fix is to weight the sine and cosine by $\sqrt{\phi}$. Based on the coordinates $(u,v)$ introduced above, in a neighborhood of the origin, we see that a weight proportional to $\sqrt{\phi}$ would regularize the sine and cosine, since $\sqrt{\phi}\cos\theta \propto u$ and $\sqrt{\phi}\sin\theta \propto v$, which are well-defined and smooth near the origin.
Accordingly, we define the variables
\ba
&\Phi := \frac{1}{2\pi}\phi \no
&X := \sqrt{2\Phi} \cos\theta \no
&Y := \sqrt{2\Phi} \sin\theta\,, \label{closedBcharges}
\ea
which are well-defined and smooth on the full phase space. Together with constant function $1$, they form a closed algebra
\ba
&\{X, Y\} = -1 \no
&\{X, \Phi\} = - Y \no
&\{Y, \Phi\} = X \,.\label{slcharges}
\ea
This is $\wt{\fr e(2)}$, the centrally-extended Euclidean algebra in 2-dimensions.\footnote{Note that the redefinition of variables in \eqref{closedBcharges} then amounts to centrally extending the algebra defined by the variables in \eqref{closedBcharges0}.}
If $\Phi$ is unbounded from above, as assumed in the beginning of this section,\footnote{This assumption is actually necessary for 
validity of the guiding center approximation in the underlying microscopic theory.
Consider for example a centrally symmetric field with $B(r) \sim r^{-2 -\alpha}$ for large $r$, so the  large $r$ contribution to the 
flux integral is $\sim \int^r r'^{-1-\alpha}dr'$. This is bounded above if $\alpha>0$. 
For fixed $\mu$ (which is proportional to the magnetic flux through the gyro orbit), 
the gyro radius $\rho$ scales according to  $\rho \sim \sqrt{\mu/B}\sim r^{1 +\alpha/2}$.
But if $\alpha >0$ then for sufficiently large $r$ this becomes inconsistent with 
the guiding center approximation requirement $B'\rho \lesssim B$, since 
with this sort of $B(r)$ we have $B/B'\sim r$.} this algebra does exponentiate to a group of symplectomorphisms on the phase space, namely the centrally-extended Euclidean group in 2-dimensions, $\wt{E(2)}$.
In fact, $X$ and $Y$ are the generators of translations, $\Phi$ is the generator of rotations (with respect to the canonical angle $\theta$) 
about the origin, and the central element acts trivially.

It is worth asking whether the proposed fix of multiplying the sine and cosine by $\sqrt{\phi}$ was the most general that would produce an algebra of observables that exponentiates to a group action on the phase space. The answer is negative, and in App.~\ref{AppClosedBAlt} we show that the most general $\phi$-dependent weight for the sine and cosine is actually $\sqrt{2\Phi + \eta\Phi^2}$, with $\eta\ge 0$. The case $\eta =0$ is the one we considered above, which will lead to a quantization based on $\wt{E(2)}$. The case $\eta>0$ would lead to a quantization based on $\text{PSL}(2,\bb R)$, which we also describe in App.~\ref{AppClosedBAlt}. We regard the $\wt{E(2)}$ quantization as the most natural choice, since the alternative $\text{PSL}(2,\bb R)$ quantization is based on a classical algebra that depends on the parameter $\eta$ which has dimensions of inverse flux, and there is no scale in the classical theory that can be used to define such a parameter.

\subsubsection{Quantization}
\label{QuantizationclosedBlines}

To quantize the algebra \eqref{slcharges}
is to choose an irreducible unitary projective representation 
of the group it generates, $\wt{E(2)}$.
To this end, one
promotes the generators to operators, 
replaces the Poisson brackets by 
$1/i\hbar$ times commutators, 
and chooses an irreducible representation of the 
algebra in which the operators are self-adjoint and the states
have positive Hilbert space norm. When choosing the
representation, we have advanced in \cite{AndradeeSilva:2020ofl} 
that one should presumably require matching of any 
Casimir invariants of the Poisson algebra, since those commute with all observables and hence are not subject to quantum uncertainty relations. 
For the algebra \eqref{slcharges},
the only Casimir invariant is
\be\label{ClassicalCasimir}
C = 2\Phi - X^2 - Y^2
\ee
which classically evaluates to $0$.
Another well-motivated quantization
principle (discussed
for example in Sec.~2.1 of \cite{e2023quantization}) is that if the classical algebra contains a classically-complete subalgebra (i.e., a subalgebra that separates points in phase space), then that should also be represented irreducibly (since irreducibility is the quantum equivalent of completeness).

The quantized algebra thus reads
\ba
&[\wh X, \wh Y] = -i\hbar \no
&[\wh X, \wh \Phi] = - i\hbar \wh Y \no
&[\wh Y, \wh \Phi] = i\hbar \wh X \label{quantume2}
\ea
and the quantum Casimir is
\be
\wh C = 2\wh\Phi - \wh X^2 - \wh Y^2\,.
\ee
In any (complex) irreducible representation, Schur's lemma implies that a Casimir operator must be a multiple of the identity. Thus 
\be
\wh C = \gamma \,\hbar
\ee
for some $\gamma \in \bb R$ (since $\wh C$ is symmetric). It follows that $\wh \Phi$ can be expressed as
\be\label{whPhirep}
\wh\Phi = \frac{\wh X^2 + \wh Y^2}{2} + \frac{\gamma\,\hbar}{2}\,.
\ee
Since the Heisenberg subalgebra of  \eqref{quantume2} 
is classically complete, it should be 
represented irreducibly, according to the quantization 
principle stated below \eqref{ClassicalCasimir}. In fact, this is 
automatically satisfied in any irreducible representation of the 
full algebra since, according to  \eqref{whPhirep}, 
$\wh\Phi$ is then a function of $\wh X$, $\wh Y$ and $\wh 1$,
hence any invariant subspace for the Heisenberg subalgebra would 
also be invariant for the full algebra.

By the Stone-von Neumann theorem, the only (regular---see footnote~\ref{stonevonneumann})
unitary irreducible representation of the Heisenberg algebra
acts on wavefunctions $\psi \in L^2(\bb R)$ as
\ba
\wh X \psi(X) &= X \psi(X) \no
\wh Y \psi(X) &= i\hbar \frac{\partial\psi}{\partial X}(X)\,.
\ea
The flux operator $\wh \Phi$ acts as the harmonic oscillator Hamiltonian shifted by a $\gamma$-dependent constant. The representations are thus labeled only by the value of $\gamma$.
In any $\gamma$-representation, the spectrum of $\wh\Phi$ is discrete,
with uniform spacing equal to $\hbar$, and minimal value $\hbar(\gamma +1)/2$.
Applying the Casimir matching principle \cite{AndradeeSilva:2020ofl} gives $\gamma = 0$, since $C=0$ classically.
We will take $\gamma = 0$ henceforth. The spectrum of $\wh\Phi$ is thus
\be\label{specPhieff}
\text{Spec}(\wh\Phi) = \left\{ \hbar\left(j+\frac{1}{2}\right),\,\, j\in \bb N \cup\{0\} \right\}\,.
\ee

It is convenient to introduce a basis of normalized eigenvectors of $\wh\Phi$, with elements denoted by $|j\ra$, for $j\in \bb N \cup\{0\}$,
\be
\wh \Phi |j\ra = \hbar\left(j+\frac{1}{2}\right) |j\ra\,.
\ee
Also, importing from the familiar objects employed in the analysis of the quantum harmonic oscillator, we define 
\be\label{whadef}
\wh a := \frac{\wh X - i \wh Y}{\sqrt{2\hbar}} \quad\text{and}\quad
\wh a^\dag := \frac{\wh X + i \wh Y}{\sqrt{2\hbar}}
\ee
so $[\wh a, \wh a^\dag] = 1$. They act as ladder operators for $|j\ra$,
\ba
\wh a|j\ra &:= \sqrt{j}|j-1\ra \no
\wh a^\dag|j\ra &:= \sqrt{j+1}|j+1\ra
\ea
where a particular choice for the phases of $|j\ra$ is assumed.
Define also the number operator 
\be
\wh N := \wh a^\dag \wh a
\ee
which acts as $\wh N|j\ra = j |j\ra$.
Two other operators worth defining would be $\wh{e^{i\theta}}$ and $\wh{e^{-i\theta}}$, that is, the quantum correspondents of the classical functions $e^{i\theta}$ and $e^{-i\theta}$. Classically, these phases form a basis for complexified observables, since they must be periodic in $\theta$. Note that, classically, if we define $a$ and $a^*$ by the analogue of \eqref{whadef}, but omitting the $\hbar$, we would have
\be
a^* = e^{i\theta} \sqrt{\Phi} = e^{i\theta}|a^*|\,.
\ee
Of course one could try to directly quantize this relation, but an operator ordering prescription would be required. It seems more straightforward to see this as a polar decomposition of the function $a^*$, and consider the corresponding polar decomposition of operators in quantum mechanics. That is, for any closed, densely-defined (possibly unbounded) operator $A$, there is a unique polar decomposition given by
\be\label{polardecompgen}
A = U|A|
\ee
where $|A| = \sqrt{A^\dag A}$ and $U$ is a partial isometry 
with kernel equal to  $\ker(|A|)$ \cite{ReedSimon1980}. 
(A partial isometry is a linear map that is an isometry on the orthogonal complement of its kernel.) Applying this to $A = \wh a^\dag$
we obtain
\be
\wh a^\dag = \wh{e^{i\theta}} |\wh a^\dag|\,,
\ee
where $\wh{e^{i\theta}}$ is the (partial) isometry
defined by this decomposition.\footnote{We could have similarly defined the polar decomposition with the (partial) isometry on the right, $A = |A|U$, where $|A|$ would now be $\sqrt{AA^\dag}$ but $U$ would \emph{remain the same operator}. Therefore, the chosen order in the polar decomposition does not affect the definition of $\wh{e^{i\theta}}$.}
Since $|\wh a^\dag| = \sqrt{\wh a \wh a^\dag} = \sqrt{\wh N + 1}$ has a trivial kernel,
$\wh{e^{i\theta}}$ is an actual isometry on the Hilbert space. In fact, we have
\be
\wh{e^{i\theta}} = \wh a^\dag (\wh N + 1)^{-1/2}
\ee
so 
\be
\wh{e^{i\theta}}|j\ra = |j+1\ra\,.
\ee
Notice that, while this is obviously an isometry,
it is not unitary since its range does not include $|0\ra$. 
We define $\wh{e^{-i\theta}}$ as its adjoint,
\be
\wh{e^{-i\theta}} := \wh{e^{i\theta}}^\dag
\ee
which acts as
\be
\wh{e^{-i\theta}}|j\ra = \left\{ 
  \begin{array}{ c l }
    |j-1\ra & \quad \text{if $j \ge 1$} \\
    0                 & \quad \text{if $j=0$}\;.
  \end{array}
\right.
\ee
We see that $\wh{e^{-i\theta}}\wh{e^{i\theta}} = 1$, mimicking the classical relation $e^{-i\theta}e^{i\theta} = 1$, but $\wh{e^{i\theta}}\wh{e^{-i\theta}} = 1 - |0\ra\la 0|$.
Another good property of this definition of $\wh{e^{i\theta}}$ is that it transforms 
naturally under $\theta$-rotations, that is 
\be
[\wh\Phi, \wh{e^{i\theta}}] = \hbar\, \wh{e^{i\theta}}\,,
\ee
which is the direct quantization of the classical relation $\{\Phi,e^{i\theta}\} = -i e^{i\theta}$. The same good property holds for $\wh{e^{-i\theta}}$.

Let us study some dynamical properties of this quantum theory.
First, note that the classical Hamiltonian, given in \eqref{H}, is a function only of $\Phi$ in this setting, $H=\mu B(\Phi)$, so
a classical particle orbits on an isomagnetic loop, with fixed enclosed flux. 
Hamilton's equation  $\partial_t\theta = \partial H/\partial\Phi$
implies that $\theta$ evolves as
\be\label{thetat}
\partial_t \theta = \mu B'(\Phi)
\ee
so the particle orbits with the ``canonical angle''
$\theta$ changing at the constant rate $\mu B'(\Phi)$.
Quantum mechanically, the Hamiltonian acts as 
\be
\wh H |j\ra = \mu B(\wh\Phi)|j\ra = \mu B\big(\hbar(j+\sfrac{1}{2})\big) |j\ra\,.
\ee
Since $\wh H$ commutes with $\wh\Phi$, it immediately follows that the enclosed magnetic flux $2\pi\Phi$ is conserved in time, like in the classical dynamics.
The time evolution of $\wh{e^{i\theta}}$
is governed by the Heisenberg equation
\be
\partial_t \wh{e^{i\theta}} = \frac{1}{i\hbar}[\wh{e^{i\theta}},\wh H]\,.
\ee
By a straightforward computation we obtain
\be\label{EOMeithete}
\partial_t \wh{e^{i\theta}} =i{\mu}B^\Delta(\wh\Phi)\wh{e^{i\theta}}
\ee
where 
\be\label{qB'}
B^\vartriangle(\wh\Phi) := \frac{B(\wh\Phi)-B(\wh\Phi-\hbar)}{\hbar}
\ee
is the ``quantum derivative"
of $B$ with respect to its argument.
The solution to this evolution equation is
\be\label{evo}
 (\wh{e^{i\theta}})_t =e^{i\mu B^\vartriangle(\wh\Phi)t} (\wh{e^{i\theta}})_0\,.
\ee
Note that it has the same form as the classical evolution implied by \eqref{thetat},
\be
e^{i\theta(t)} = e^{i\mu B'(\Phi)t} e^{i\theta(0)}\,,
\ee
except in the replacement
of the derivative $B'(\Phi)$ by the quantum derivative $B^\vartriangle(\wh\Phi)$.
We can also define the \emph{angular velocity} operator by  
\be
\wh{\partial_t\theta} := -i \wh{e^{-i\theta}} \partial_t\wh{e^{i\theta}}.
\ee
Using the equation of motion \eqref{EOMeithete} and the relation $\wh{e^{-i\theta}}f(\wh\Phi)\wh{e^{i\theta}} = f(\wh\Phi+\hbar)$, we find
\be
\wh{\partial_t\theta} = \mu B^\vartriangle(\wh\Phi + \hbar)
\ee
which, in particular, shows that this operator is self-adjoint, and it is diagonalized by the energy eigenstates $|j\ra$.
Also, note that it matches with the classical angular velocity, given in \eqref{thetat}, only in the limit $\hbar\to 0$.

Finally, let us discuss the interesting quantization of flux implied by the representation theory. The spatial gap between neighboring states encloses a magnetic flux
\be
\Delta\phi = 2\pi\hbar
\ee
and the smallest flux is equal to half of this gap. 
As $\phi$ labels the circular isomagnetic orbits, it serves as a proxy for ``radius''. We thus interpret this an actual quantization of space, along the radial direction, in the effective theory. 
We would like to compare these gaps with radial distances in the microscopic
theory.
If $B$ is a centrally symmetric field that
does not vary much in the region between two 
metric radii $r$ and $r+\Delta r$ corresponding to adjacent allowed 
orbits in the effective theory, then the flux gap is 
$\Delta \phi \sim \pi B((r+\Delta r)^2 - r^2)
\sim 2\pi Br\Delta r$, so 
\be
\Delta r \sim \frac{\hbar}{Br} = \frac{\ell^2}{r}
\ee
where $\ell=\sqrt{\hbar/B}$ is the characteristic magnetic length.
The radius gap $\Delta r$ is largest when
$r$ is smallest, and since the minimal flux eigenvalue 
is $\phi = \pi \hbar$, the minimal
magnetic contour radius is $\sqrt{\hbar/B}= \ell$, so the 
maximal radius gap is $\Delta r\sim \ell$. 
The gap is therefore never larger than 
the minimal resolution scale of the microscopic quantum states,
and as $r$ grows the ratio $\Delta r/\ell$ 
decreases as $\ell/r$.

\section{Match to the microscopic theory}
\label{SecMatch}

The classical effective theory is derived by considering the limit where the gyro motion is very small compared to macroscopic scales.
In that limit, the magnetic moment (which is proportional to the angular momentum about the center of the gyro orbit) is an adiabatic invariant, fixed in the effective theory. This consists of a restriction on the class of classical states described by the effective theory. We thus expect that, in the quantum microscopic theory, a similar restriction to a subset of the full Hilbert space should  lead to the effective quantum theory. 
In this section we consider proposals for identifying such sectors of the quantum microscopic theory that can be matched to the effective theory (at each $\mu$).

To develop intuition we start by studying the case of a constant magnetic field, 
noticing that the effective theory corresponds closely to the restriction of the microscopic theory to a single Landau level. This establishes the matching between the state spaces of the two theories in a gauge invariant manner, but to understand the state-to-state matching we consider two choices of gauge: 
the Landau gauge and the symmetric gauge. 
The Landau gauge is well-adapted to the case of open isomagnetic contours.
In particular, the Landau levels are represented by a free particle wavefunction in one direction times a harmonic oscillator in the other direction, so we argue that the free particle factor corresponds to the wavefunction in the effective theory, while the harmonic oscillator factor is coarse-grained over. The symmetric gauge is well-adapted to the quantization of the closed isomagnetic contours. In particular, we can understand the relationship between the flux operator and the canonical angular momentum, which are identified in the $\wt{E(2)}$ quantization.
Finally, we propose a matching prescription for the general, non-uniform magnetic field, and analyze the predicted drift velocity in a perturbative setting of a quasi-uniform magnetic field, showing that it agrees with the prediction of our quantum effective theory.

\subsection{Constant $B$ case}
\label{MatchConstantB}

Let us begin with the case of a constant magnetic field in a two-dimensional Euclidean space. 
Adopting Cartesian coordinates $x^i$, 
the Hamiltonian of the microscopic quantum theory is given by
\be
\wh H = \frac{1}{2m} \left(\wh\pi_1^2 + \wh\pi_2^2\right)
\ee
where
\be
\wh\pi_i = \wh p_i -  A_i(\wh x)
\ee
is the kinematical momentum (again we are choosing units such that $q=1$).

In both the classical and quantum effective theory, for a uniform magnetic field, all states will have the same energy given by
\be\label{HmuBrel}
H = \mu B\,.
\ee
Underlying this relation is the assumption that (almost) all the energy is stored in the fast gyroscopic motion, with negligible contribution from the slow drift motion. 
It is thus reasonable to consider that the subspace of states in the quantum microscopic theory that corresponds to the effective theory is spanned by eigenstates of energy, with eigenvalue close to $\mu B$.
The eigenstates of energy are the well-known Landau levels, labeled by a non-negative integer $n$,
\be
\wh H\Psi_n = \hbar\Omega_B \left(n + \frac{1}{2} \right) \Psi_n
\ee
which leads to the identification
\be
\text{\it Microscopic theory at Landau level } n \,\longleftrightarrow \,\,\text{\it Effective theory with } \mu = \frac{\hbar}{m} \left(n + \frac{1}{2} \right)\,.
\ee
We note that this constitutes a kind of ``minimal identification'' of the theories, in the sense that only those states of a single Landau level are compared to the effective theory. It provides a linear subspace of the Hilbert space of the microscopic theory in which {\it all} states are (expected) to behave like states in the effective theory for a given $\mu$. There are of course states containing 
small admixtures of other Landau level states that are indistinguishable from these via measurements accessible in the effective theory, but we will not attempt here to derive an inequality characterizing such distinguishability.

Next we elaborate on further aspects of the matching prescription, for uniform $B$, by particularizing to the Landau and symmetric gauges. In each gauge, we propose a state-to-state mapping from the microscopic to the effective theory, refining the correspondence between state spaces enunciated above. 

\subsubsection{Landau gauge}
\label{Landaugauge}

The Landau gauge consists of taking
\be
A = B x\, dy\,.
\ee
This gauge maintains translation-invariance along the $y$-direction, and in this manner it is well-adapted to the quantization of open isomagnetic contours.\footnote{One could also take $A = -B y dx$, which maintains translation invariance along the $x$-direction.}
In this subsection, we shall study the state-to-state mapping from the quantum microscopic theory to our effective theory.

The wavefunctions of the lowest Landau level take the form~\cite{Landau:1991wop}
\be\label{LLLLandaugauge}
\Psi_{0,k}(x,y) = \frac{1}{\pi^{3/4}\sqrt{2l_B}}e^{iky} e^{-(x - l_B^2 k)^2/2l_B^2}
\ee
where $\hbar k$ is the (canonical) momentum in the $y$-direction, $\wh p_y \Psi_{0,k} = \hbar k \Psi_{0,k}$. These wavefunctions consist of plane waves along the $y$-direction times the ground state of a harmonic oscillator in the $x$-direction centered at $l_B^2 k$.
The lowering  operator (with the somewhat unconventional
phase of $-i$ times the standard one)
takes the form 
\be\label{a-Landau}
a = -i \frac{\ell_B}{\sqrt{2}} \left( \frac{\partial}{\partial x} + i \frac{\partial}{\partial y} + \frac{x}{\ell_B^2} \right)
\ee
and higher Landau levels are given by
\be\label{Landaugaugebasis}
\Psi_{n,k} = \frac{1}{\sqrt{n!}} (a^\dag)^n \Psi_{0,k}
\ee
which have energy eigenvalue $\hbar\Omega_B \left(n + \frac{1}{2} \right)$ and
are normalized as $\la\Psi_{n',k'}|\Psi_{n,k}\ra = \delta_{nn'}\delta(k-k')$.

As $[p_y, a^\dag] = 0$, the higher Landau levels $\Psi_{n,k}$ are also eigenstates of $\wh p_y$, with eigenvalues $\hbar k$. They are
plane waves along the $y$-direction times an excited state of the harmonic oscillator in the $x$-direction. 
We can define proper smearings of 
these states,
\be\label{Psismearing}
\Psi_{n,f} := \int\!dk\, f(k) \Psi_{n,k}
\ee
whose norms are $|\!|\Psi_{n,f}|\!|^2 = \int\!dk\,|f(k)|^2$.
The dispersion of $x$ in a normalized state $\Psi_{n,f}$ is, using formula \eqref{matrixLandaugauge},
\be
\Delta_{n,f}x = \sqrt{\la x^2 \ra_{n,f} - \la x\ra_{n,f}^2} = \ell_B \sqrt{n + \frac{1}{2} + \int\!dk\, |f(k)|^2\ell_B^2 k^2  - \left(\int\!dk\, |f(k)|^2\ell_B k \right)^2} \,.
\ee
In a state highly concentrated near a particular value of $k$, $f(k) = \sqrt{\delta(k-k_0)}$, the wavefunction will be completely delocalized in the $y$-direction, and spread in $x$ as
\be
\Delta_nx = \ell_B \sqrt{n + \frac{1}{2}}
\ee
independently of $k_0$.\footnote{This is the minimal possible spread at level $n$, i.e., $\Delta_{n,f}x > \Delta_nx$ for all $f$ with unit $L^2$-norm.}
We can interpret those states as a superposition of gyro orbits uniformly distributed along the $y$-direction (so as to ensure they are translation invariant along $y$). 
Note that the ``size'' of the gyro orbits are of the same order as the classical radius of gyration --- in fact, using the matching prescription to relate $\mu$ with $n$, we have
\be
\mu = \frac{\hbar}{m}\left( n + \frac{1}{2} \right) \sim \frac{1}{2}\Omega_B \left(\sqrt{2}\Delta_n x\right)^2
\ee
so, comparing with the classical relation \eqref{mudef}, we identify 
\be\label{rhonestimation}
\rho_n := \sqrt{2}\Delta_n x = \ell_B \sqrt{2n+1}
\ee
as the approximate radius of gyration. 

The quantization of the effective theory, for the case of constant $B$, was developed in Sec.~\ref{constantBeff}. That quantization is in some regards close to the quantization of the open isomagnetic contours case, based on the Heisenberg algebra, except that 
among the area-adapted coordinates
there is no ``preferred'' choice of coordinates (such as $\chi$ and $\log B$). Due to the absence of a metric structure at the effective level, the comparison with the microscopic theory is 
therefore not completely clear. 
For an area-adapted pair of coordinates $x$ and $y$, 
in the $y$-representation of the effective theory we have
\ba
\wh y \psi(y) &= y \psi(y) \\
B\wh x \psi(y) &= -i\hbar\frac{\partial\psi}{\partial y}(y)
\ea
so  
that $B\wh x$ acts as a (canonical) linear momentum, generating translations in $y$. In particular, $e^{iky}$ is an eigenstate of $B\wh x$ with eigenvalue $\hbar k$. 
If $x$ and $y$ are moreover Cartesian with respect
to the Euclidean metric of the microscopic theory, this 
suggests a state-to-state mapping between the theories,
\be\label{statemapping}
\Psi_{n,f}(x,y) \,\,\longleftrightarrow\,\, \psi_f(y) := \frac{1}{\sqrt{2\pi}} \int\!dk\, f(k) e^{iky}\,,
\ee
where $\psi_f$ belongs to the effective theory with $\mu = \frac{\hbar}{m}\left(n + \frac{1}{2}\right)$. This map is clearly isometric with respect to each theory's Hilbert space inner product. If on the other hand $x$ and $y$ are not Cartesian,
the particular state mapping \eqref{statemapping} would not apply exactly, but the
effective theory would differ from that defined by quantizing the metric-adapted coordinates only by higher order terms in $\hbar$.

\subsubsection{Symmetric gauge}
\label{secsymgauge}

Another popular gauge choice is the symmetric gauge, which realizes rotation symmetry 
(about one preferred point)
exactly.
The vector potential is
\be
A = \frac{B}{2}\left(-x_2 dx_1 + x_1 dx_2\right)\,.
\ee
The lowest Landau level, $n=0$, is spanned by the (normalized) wavefunctions~\cite{Landau:1991wop}
\be\label{LLLwavef}
\Psi_{0,j}(z,\bar z) = \frac{1}{\sqrt{\pi j! (2l_B^2)^{j+1}}} z^j e^{-|z|^2/4l_B^2}
\ee
where $j \in \{0\}\cup\bb N$ and $l_B = \sqrt{\hbar/B}$. These states are all annihilated by the lowering operator
\be\label{c-symmetric}
c = -i \frac{l_B}{\sqrt{2}} \left( 2 \frac{\partial}{\partial \bar z} + \frac{z}{2l_B^2} \right)
\ee
and higher Landau levels are obtained by acting with $c^\dag$,
\be
\Psi_{n,j} = \frac{1}{\sqrt{n!}} (c^\dag)^n \Psi_{0,j}\,.
\ee
These states are eigenvectors of the Hamiltonian $\wh H$ and the canonical angular momentum $\wh J$ $(= \wh x_1 \wh p_2 - \wh x_2 \wh p_1)$, with eigenvalues 
\ba
\wh H \Psi_{n,j} &= \hbar \Omega_B \left(n + \frac{1}{2}\right) \Psi_{n,j} \label{Heigenmicro} \\
\wh J \Psi_{n,j} &= \hbar(j-n) \Psi_{n,j} \label{Jeigenmicro}
\ea
respectively. 

It is worth making a brief comment on the ``size'' of the states. A straightforward computation using \eqref{matrixsymmetricgauge0} reveals that 
\be\label{r2micro}
\la \Psi_{n,j'}| \wh r^2 |\Psi_{n,j}\ra = 2l_B^2(n+j+1) \delta_{j'j}
\ee
which may lead one to think that, at fixed level $n$, states with higher $j$ are gyrating on a larger orbit. However, the situation is actually analogous to what happened in the Landau gauge: all states at level $n$ are gyrating with the same orbit size, but superposed uniformly along circles (with radii growing with $j$) so as to ensure the rotation symmetry of the state. Thus, it is only really the size of $\Psi_{n,0}$ that yields a good estimation for the size of the gyro orbits at level $n$. Namely, $\rho_n \sim \sqrt{\la \wh r^2 \ra_{n,0}} = l_B \sqrt{2(n+1)}$, in (approximate) agreement with the estimation in \eqref{rhonestimation}.

To discern the state-to-state matching for the case of 
closed isomagnetic contours, the key is to 
understand why $\Phi$ behaves as the generator of ``rotations'' in the effective theory, whereas it is the canonical angular momentum $J$ that plays this role in the macroscopic theory.\footnote{While the $\wt{E(2)}$ quantization is not the most natural one for the constant $B$ case, nothing prevents one from applying it, based on an arbitrary choice of origin and $S^1$ foliation.}
This can be understood both classically and quantum mechanically.
Classically, if we adopt $\Phi$ and $\theta$ 
as ``magnetic coordinates'' for space, and adopt the gauge with
    the only nonzero component of the vector potential being $A_\theta = \Phi$,
the action \eqref{microS} for the microscopic theory 
(with $q=1$) takes the form 
\be\label{microS2}
S = \int\!dt \left[ \frac{1}{2}m\left(g_{\Phi\Phi}\dot\Phi^2 +2g_{\Phi\theta}\dot\Phi\dot\theta + g_{\theta\theta}\dot\theta^2\right) + \Phi\dot\theta\right].
\ee
The momentum conjugate to $\theta$ is 
\be
p_\theta = \Phi + mg_{\Phi\theta}\dot\Phi + mg_{\theta\theta}\dot\theta
\ee
In the effective theory, the kinetic terms are dropped, so the action 
reduces to \eqref{macroS} and the momentum conjugate to $\theta$ 
is simply $p_\theta =\Phi$. 
This is how the flux and the generator of $\theta$
translations become identified. In the special case of a constant $B$
this becomes precisely the canonical angular momentum if concentric circles
are chosen for the $\Phi$ contours.

At the quantum level, we see it as follows.
As stated in \eqref{Jeigenmicro}, $\Psi_{n,j}$ are eigenvectors of the canonical angular momentum with eigenvalues $\hbar(j-n)$. 
Assuming that the contours are taken to be circular with respect to the metric in the microscopic theory,
the ($2\pi$-rescaled) flux  operator will be given by
\be\label{Phimicro}
\wh\Phi = \frac{1}{2\pi}B\pi \wh r^2
\ee
From \eqref{r2micro}, we see that $\wh\Phi$ is diagonal within the level $n$, with eigenvalues $\hbar (n+j+1)$. Thus, restricted to the level $n$, we have
\be
\big(\wh\Phi - \wh J\,\big) \big|_n = \hbar (2n +1) \sim B\rho_n^2 
\ee
The quantity on the right-hand side roughly corresponds to ($1/\pi$ times) the flux enclosed by the gyro orbit. This quantity is negligible in the effective theory, compared to the flux enclosed by macroscopic circles. In other words, the effective theory is capable, through coarse-grained measurements, of detecting only $j$, while $n$ is hidden inside the effective parameter $\mu$. 
We emphasize that, in the microscopic theory, $\wh J$ and $\wh \Phi$ are entirely different operators, with the former having discrete spectrum $\hbar \bb Z$ 
and the latter a continuous spectrum $\hbar (\bb R^+ \cup \{0\})$. However, $\wh r^2$ does not act within each Landau level, since $[\wh r^2, \wh H] \ne 0$, so it does not directly correspond to an operator in the effective theory. Instead, it is $\wh r^2|_n = 2\wh \Phi|_n/B$ that  passes through the matching. However, as we saw above, $\wh\Phi|_n$ is equal to $\wh J|_n$ up to an additive ($n$-dependent) constant, so it does have a discrete spectrum $\hbar(j+n+1)$, suggesting an explanation for the spurious quantization of radii there.

The analysis of the previous paragraph revealed that, in the microscopic theory, the spectrum of $\wh\Phi|_n$ is given by $\hbar\{n+1, n+2, \ldots\}$, but in the effective theory we found (Sec.~\ref{QuantizationclosedBlines}) that the spectrum of the effective operator $\wh\Phi$ is $\hbar\{\sfrac{1}{2},\sfrac{3}{2},\cdots\}$.
Thus, before identifying the precise state-to-state mapping from the microscopic to the effective theory, it is worth refining slightly this analysis to address this spectral shift.
In fact, this relative shift by $\hbar (n+\sfrac{1}{2})$ can be explained by noticing that, as defined in \eqref{Phimicro}, $\Phi$ measures the flux through the circle 
defined by the microscopic radial coordinate $r$ of the particle, while in the effective theory $\Phi$ measures the flux through the circle 
defined by the radial coordinate of the center of the gyro orbit. 
In the Heisenberg picture, the position operator, expressed as a complex variable $\wh r = \wh x + i \wh y$, evolves as
\be
\wh r(t) = \wh c + \frac{i}{m\Omega_0}\wh \pi e^{-i\Omega_0 t}
\ee
where $\wh c = \wh r - \frac{i}{m\Omega_0}\wh \pi$. We can therefore identify $\wh c = \wh c_x + i \wh c_y$ as the center of the gyro orbit, so the flux enclosed by it is given by
\be
\wh\Phi' := \frac{1}{2\pi} B\pi\! \left( \wh c_x^2 + \wh c_y^2\right) 
= \frac{\wh H}{\Omega_0} + \wh J\,.
\ee
This operator acts within each Landau level and, according to \eqref{Heigenmicro} and \eqref{Jeigenmicro}, its restriction to the level $n$ gives
\be\label{Phi'spec}
\wh\Phi'\Psi_{n,j} := \hbar \left( j + \frac{1}{2} \right)\Psi_{n,j}\,.
\ee
The spectrum of $\wh\Phi'$ is thus $\hbar\{\sfrac{1}{2}, \sfrac{3}{2}, \ldots\}$,\footnote{The intuition for this difference in the spectra of $\wh\Phi$ and $\wh\Phi'$ is that $\pi |r|^2$ measures the area of the ``blurred'' region that the particle encloses, which is larger than the region enclosed only by its center, with area $\pi |c|^2$. In fact, this can be understood from a semi-classical perspective by noticing that the time-averaging over the period of the gyro orbit would give $\text{T.A.}[\frac{B}{2}|r|^2] = \frac{B}{2}(|c|^2 + \rho^2)$, and quantum mechanically $\frac{B}{2}\rho_n^2 \sim \hbar (n + \sfrac{1}{2})$.}
which matches exactly with the spectrum of $\wh\Phi$ in the effective theory, given in \eqref{specPhieff}.
Accordingly, this suggests a state-to-state mapping where eigenvectors of the flux operator through the ``guiding center radius''
in the microscopic theory are matched to eigenvectors of the flux operator in the effective theory,
\be\label{closedlinesStS}
\Psi_{n,j} \,\,\longleftrightarrow\,\, |j\ra
\ee
where $|j\ra$ belongs to the Hilbert space of the quantum effective theory, described in Sec.~\ref{QuantizationclosedBlines}, with $\mu = \frac{\hbar}{m}\left(n + \frac{1}{2}\right)$.

\subsection{Non-uniform $B$ case}
\label{nonuniformBmatch}

The most interesting case, to compare with our effective theory, is that of a non-uniform magnetic field. 
The aim is to identify the states and observables of the effective theory
with a suitable space of states and coarse-grained observables of the microscopic theory, and to characterize the
fidelity of the effective description for these states. 

To identify the subspace of the microscopic theory corresponding to the state space 
of the effective theory, we employ a similar logic as that used in Sec.~\ref{MatchConstantB}, 
based on the energy relation $H=\mu B$.
Define the operator
\be
\wh W_\mu := \wh H - \mu \wh B
\ee
where $\wh B := B(\wh x, \wh y)$ and $\mu$ is a real parameter (with dimensions of magnetic moment).
A first thought would be to identify the kernel of $\wh W_\mu$ as the space of 
microscopic states associated with the effective theory at magnetic moment $\mu$, 
thereby selecting states that satisfy the effective theory energy relation, $\wh H = \mu \wh B$.
However, the kernel may be trivial or contain too few elements to match to the effective theory.
In view of this, we should allow for eigenvectors of $\wh W_\mu$ with eigenvalues in some 
\emph{neighborhood of zero}, according to the following matching prescription
\be\label{matchnonuniformB}
\text{\it Microscopic theory restricted to } \ca V_\mu := \text{Eigen}\left(W_\mu; \approx 0 \right)\,\longleftrightarrow \,\,
\text{\it Effective theory at } \mu
\ee
where it is \emph{assumed} that the spectrum of $\wh W_\mu$ has a cluster (or band) of eigenvalues near $0$, 
separated by a finite gap to other eigenvalues, and 
$\text{Eigen}\left(W_\mu; \approx 0 \right)$ denotes the span of all 
eigenvectors associated with this cluster 
near $0$.\footnote{In many instances, in this section, we use the terms 
``eigenvalues'' and ``eigenstates'' in the loose sense, i.e., 
even when referring to continuous portions of the spectrum 
(whose points are not technically eigenvalues of any 
eigenvector in the Hilbert space). Rigorously, 
we can identify $\ca V_\mu$ with the image of the 
{\it Riesz projector}, $P_\mu =\frac{1}{2\pi i}\oint_{\ca C}\!dz\,(z - W_\mu)^{-1}$, 
where $\ca C$ is a contour in $\bb C \backslash \text{Spec}(W_\mu)$ 
enclosing only the spectral cluster/band near $0$.\label{Rieszfn}}
When $B$ is uniform this proposal is clearly 
equivalent to the one employed in Sec.~\ref{MatchConstantB}, provided that $\mu = \mu_n :=\frac{\hbar}{m}(n + \frac{1}{2})$, for integer $n$. 
For a non-uniform $B$, we expect that the values of $\mu$ must be similarly restricted
in order for the spectral assumption to hold.
That this spectral assumption is plausible, 
under the conditions of the guiding center approximation, can be seen as follows.
For an energy eigenstate sufficiently localized around 
a point $x$, we expect its energy to be approximately given by 
$E_n(x) \approx \mu_n B(x)$.
Then, also in this local approximation, $\wh W_\mu$ would have an 
eigenvalue $w_n(x) \approx \mu_n B(x)- \mu B(x)$.
If $\mu = \mu_n$, the eigenvalues $w_n(x)$ associated with the ``Landau band'' at level $n$ will be nearly zero, and separated 
from the eigenvalues $w_{n\pm 1}(x)$ of the neighboring bands by roughly 
$\frac{\hbar B(x)}{m} \ge \hbar \Omega_\text{min}$, where $\Omega_\text{min} := B_\text{min}/m$ (with $B_\text{min}>0$ the minimum value of $B$).
Finally, if the matching proposal is to be well-defined, it must be that $\ca V_\mu$ is isomorphic to the Hilbert space of the effective theory (which is reasonable to expect based on the uniform $B$ case analysis).
\emph{For notational simplicity in the following, we omit the label $\mu$ in $W_\mu$ 
and $\ca V_\mu$, and also the ``hats'' on operators}.

A central feature of the guiding center approximation is that $\mu$ is an adiabatic invariant. It is thus natural to consider that, if the criterion above properly implements the guiding center approximation in the quantum microscopic theory, the subspace $\ca V(t):= \text{Eigen}\left(W(t); \approx 0 \right)$ should be approximately conserved in time, where 
\be
W(t) := e^{iHt/\hbar}W e^{-iHt/\hbar} = H -\mu B(t)
\ee
is the Heisenberg picture evolution of $W$, and $B(t) := e^{iHt/\hbar}B(x,y) e^{-iHt/\hbar}$. 
Note that since $W(t)$ is a unitary evolution of $W$, its spectrum is time independent; in particular, for all $t$, $\ca V(t)$ remains associated to the same cluster of eigenvalues near zero, so its (approximate) conservation  would imply in the Schr\"odinger picture that the state remains close 
to a zero eigenvalue eigenvector of $W$.
Such a conservation of $\ca V(t)$ is not guaranteed, but should rather be ensured by further conditions which we describe below.
These conditions can be viewed as the quantum counterparts of the classical guiding center conditions derived in App.~\ref{AppClassicalGCconditions}.
The fact that such conditions exist, permitting us to establish consistency of the matching prescription under time evolution given natural hypotheses, provides further evidence for the soundness of this proposal.

Let $\ca H$ be the whole Hilbert space of the microscopic theory, and $P(t)$ be the projector onto $\ca V(t) \subset \ca H$. It satisfies a Heisenberg-type evolution law
\be
P(t) = e^{iHt/\hbar}P(0) e^{-iHt/\hbar} \,.
\ee
The statement that $\ca V(t)$ is approximately conserved in time can be translated into\footnote{A bound for an operator $\scr O$ is a non-negative number $b$ such that $|\scr O \psi| \le b |\psi|$ for all $\psi$ in the Hilbert space; if no such $b$ exists the operator is unbounded. The norm of an operator $\scr O$, denoted by $|\!|\scr O|\!|$, is defined as the infimum of its bounds, or $\infty$ if it is unbounded.} 
\be
|\!| P(t) - P(0) |\!| \ll 1\,.
\ee 
It may be difficult to verify this condition exactly for a general magnetic field, but we can have some control over it by applying a strong version of the adiabatic theorem (Theorem~\ref{TheoremAdiabatic}), as we explain in App.~\ref{AppAdiabIneq}. In that appendix, we describe the precise necessary assumptions and derive a rigorous, exact formula for a bound on $|\!| P(t) - P(0) |\!|$. In the present section, we content ourselves with stating the result schematically. 

The main assumption is that the part of the spectrum of $W(t)$ associated with $P(t)$ is separated by a gap from the rest.
Since the spectrum of $W(t)$ is time independent, the gap is the same for all $t$,
so let
\be
\Delta := \text{Dist}\left[\text{Spec}\big( W|_{\ca V} \big), \text{Spec}\big( W|_{{\ca V}^\perp} \big)\right] > 0
\ee
where ${\ca V}^\perp$ is the orthogonal complement of $\ca V$. We also assume that $B$ is bounded and preserves the domain of $H$, i.e., $B\text{Dom}(H)\subset\text{Dom}(H)$. (A sufficient condition is that $B$ and all its spatial derivatives up to second order are bounded.)
It then follows (Corollary~\ref{CorollaryAdiabatic})
that there exist constants $\varepsilon_1$ and $\varepsilon_2$ such that
\be\label{AdiabIneq}
|\!| P(t) - P(0) |\!| \le \varepsilon_1 + \varepsilon_2 t\,.
\ee 
The constant $\varepsilon_1$ is dominant for small times, expressing the ``periodic error'', 
and it is itself bounded (Theorem~\ref{theoremvareps1bound}) by
\be\label{vareps1bound}
\varepsilon_1 \le 2\varepsilon := \frac{4\mu}{\Delta^2} |\!| Q[B,H] P|\!|
\ee
where $P = P(0)$ is the projector to $\ca V = P\ca H$, and $Q = 1 - P$ is the complementary projector.
The term associated with $\varepsilon_2$ grows linearly in time and represents the ``secular error'', that is, deviations from exact adiabatic conservation that accumulate over time.
We also derive an exact bound for $\varepsilon_2$ (Theorem~\ref{theoremvareps2bound}), and by a subsequent rough estimation described in App.~\ref{AppAdiabIneq} show that
\be
\varepsilon_2 \lesssim 100 \varepsilon^2 \Omega
\ee
where $\Omega \sim B/m$. 
As long as $\Omega t \ll 0.01/\varepsilon^2$, the secular error remains small. Thus, if $\varepsilon$ is sufficiently small, the adiabatic conservation of $P(t)$ will hold over macroscopically large times.

To roughly estimate $\varepsilon$, defined in \eqref{vareps1bound}, let us consider a quasi-uniform $B = B(x)$ which behaves homogeneously under differentiation (i.e., if $B$ varies over a length scale $L$, then $\partial^sB/\partial x^s \sim B/L^s$).
Roughly, say that $\mu \sim \hbar n/m$ and $\Delta =: \hbar \Omega\,\, (\sim \hbar B/m)$. 
Say also that the eigenstates of $W$ are approximately $\Psi_{n,f}$, the Landau-gauge energy eigenstates, so $P \sim P_n$, where $P_n$ is the projector to the Landau level $n$. 
The only non-vanishing matrix elements of $Q_n[ B,H] P_n$
are between states with $n' = n +s \ne n$, 
\ba
\la\Psi_{n +s,f'}|Q_n[B,H] P_n|\Psi_{n,f}\ra &\sim \la\Psi_{n +s,f'}|[B,H]|\Psi_{n,f}\ra \no
&\sim -s\hbar\Omega \la\Psi_{n +s,f'}|B |\Psi_{n,f}\ra \no
&\sim -i^{-s}s\hbar\Omega \frac{\sqrt{n!(n+s)!}}{\text{min}(n!,(n+s)!)s!} \la\Psi_{0,f'}| \left(\frac{\ell}{\sqrt{2}}\frac{\partial}{\partial x}\right)^{|s|} B(x) |\Psi_{0,f}\ra
\ea
where we used formula \eqref{matrixLandaugauge} and kept only the term with the lowest order derivative acting on $B(x)$, as higher derivatives will contain extra factors of $\ell/L$ which are presumably small. 
In fact, the dominant contribution will come from $s=\pm1$, which gives 
\be
\la\Psi_{n +1,f'}|Q_n[B,H] P_n|\Psi_{n,f}\ra \sim i\hbar\Omega \sqrt{n+1}\la\Psi_{0,f'}| \frac{\ell}{\sqrt{2}}\frac{\partial B}{\partial x} |\Psi_{0,f}\ra
\ee
so, roughly,
\be
|\!|Q_n [ B,H] P_n |\!| \sim \frac{1}{2}\hbar\Omega \rho_n \Big|\!\Big| \frac{\partial B}{\partial x} \Big|\!\Big|
\ee
where $\rho_n = \sqrt{2(n+1)}\,\ell$ is of order of the size of the gyro orbit. Then we get
\be
\varepsilon \sim n\rho_n \Big|\!\Big| \frac{\partial \log B}{\partial x} \Big|\!\Big|\,.
\ee
The condition $\rho_n\ll L/n$, that the size of the gyro orbit be much smaller than the length scale over which $B$ varies divided by $n$, is thus sufficient to guarantee the 
adiabatic conservation of $P(t)$, at least up to times where the secular error (the $\varepsilon_2 t$ term) remains small. 

The factor of $n$ in the denominator renders this condition stronger than that required for accuracy of the guiding approximation in the classical case.
The reason for this can be understood as follows. The matching proposal \eqref{matchnonuniformB} entails a sharp restriction to a particular subspace of states, which can only be satisfied if transitions to neighboring orthogonal states are strongly suppressed. In the adiabatic approximation, the transition rate $|\!|\dot P(t)|\!|$ scales with how fast $W(t)$ is changing in time, compared to the spectral gap, which as shown in App.~\ref{AppAdiabIneq} (following from Lemma~\ref{lemmaQdotPPsylvester} and formula \eqref{QdotWPQBHP}) is proportional to $\mu/\Delta \sim n/\Omega$.
This explains the $n$-dependence of the more stringent quantum condition.
Note that such a condition is only relevant when resolving states with a sharply defined value of $n$.
In the classical regime, where $n$ is large, states will tend to be like coherent states, which are superpositions of many levels, with a spread of order $\sqrt{n}$. In that case, the matching proposal could be modified to include such states spread across many eigensubspaces of $W$, centered around near-zero eigenvalues, so that the suppression of transitions to neighboring levels would not need to be enforced as strongly. 

We need also a prescription for matching operators in the microscopic theory to operators in the effective theory.
Note that the guiding center approximation relies on coarse graining not only over small length scales (of the order of the gyro radius), but also over fast time scales (of the order of the cyclotron period). This leads us to the following additional ingredient for the matching proposal. Let $\scr O(t)$ be a Heisenberg picture operator in the microscopic theory, and define its ``time coarse graining'' by
\be\label{timeavgprescription}
\ol{\scr O}(t) := \int\! d\tau K(\tau) \scr O(t-\tau) \,,
\ee
where $K$ is a real-valued convolution kernel whose purpose is to suppress the fast time scales of $\scr O(t)$.\footnote{Assuming that $K$ is real ensures that the time coarse graining preserves self-adjointness. One could further assume that it is even to make the $K$-weighted average centered at $t$, without introducing a directional bias.} In other words, $K$ is a low pass filter and, in QFT terms, this time coarse graining amounts to ``integrating out the high frequencies''. 
As the Fourier transform gives\footnote{For possibly unbounded operator-valued functions, it is simpler to interpret the relation below in the ``weak sense'', i.e., as a relation between matrix elements.}
\be
\ca F[\ol{\scr O}](\omega) = \ca F[K](\omega) \ca F[\scr O](\omega) \,,
\ee
we could take $K$ such that $\ca F[K](\omega)$ vanishes for high frequencies and equals $1$ for low frequencies, assuming that there is a clear separation of time scales between the fast microscopic gyro motion ($\Omega\sim B/m$) and the slower time scales resolved in the effective regime. Note that since $\ca F[K](0) = 1$, $K$ satisfies the normalization condition $\int\!d\tau K(\tau) = 1$, which implies that the time coarse graining of a time-independent operator is equal to the operator itself.
This proposal is obviously not fully precise in general, due to the arbitrary choice of $K$, but we will see in the perturbation analysis of the next section that relevant, effective results can be insensitive to the particular choice of $K$.
Finally, as noticed in Sec.~\ref{secsymgauge}, it is possible that $\ol{\scr O}(t)$ does not act within the subspace $\ca V$. Accordingly, one must project the time-coarse grained operator onto $\ca V = P\ca H$,
\be
\scr O(t) \mapsto \ol{\scr O}(t) \mapsto P\ol{\scr O}(t) P
\ee
and the resulting operator on $\ca V$ is then to be matched with a corresponding operator in the effective theory.

\subsection{First order in field gradient perturbative analysis}
\label{FirstOrderPertAnalysis}

In a generic magnetic field
the microscopic dynamics is not analytically tractable. Nevertheless, the 
effective theory is only expected to be effective when the magnetic gradient length scale
is long compared to the radius of the gyro orbits, so we shall content ourselves to examine
the precise correspondence in the neighborhood of a generic point, where the magnetic field
can be approximated as a linear function of one of the Cartesian coordinates.
As a check of consistency, we will compute the drift velocity for microscopic states in the subspace $\ca V_\mu$, and confirm that they match the corresponding result for the effective theory.\footnote{The quantum drift velocity has also been studied perturbatively in \cite{chan2016quantum,Chan:2017dex}. They derive a Heisenberg picture formula for the $\text{grad}(B)$ and $E\times B$ drift velocities, and compare with a numerical evolution of the Schr\"odinger equation.
Our aim, in contrast, is to identify a correspondence between a suitable set of quantum states of the microscopic 
theory with states of the effective theory, and 
hence our perturbative analysis is organized differently.}

Let the magnetic field depend linearly on the coordinate $x$,
\be\label{BlinearL}
B = B_0 \left(1 +\frac{x}{L} \right)
\ee
where $L \gg \ell$ ($:= l_{B_0}$).
We can use perturbation theory to study the kernel of $W = H - \mu B$. The perturbation theory is formally in $L^{-1}$, but as this is a dimensionful parameter we should understand that the perturbation works only for states sufficiently concentrated in a region $|x| \ll L$. 
From the classical intuition, we expect that a state will evolve in time by moving in the $y$-direction, thus allowing us to track it for a macroscopic time scale without leaving the perturbative regime.
Taking the magnetic potential to be
\be
A = B_0\left(x +\frac{x^2}{2L}\right) dy
\ee
the Hamiltonian becomes, to first order in $L^{-1}$,
\be
H = H^{(0)} + L^{-1} H^{(1)}
\ee
with
\ba
H^{(0)} &= \frac{1}{2m} \left(p_x^2 + (p_y - B_0 x)^2 \right) \\
H^{(1)} &= - \frac{B_0}{2m} x^2 (p_y - B_0 x)\,.
\ea
Similarly, $W = H - \mu B$ is decomposed as 
\be
W = W^{(0)} + L^{-1}W^{(1)}
\ee
with $W^{(0)} = H^{(0)} - \mu B_0$ and $W^{(1)} = H^{(1)} - \mu B_0 x$.

At zeroth order in $L^{-1}$, the eigenstates of $W^{(0)}$ coincide with the Landau levels. Moreover, since the perturbation commutes with $p_y$, it is convenient to use the Landau-gauge basis, $\{\Psi_{n,k}\}$, defined in \eqref{Landaugaugebasis}.
The exact eigenstates of $W$ can thus be taken to be $p_y$ eigenstates $\propto e^{iky}$,
and the perturbation will only mix states with the same value of 
$k$.\footnote{We could factor this $y$-dependence out of the states, but to maintain a unified notation we will subsume that factor into the notation for the states, at the cost of including $\delta(k'-k)$ in many formulas.}
To first order in $L^{-1}$, the degeneracy of the level $n$ may be lifted according to the eigenvalues of $P_nW^{(1)}P_n$, where $P_n$ is the projector to that level.
Using formula \eqref{matrixLandaugauge}
for computing matrix elements of functions of $x$, 
we find
\ba
\la \Psi_{n,k'}| W^{(1)} |\Psi_{n,k}\ra &= -B_0 \la \Psi_{n,k'}| \left( \frac{\hbar}{2m} x^2 \left(k - \frac{x}{\ell^2}\right) + \mu x \right)|\Psi_{n,k}\ra \no
&= B_0 k \ell^2 \left[ \frac{\hbar}{m} \left( n + \frac{1}{2} \right) - \mu \right] \delta(k-k')\,.
\ea
As expected from the commutativity of the perturbation with $p_y$, $\Psi_{n,k}$ are eigenvectors of $P_nW^{(1)}P_n$. 
Consequently, up to first order in $L^{-1}$, the corrected eigenvectors of $W$ are expressed as
\be
\Upsilon_{n,k} = \Psi_{n,k} + L^{-1} \Upsilon_{n,k}^{(1)}
\ee
with eigenvalues
\be\label{eigenvaluemunonuniB}
W \Upsilon_{n,k} = B_0 \left( 1 + \frac{k \ell^2}{L} \right) \left[ \frac{\hbar}{m} \left( n + \frac{1}{2} \right) - \mu \right] \Upsilon_{n,k}
\ee
and\footnote{Technically, in degenerate perturbation theory there could be also corrections to $\Psi^{(1)}$ coming from higher order perturbations, which would mix states within each degenerate subspace. In the present case, however, we do not need to worry about this aspect since the perturbation commutes with $p_y$, so no mixing within a level $n$ (among different $k$'s) can occur.} 
\be\label{Psi0k(1)}
\Upsilon_{n,k}^{(1)} = -\sum_{n'\ne n}\int\!dk'\, \Psi_{n',k'} \,\frac{m}{B_0\hbar (n' - n)} \la \Psi_{n',k'}|W^{(1)} |\Psi_{n,k}\ra \,.
\ee
We define also proper smearings of these states,
\be\label{Upsilonsmearing}
\Upsilon_{n,f} := \int\!dk\, f(k) \Upsilon_{n,k}
\ee
whose norms are $|\!|\Upsilon_{n,f}|\!|^2 = \int\!dk\,|f(k)|^2$.

Suppose there is an $n \in \bb N \cup \{0\}$ such that $\mu = \tfrac{\hbar}{m} \left( n + \tfrac{1}{2} \right)$. In this case, the spectrum of $W$ will contain an infinitely-degenerate zero value, separated by a gap of order $\hbar\Omega_0$ to the nearest clusters (associated with $n\pm 1$).\footnote{Note that formula \eqref{eigenvaluemunonuniB} should be considered valid only within the perturbative regime, which presumes $n\ell^2k/L \ll 1$ (see \eqref{pertcondalpha2} below). 
For large $k$ ($\ell^2k \sim L$) one would even go beyond the regime where the linear $B$ field would permit the guiding center approximation (as such $B$ is unbounded from below). Thus, only with this restriction on $k$ we can say that the clusters in the spectrum are separated by $\sim \hbar\Omega_0$.}
The matching prescription then states
\be\label{matchingBpert}
\text{\it Microscopic theory restricted to states $\Upsilon_{n,f}$} \,\longleftrightarrow \,\,\text{\it Effective theory at } \mu = \frac{\hbar}{m}\left(n+\frac{1}{2}\right)
\,.
\ee
Note that this example illustrates how an effective theory at arbitrary $\mu$ cannot be sharply matched to the microscopic theory: if $\frac{m\mu}{\hbar} + \frac{1}{2}$ were not a natural number, one would have to take the $n$ that best approximates it in order to ``best match'' the two theories.
Moreover, although the degeneracy of the zero eigenvalue of 
$W_\mu$ is preserved to first order in the particular magnetic field \eqref{BlinearL}, 
it is possibly lifted at second order, and is presumably lifted even at first order
in more generic fields. The spectrum of $W$ near zero is itself therefore generically 
a spread-out cluster (as alluded to in Sec.~\ref{nonuniformBmatch}) for all $\mu$.

Using again formula \eqref{matrixLandaugauge}, we find that the matrix elements of $W^{(1)}$, with $n' = n + s\ne n$, are
\be
\la\Psi_{n+s,k'}|W^{(1)}|\Psi_{n,k}\ra = -\frac{\hbar B_0\ell}{2m} \alpha_{n,s}(k\ell) \delta(k-k') 
\ee
with
\ba
\alpha_{n,1}(k\ell) &= \frac{i}{\sqrt{2}} \sqrt{n+1} \left(\ell^2k^2 - \frac{n}{2} + \frac{1}{2}\right) \no
\alpha_{n,2}(k\ell) &=  \sqrt{(n+2)(n+1)}\, \ell k \no
\alpha_{n,3}(k\ell) &= -\frac{i}{2\sqrt{2}} \sqrt{(n+3)(n+2)(n+1)}
\label{alphanscoef}
\ea
and $\alpha_{n,-s} = \alpha_{n-s,s}^*$, and all other $\alpha_{n,s}$ ($s \ne \pm 1, \pm 2, \pm 3$) vanish.
Thus, the corrected state is
\be\label{Upsilonalpha}
\Upsilon_{n,k} = \Psi_{n,k} + \frac{\ell}{2L} \sum_{s\ne 0} \frac{\alpha_{n,s}(k\ell)}{s} \Psi_{n+s,k}\,.
\ee
The perturbation theory is expected to be accurate, up to first order, only if the corrections to the states are small. Considering the magnitudes of the $\alpha_{n,2}$ and $\alpha_{n,3}$ coefficients, this requires that
\ba
\frac{\ell}{2L}n\ell k &\ll 1 \label{pertcondalpha2} \\
\frac{\ell}{2L} \left(\frac{n}{2}\right)^{3/2} &\ll 1
\ea
in the limit of large $n$. (In this limit, the condition coming from $\alpha_{n,1}$ is implied by the ones above.)

Let us now study the drift velocity in this microscopic theory. 
At first sight, $v_y = \pi_y/m$ is the natural operator in the microscopic theory corresponding to the drift velocity in the effective theory. However, its Heisenberg picture version oscillates on time scales of the ``fast'' cyclotron period.
Accordingly, to identify the actual operator corresponding to the drift velocity, we apply 
the time-coarse-graining 
(and projection) prescription described at the end of Sec.~\ref{nonuniformBmatch}. For a time-independent (Schr\"odinger picture) operator $\scr O$, we have in terms of matrix elements
\be\label{timeavgprescription1}
\la\Upsilon_{n,f'}|\ol{\scr O}(t) |\Upsilon_{n,f}\ra = \int\!d\tau K(\tau) \la\Upsilon_{n,f'}| e^{iH(t-\tau)/\hbar}\scr O e^{-iH(t-\tau)/\hbar}|\Upsilon_{n,f}\ra \,.
\ee
These matrix elements should then be matched to matrix elements of a corresponding (Heisenberg picture) operator $\wt{\scr O}(t)$ in the effective theory, $\la\psi_{f'}|\wt{\scr O}(t)|\psi_{f}\ra$. 

The simplest way to evaluate the integral above is to find the eigenstates of energy. The corrected eigenstates of $H$, up to first order in $L^{-1}$, are given by
\be
\Psi'_{n,k} = \Psi_{n,k} + \frac{\ell}{2L} \sum_{s\ne 0} \frac{\beta_{n,s}(k\ell)}{s} \Psi_{n+s,k}
\ee
where 
\ba
\beta_{n,-1}(k\ell) &= -\frac{i}{\sqrt{2}} \sqrt{n} \left(\ell^2k^2 + \frac{3}{2}n \right)  \no
\beta_{n,1}(k\ell) &= \frac{i}{\sqrt{2}} \sqrt{n+1} \left(\ell^2k^2 + \frac{3}{2}(n+1)\right) 
\ea
and $\beta_{n,s} = \alpha_{n,s}$ for all other $s$.
The corresponding eigenvalues are
\be
H \Psi'_{n,k} =  \hbar \Omega_0 \left( 1 + \frac{k \ell^2}{L} \right) \left( n + \frac{1}{2} \right) \Psi'_{n,k}
\ee
where $\Omega_0 := B_0/m$.
Rewriting $\Upsilon_{n,k}$ in the energy eigenbasis yields, to first order in $L^{-1}$,
\be
\Upsilon_{n,k} = \sum_{n'}\int\!dk'\, \Psi'_{n',k'} \la\Psi'_{n',k'}|\Upsilon_{n,k}\ra = \Psi'_{n,k} + \frac{\ell}{2L} \sum_{s\ne 0} \frac{\alpha_{n,s}(k\ell) - \beta_{n,s}(k\ell)}{s} \Psi_{n+s,k}
\ee
and so we have
\ba
\la\Upsilon_{n,k'}| e^{iHt/\hbar}\scr O e^{-iHt/\hbar}|\Upsilon_{n,k}\ra &=  e^{i\Omega_0(n+\sfrac{1}{2})(k'-k)\ell^2t/L}\la\Psi'_{n,k'}| \scr O |\Psi'_{n,k}\ra \nonumber\\
&\qquad + \frac{\ell}{2L}\times \text{\it ``terms with factors $e^{\pm i \Omega_0 t}$''}\,.
\ea
The phases $e^{\pm i \Omega_0 t}$ are fast oscillating from the perspective of the effective theory, so $K$ will eliminate such terms. 
The phase factor on the first term is, on the other hand, slowly varying due to condition \eqref{pertcondalpha2}, and therefore 
the low-pass filter $K$ will preserve it. We thus obtain
\be
\la\Upsilon_{n,k'}|\ol{\scr O}(t) |\Upsilon_{n,k}\ra =  e^{i\Omega_0(n+\sfrac{1}{2})(k'-k)\ell^2t/L}\la\Psi'_{n,k'}| \scr O |\Psi'_{n,k}\ra \,.
\ee
For operators $\scr O$ that, like $v_y$, commute with $p_y$,  this phase factor drops out since its argument is proportional to $k-k'$ and $\la\Psi'_{n,k'}| \scr O |\Psi'_{n,k}\ra \propto \delta(k-k')$.
Therefore
\be
\la\Upsilon_{n,k'}|\ol{\scr O} |\Upsilon_{n,k}\ra = \la\Psi'_{n,k'}| \scr O |\Psi'_{n,k}\ra
\ee
revealing that, for such operators and to first order in $L^{-1}$, matrix elements in the $\Upsilon$ basis should be replaced by the corresponding matrix elements in the $\Psi'$ basis.

For the drift velocity we have
\ba
\la \Psi'_{n,k'}| v_y | \Psi'_{n,k}\ra =&\,\, \la \Psi_{n,k'}| v_y^{(0)} | \Psi_{n,k}\ra + \frac{1}{L} \la \Psi_{n,k'}| v_y^{(1)} | \Psi_{n,k}\ra \no
& + \frac{\ell}{2L} \sum_{s\ne 0} \frac{\beta_{n,s}(k\ell)}{s} \la \Psi_{n,k'}| v_y^{(0)} | \Psi_{n+s,k}\ra \no
& + \frac{\ell}{2L} \sum_{s\ne 0} \frac{\beta_{n,s}(k'\ell)^*}{s} \la \Psi_{n+s,k'}| v_y^{(0)} | \Psi_{n,k}\ra
\ea
where $v_y = v_y^{(0)} + L^{-1} v_y^{(1)}$, with
\ba
v_y^{(0)} &= \frac{1}{m} (p_y - B_0 x) \\
v_y^{(1)} &= - \frac{B_0}{2m}x^2\,.
\ea
Using formula \eqref{matrixLandaugauge} we find that the only non-vanishing matrix elements are
\ba
\la \Psi_{n,k'}| v_y^{(1)} | \Psi_{n,k}\ra &= -\frac{\hbar}{2} \left(\ell^2k^2 + n + \frac{1}{2}\right) \delta (k-k') \\
\la \Psi_{n+1,k'}| v_y^{(0)} | \Psi_{n,k}\ra &= \frac{i\hbar}{\sqrt{2}\ell} \sqrt{n+1} \,\delta(k-k') \\
\la \Psi_{n-1,k'}| v_y^{(0)} | \Psi_{n,k}\ra &= -\frac{i\hbar}{\sqrt{2}\ell} \sqrt{n}\, \delta(k-k')
\ea
and their complex conjugates. 
In any normalized state $\Upsilon_{n,f}$, the expectation value of the time-coarse-grained velocity operator is then
\be
\la \Upsilon_{n,f}|\, \ol{v_y} \,| \Upsilon_{n,f}\ra = \frac{\hbar}{mL} \left( n + \frac{1}{2} \right)\,.
\ee
Given the matching prescription stated in \eqref{matchingBpert}, we see that the right-hand side is identified with $\mu/L$, which exactly matches the classical result.

The result is also compatible with the drift velocity predicted by the quantum effective theory for open isomagnetic contours, discussed in Sec.~\ref{openBlines}.
For $B$ given by \eqref{BlinearL}, the variable (classically) conjugated to $\log B$ is given by 
\be
\chi(x,y) = \frac{L}{B_0} B(x)^2 y\,.
\ee
To first order in $L^{-1}$ we get, for the quantum operator,
\be
\wh y = \frac{1}{LB_0}\wh\chi
\ee
noticing that no operator-ordering ambiguities appear (at this order).
We thus see from \eqref{chisol} that
\be
\wh v_y = \frac{d\wh y}{dt} = \frac{1}{LB_0} \frac{d\wh\chi}{dt} = \frac{\mu}{L} + \ca O(L^{-2})
\ee
which also gives no quantum correction to the classical result.\footnote{Note that Oikawa {\it et al.}~\cite{chan2016quantum,Chan:2017dex} 
have studied the quantum mechanical drift velocity from a Heisenberg picture perspective, at first order in $L^{-1}$. In \cite{chan2016quantum} the operator-level formula for the drift velocity is $v_y = H/B_0$, which matches our result when restricted to states in $\ca V_\mu$. However, they compute the expectation value of that operator in a particular Gaussian state (with dispersion equal to $\ell$ in both directions), obtaining
\[
v_y = \frac{1}{L}\frac{\la \pi_x \ra^2 + \la \pi_y \ra^2}{2m B_0} + \frac{3\hbar}{4mL} \,,
\]
and propose that the second term represents a quantum correction to the classical drift result. In our interpretation, that $\hbar$-term is due to a different choice of state (outside our $\ca V_\mu$, defined in \eqref{matchnonuniformB}) and different identification of $\mu$ (because $\la \pi \ra^2 \ne \la \pi^2 \ra$).}

\section{Discussion}
\label{SecDiscussion}
 
We analyzed the non-perturbative quantization of the effective theory of 
guiding center drift motion of a charged particle in two spatial dimensions in a 
time-independent magnetic field. The effective theory is a good approximation when the gyro radius is small compared to the gradient scale of the magnetic field.
The quantization is developed from the perspective of beings that can only observe the guiding center motion. In particular, the guiding center Lagrangian involves no
metric structure (notion of lengths and angles) on such space; rather, the space is only equipped with an area structure, so the quantization cannot refer to a metric
without importing information from outside the effective theory.

The phase space of the system is just the physical space itself, with a symplectic form given by the magnetic flux 2-form.
We adopted the viewpoint that the quantization should  be based on constructing 
a unitary representation of
a transitive group of symplectomorphisms on the phase space, preserving any other background structure to the best extent possible. In this case, the group was required to possess a non-trivial subgroup preserving the area 2-form. The 
Hamiltonian
charges associated with the algebra are then the ``quantizing observables'', undergoing (undeformed) quantization. The area-preserving requirement
implies that at least one of these observables, $Q$, must be a function of the magnetic scalar $B$ (or, more precisely, $dB\wedge dQ = 0$). The quantization always leads to a non-commutative physical space.

Since the area structure is much weaker than a metric structure, the quantization is generally more ambiguous. But some special configurations of the magnetic field admit a ``most natural'' quantization, such as 
the case where space is foliated by open or concentric closed isomagnetic contours. 
The quantization with open contours 
(and nowhere vanishing $dB$)
is relatively simple, based on the standard Heisenberg group (where the configuration variable was taken as $\log B$), and it does not reveal any surprising features. In particular, the Heisenberg equations of motion mimic exactly the classical equations. 
The case of closed contours is 
more subtle. We 
quantize it in terms of the $\wt{E(2)}$ group, where the $SO(2)$ generator is associated with the magnetic flux enclosed by contours of constant $B$,
and the non-compact generators
correspond to non-commutative translations in the plane. 
By studying the unitary irreducible representations of the Lie algebra $\wt{\fr e(2)}$, we find that the spectrum of the magnetic flux operator $\wh\phi$ has uniform spacing $2\pi\hbar$ and is bounded from below, with minimum value $\pi\hbar$ if 
the value of the quantum Casimir matches the classical one. Some 
surprising aspects emerge. The quantum equations of motion feature a ``quantum derivative'' of the magnetic scalar, $B^\Delta(\wh\Phi) := [B(\wh\Phi)-B(\wh\Phi-\hbar)]/\hbar$, instead of the usual derivative appearing in the classical equations. 
Most surprisingly, the effective theory predicts a quantization of space itself:
the flux is a proxy for the radius operator (relative to the central point
of the nested isomagnetic contours), so its discrete spectrum entails a quantization of 
the radial position of the particle.
In the microscopic theory a particle can be placed anywhere, whereas
in the effective theory the particle can ``live'' only on certain contours.

One of the main questions posed in the introduction, as a motivation, was whether coarse-graining commutes with non-perturbative (i.e., global in phase space) quantization. The theory we considered is simple enough that one can quantize both the microscopic and macroscopic theories (the coarse graining being over the small size and fast time scales of the gyro motion). 
Referring to the diagram presented in the introduction, we also analyzed the ``quantum coarse graining arrow'' from MQT to EQT. We identified a certain subspace $\ca V$ of the Hilbert space of the MQT, associated with an isolated spectral cluster of $\wh W_\mu = \wh H - \mu \wh B$ near zero, that could be matched to the Hilbert space of the EQT. We derived
a quantitative condition for this matching prescription to be stable under dynamical evolution of the MQT, provided that $\wh W_\mu$ satisfies a natural spectral condition. In a rough estimation, under the condition that $n\rho_n/L \ll 1$ (where $n \sim m\mu/\hbar$, $\rho_n$ is the radius of the gyro orbit for states in $\ca V$, and $L$ is the length scale over which $B$ varies), a state initially in $\ca V$ will evolve adiabatically, 
over macroscopically large times,
to another state approximately in $\ca V$. This condition is more stringent than its classical counterpart, $\rho/L \ll 1$, which we interpreted to be a consequence of defining the matching prescription with respect to a sharp subspace $\ca V$ of the MQT, as opposed to considering coherent states peaked at $\ca V$. 
To complete the matching prescription, we also proposed a prescription to map MQT operators to EQT operators via time coarse graining over the fast oscillatory gyro motion. With this, we concluded that the drift velocity predicted by the MQT, in the case of a quasi-uniform linear $B$ field, agrees with the prediction of the EQT at leading order in $L^{-1}$.

Regarding the aforementioned spurious EQT prediction 
of spatial quantization in the case of closed isomagnetic contours,
we have argued that this is not truly incompatible with the fundamental theory, since the discretization scale is never larger than the resolution scale (i.e., the gyro orbit size). 
By considering the matching prescription in the case of a uniform magnetic field, we proposed an explanation for this flux-quantization feature from the viewpoint of the MQT. In fact, there are two aspects at play. First, the flux operator (proportional to the squared-radius), which has a continuous spectrum in the MQT, does not act within a single Landau level, 
so cannot be directly compared to an operator in the EQT.
Projecting it to any particular level produces an operator with discrete spectrum 
that agrees, up to a shift, with that of the squared-radius operator of the EQT.
Remarkably, identifying the radius operator corresponding to the center of the gyro orbit by subtracting the fast-oscillating part of $\wh r(t)$, one obtains a guiding center squared-radius operator in the MQT that \emph{does} act within each Landau level, 
and has precisely the same discrete spectrum as the squared-radius operator of the EQT 
based on the $\wt{E(2)}$ quantization.
(It is interesting to notice that the alternative EQT based on the $\text{PSL}(2)$ quantization, which we deemed \emph{a priori} less natural due to the introduction of a parameter with dimensions of flux in the classical algebra, predicts a shifted spectrum for the squared-radius operator, yielding a slightly imperfect match to the MQT.)
In the non-uniform $B$ case we can define the ``contour flux operator'' in the MQT
by $\phi(\wh x, \wh y)$, where $\phi(x,y)$ is the magnetic flux enclosed by the isomagnetic contour passing through $(x,y)$. 
We conjecture that time coarse graining 
and projecting onto $\ca V$, according to the operator matching prescription,  will similarly produce an operator with discrete spectrum equal (or at least close) to that of the EQT.

We conclude that,  for the selected aspects considered in this paper,
the quantization-coarse graining diagram commutes within the regime of the guiding center approximation.
However, we must emphasize that nothing in the quantization of the effective theory indicates that the prediction of space quantization, in the case of closed isomagnetic contours, is spurious. Thus, scientists in a hypothetical world that have attempted to quantize their classical guiding center theory would not suspect, from purely theoretical considerations, that space is {\it not} quantized. Of course, if they ever reached the necessary experimental precision to probe this effect, they would realize that their theory breaks down, and a new degree of freedom (the gyro motion) would manifest. 
This suggests the more general lesson that non-perturbative predictions of the quantization of an effective classical theory should be taken with caution, since there may be no internal, theoretical indications of 
a description breakdown.
In particular, such predictions coming from nonperturbative quantization of general relativity,
such as for instance the quantization of the area operator in Loop Quantum Gravity \cite{Ashtekar:2021kfp},
might similarly be spurious even if the theory is internally consistent at all scales.

We considered in this paper only the simplest settings in which 
to examine the quantization of guiding center motion 
in a static magnetic field on a plane
and compare to the microscopic theory. The study
could be extended in a number of  directions.
One we have already pursued \cite{GCHallHawking} is to 
add an external electrostatic potential to the case
of a uniform magnetic field, to study the 
quantum Hall effect and  
Stone's model of Hawking radiation 
of edge modes in a quantum Hall system \cite{Stone:2012cx},
where we find that 
the predictions of the effective quantization match 
with the main results of the microscopic quantum theory. 
Other potentially interesting ideas could be to consider
magnetic fields with more complicated contour topology,
motion on two-dimensional 
spatial manifolds with nontrivial topology,
or motion in three dimensions. In the latter case the reduced phase space would be four dimensional, with two coordinates
labeling the field lines, and a pair
labeling position and momentum along the
field lines. The symplectic form
likely couples these variables, 
and the resulting quantum theory
would include something like one-dimensional 
quantum mechanics 
along the magnetic field lines, 
coupled  with transverse
drift due not only to $B$ gradients but also to field line 
curvature.
Alternatively, one could consider 
time-dependent magnetic fields, which would entail 
both induced electric fields and 
a time-dependent symplectic form, and thus a
time-dependent non-commutative quantum geometry for
spatial coordinate operators. 
Finally, a dynamical magnetic moment degree of freedom could be introduced, 
allowing for transitions which might correspond in the microscopic quantum theory
to Landau level transitions.

\section*{Acknowledgements}
This work was supported in part by NSF grant PHY-2309634
and by Perimeter Institute for Theoretical Physics. 
Research at Perimeter Institute is supported in part by the Government of Canada through the Department of Innovation, Science and Economic Development and by the Province of Ontario through the Ministry of Colleges and Universities.


\appendix
\numberwithin{equation}{section}

\section{Validity conditions for the guiding center approximation}
\label{AppClassicalGCconditions}

If $\Omega=qB/m$ were uniform, the particle would move along a circle with angular velocity $\Omega$. Up to a choice of center and initial phase, the trajectory (of radius $\rho$) would be given by
\be\label{rBconst}
r(t) = \rho \left( \cos(\Omega t), - \sin(\Omega t) \right)\,.
\ee
We are interested in the limit where 
the particle orbit is small compared to the scale over which the magnetic field 
changes. In this case the particle moves in loops, with frequency $\Omega$, velocity $v$ and radius $\rho \approx v/\Omega$, while its center slowly drifts due to the inhomogeneous magnetic field. When $\rho$ is very small, one would only notice the overall drifting motion, and not the fast gyro motion of the particle. Starting from the exact, ``microscopic'' theory, let us study this coarse-grained, ``macroscopic'' dynamics.

Let $x(t) = z(t) + r(t)$, where $z$ should describe the trajectory of the {\emph guiding center} and $r$ the {\emph gyro motion}. We can approximate \eqref{exactEOM}, expanding $\Omega(z+r)$ to first order in $r$, as
\be\label{approxEOM}
\frac{d^2 z^i}{dt^2} = \Omega(z) \epsilon_{ij} \frac{dz^j}{dt} + \epsilon_{ij} \frac{\partial\Omega(z)}{\partial z^k} r^k \frac{dz^j}{dt} + \epsilon_{ij} \frac{\partial\Omega(z)}{\partial z^k} r^k \frac{dr^j}{dt}
\ee
where we have taken $r(t)$ to be given as in \eqref{rBconst}, with $\Omega$ evaluated at $z(t)$, so $\frac{d^2 r^i}{dt^2} = \Omega(z) \epsilon_{ij} \frac{dr^j}{dt}$.  
As $z$ presumably moves slowly, not varying much through several cycles of $r$, we may time average this equation over time scales much longer than $\Omega^{-1}$. 
That is, we replace
\ba
r^k &\rightarrow \ol{r^k} = 0 \no
r^k \frac{dr^j}{dt} &\rightarrow \ol{r^k \frac{dr^j}{dt}} = \frac{1}{2}\Omega(z) \rho^2 \epsilon_{jk}\,.
\label{timeavg}
\ea
The quantity multiplying $\epsilon_{jk}$ above can be interpreted as the {\it dipole moment}, $\mu$, associated with the gyro motion: the electric current is one per cycle (since $q=1$), whose period is $2\pi/\Omega$, and the enclosed area is $\pi \rho^2$, so we have
\be\label{mudef}
\mu := Ia = \frac{\Omega}{2\pi}(\pi\rho^2) = \frac{1}{2}\Omega \rho^2\,.
\ee
In this regime, $\mu$ is an adiabatic invariant.\footnote{This can be seen as follows. Given a closed curve $\ca C$ in a symplectic manifold, with symplectic potential $\theta$, the quantity $J := \int_{\ca C}\theta$ is conserved as $\ca C$ is pushed through a Hamiltonian flow. In our situation, a cycle of gyro motion is almost a closed curve, so we can take $\ca C$ to be a closed loop approximating it. Also, in a cycle, $\theta = p_i dx_i \approx m v_i dr_i + qA_idx^i$. So $J \approx \int_0^{2\pi/\Omega} m v^2 dt  + q\int_{\ca C}A = 2\pi m v^2/\Omega + q(-B)\pi\rho^2$. Finally, $v \approx \rho\Omega$, implying that $J \approx 2\pi m \mu$.} We thus have, at this stage,
\be\label{approxEOM2}
\frac{d^2 z_i}{dt^2} = \Omega(z) \epsilon_{ij} \frac{dz_j}{dt} - \mu \frac{\partial\Omega(z)}{\partial z_i}\,.
\ee
It turns out that the second-derivative term is negligible in this context (which will be justified below), so the equation reduces to
\be\label{driftEOM}
\frac{dz_i}{dt} = - \mu \epsilon_{ij} \frac{\partial \log\Omega}{\partial z_j}
\ee
which is a first-order differential equation describing the guiding center motion. 

To see in what conditions equation \eqref{driftEOM} is valid, we employ an a posteriori argument. 
First, to justify the time-averaging employed in \eqref{timeavg}, we need that the drift motion should be slow compared to the gyro motion.
Assuming \eqref{driftEOM} is accurate, over a (supposedly small) period of the gyro motion, $z$ would  drift by an amount
\be
\delta z_i \approx \frac{dz_i}{dt} \frac{2\pi}{\Omega} = - \frac{1}{2} \rho^2 \epsilon_{ij} \frac{\partial \log\Omega}{\partial z_j}\,.
\ee
Thus, in order for $\delta z \ll \rho$, we need
\be
\frac{1}{2}\rho \left|\frac{\partial\log\Omega}{\partial z}\right| \ll 1\,.
\ee
Second, to justify that the second-order derivative is negligible in \eqref{approxEOM2}, we again assume that \eqref{driftEOM} is accurate, from which it would follow
\be
\frac{d^2 z_i}{dt^2} = - \mu \epsilon_{ij} \frac{\partial^2 \log\Omega}{\partial z_k \partial z_j} \frac{dz_k}{dt}\,.
\ee
The required condition is that this needs to be much smaller than the analogous term in the right-hand side of \eqref{approxEOM2},
\be
\left| \mu \frac{\partial^2 \log\Omega}{\partial z_i \partial z_j} \right|\ll \Omega\,,
\ee
which can be rewritten  as
\be
\frac{1}{2}\rho^2\left| \frac{\partial^2 \log\Omega}{\partial z_i \partial z_j} \right|\ll 1\,.
\ee
In conjunction, these conditions make precise what is meant by a regime of ``strong magnetic field'', in which the guiding center motion can be well described by equation \eqref{driftEOM}: we require that the zeroth and first  order Taylor series of $\log \Omega$ are accurate within the area enclosed by a cycle of the gyro motion (around any point).

\section{Toroidal phase space}
\label{AppTorus}

In this appendix we derive an exact formula for the density of states in the effective theory, when the region accessible to the particle is bounded.
For this purpose, we could introduce an infinite potential barrier confining the particle to such a region of the plane. This is complicated, however, due to the necessity of deciding on an operator ordering for $V(x^1,x^2)$, and then studying the corresponding eigenstates of energy. A simpler approach is to place the particle in a torus space, $T^2$. In this case, we cannot quantize the Heisenberg algebra $\{x^1,x^2,1\}$ 
since $x^1$ and $x^2$ are not well-defined on the torus --- they are not periodic. In fact, there is no complete finite dimensional quantizing algebra in this case, and one is forced to use the algebra of all smooth 
functions on the torus, denoted by $\ca A_{T^2}$. A basis for these functions is
\be\label{torusalgebra}
\left\{P_{n_1,n_2}:=e^{i2\pi (n_1 x^1/a_1 + n_2 x^2/a_2)}\,;\,\, (n_1,n_2)\in \bb Z^2\right\}
\ee
where $a_1$ and $a_2$ are the periods in the $x^1$ and $x^2$ directions, respectively.
The quantization of this problem has been studied previously (see for example 
\cite{manoliu1997quantization,  gotay1995full,
gotay2000obstructions}).
Here we provide a simple
derivation of a suitable quantization of this system.\footnote{Our result agrees with the one from geometric quantization \cite{manoliu1997quantization}. In particular, both predict that a quantization is only possible when the total magnetic flux is an integer multiple of $2\pi\hbar$, and the corresponding Hilbert space has a dimension equal to the number of flux quanta. On the other hand, it is shown in \cite{gotay1995full}  that the case of unit flux 
also admits a quantization of the full algebra of observables, represented on an infinite-dimensional Hilbert space. However, that quantization violates the Casimir matching principle \cite{AndradeeSilva:2020ofl} to a large degree, and some classically bounded observables 
are quantized to unbounded operators \cite{gotay2000obstructions}.}

The basic idea is to take advantage of the fact that the torus is a quotient space, $T^2 \sim \bb R^2/\bb Z^2$. Accordingly, we can hope to derive a
quantization of the torus by a restriction
of the Hilbert space for the plane respecting this quotient, 
as we will explain. Consider any quantization of the elements of \eqref{torusalgebra} to operators on $L^2(\bb R)$,
\be
P_{n_1,n_2} \mapsto \wh P_{n_1,n_2}\,.
\ee
In particular, one could consider promoting the functions $P_{n_1,n_2}(x_1,x_2)$ to operators $\wh P_{n_1,n_2} := P_{n_1,n_2}(\wh x_1,\wh x_2)$, using the representation \eqref{xyrep}, but of course this procedure would be subjected to operator-ordering ambiguities. 
Given any quantization prescription, denote the operator algebra generated by $\wh P_{n_1,n_2}$ on $L^2(\bb R)$ by $\wh{\ca A}_{T^2}$.
Note that $\wh{\ca A}_{T^2}$ is not required to furnish a strict representation of ${\ca A}_{T^2}$, and in general it will rather furnish a representation of some $\hbar$-deformation of ${\ca A}_{T^2}$.
Now consider the operators
\ba
\wh T_1 &:= e^{i B a_1 \wh x^2/\hbar} \\
\wh T_2 &:= e^{- i B a_2 \wh x^1/\hbar}\,.
\ea
These generate translations in $x^1$ by $a_1$ and $x^2$ by $a_2$, respectively. Classically, the corresponding symplectomorphisms are simply the identity on the torus phase space. We 
thus demand that $\wh{\ca A}_{T^2}$ must satisfy
\be\label{torusZ2invariance}
\wh T_i^\dag \wh P \wh T_i = \wh P \,,\,\,\,\text{\it for all}\,\, \wh P \in \wh{\ca A}_{T^2}
\ee
with $i=1,2$. This expresses the invariance of these operators under the $\bb Z^2$ quotient, and thus their well-definedness on the torus. 
Note that since the operators $\wh T_i$ are unitary, condition \eqref{torusZ2invariance} implies that they are Casimir operators 
for the algebra $\wh{\ca A}_{T^2}$. However, they are not proportional to the identity on $L^2(\bb R)$, 
revealing that $\wh{\ca A}_{T^2}$ is not irreducibly represented on $L^2(\bb R)$. 
The goal is then to find the appropriate invariant irreducible subspace of $L^2(\bb R)$ to which $\wh{\ca A}_{T^2}$ could be restricted.\footnote{As we will see, the candidate irreducible subspaces consist entirely of non-normalizable states, and therefore are not legitimate Hilbert subspaces. This is a peculiarity of infinite-dimensional representations, where reducible representations need not contain any nontrivial irreducible invariant subspaces. This notwithstanding, in this appendix we will pursue this ``heuristic'' approach, which nonetheless will produce the correct irreducible representation. Alternatively, one could presumably formulate the construction rigorously in the framework of rigged Hilbert spaces, in particular using the ``rigging map'' \cite{Marolf:2000iq,Giulini:1999kc,Alonso-Monsalve:2025lvt} associated with the $\bb Z^2$-action generated by $\wh T_i$, to arrive at the same irreducible representation.}

For any $\wt\psi$ in such irreducible subspace, we must have
\be\label{Tiwtpsi}
\wh T_i\wt\psi = e^{i\alpha_i}\wt\psi
\ee
where $\alpha_i \in \bb R$, to accord with Schur's lemma and the unitarity of $\wh T_i$. We shall restrict to the case $\alpha_i = 0$, which should lead to the non-projective representation.
Again using the Baker-Campbell-Hausdorff formula
we see that
\be\label{T2T1}
\wh T_2 \wh T_1 = e^{-iBa_1a_2/\hbar} \wh T_1 \wh T_2\,.
\ee
Note that there would be no $\wt\psi$ satisfying \eqref{Tiwtpsi} unless the phase factor 
in \eqref{T2T1}
equals $1$. That is,
\be\label{Diracmonopole}
\frac{Ba_1a_2}{\hbar} = 2\pi N \,,\,\, \text{\it with }\, N \in \bb Z\,.
\ee
This is the torus version of Dirac's famous monopole charge quantization condition, since $Ba_1a_2$ is the total magnetic flux through the torus. (In geometric quantization, this is the standard prequantization condition for a complex line bundle to exist, $\frac{1}{2\pi\hbar} \int_{T^2} \omega \in \bb Z$.)

Given the flux quantization condition  \eqref{Diracmonopole},
we can now find the solutions to 
\ba
\wh T_1\wt\psi(x^1) &= \wt\psi(x_1 - a_1) = \wt\psi(x_1) \\
\wh T_2\wt\psi(x^1) &= e^{- i 2\pi N x^1/a_1}\wt\psi(x^1) = \wt\psi(x^1)\,.
\ea
The first equation simply says that $\wt\psi(x^1)$ must be periodic with period $a_1$. The second equation implies that $\wt\psi(x^1)$ must be a superposition of Dirac deltas equally spaced,
\be\label{deltas}
\wt\psi(x^1) = \sum_{k\in\bb Z} C_k\, \delta\!\left(x^1 - a_1\frac{k}{N}\right)
\ee
with $C_k \in \bb C$. 
Thus, in order for the sum of Dirac deltas to be $a_1$-periodic, we need
\be\label{C's}
C_{k+N} = C_k
\ee
revealing that there are only $|N|$ independent parameters spanning the subspace with $\wh T_i = 1$. This suggests that $\wh{\ca A}_{T^2}$ should be represented irreducibly on an $|N|$-dimensional Hilbert space. 
The norm on this Hilbert space cannot be the one induced from $L^2(\bb R)$, since all states are non-normalizable. Nevertheless, we can easily construct an actual 
representation based on the action of $\wh{\ca A}_{T^2}$ on wavefunctions 
of the form \eqref{deltas} satisfying \eqref{C's}, given that
according to the requirement \eqref{torusZ2invariance} 
such an action is necessarily closed within this subspace. 

As a particular example, consider the quantization of $\ca A_{T^2}$ obtained by promoting $P_{n_1,n_2}$ with the following operator-ordering choice,
\be\label{X2n2psi}
\wh P_{n_1,n_2}\psi(x^1):= e^{i2\pi (n_1 \wh x^1/a_1 + n_2 \wh x^2/a_2)}\psi(x^1) = e^{-i2\pi^2n_1n_2 \hbar/(Ba_1a_2)} e^{i2\pi n_1 x^1/a_1}\psi\left(x^1 - \frac{2\pi\hbar}{Ba_2}n_2\right)\,,
\ee
If $\{|k\ra;\, k = 0,1,\ldots |N|-1\}$ is an orthonormal basis for an associated $|N|$-dimensional Hilbert space $\ca H_{T^2}$, then the action of the $\wh P_{n_1,n_2}$ above 
on the Dirac comb functions $\wt\psi$ \eqref{deltas}
motivates the corresponding action of associated operators $\wh P_{n_1,n_2}$ on $\ca H_{T^2}$,
\be
\wh P_{n_1,n_2}|k\ra:=  e^{-i\pi n_1n_2/N} e^{i2\pi n_1 k/N}|k + n_2 \,\,\text{mod}\,N\ra\,.
\ee
This properly defines a (non-projective) unitary, irreducible representation of $\wh{\ca A}_{T^2}$ on $\ca H_{T^2}$, and therefore yields a valid quantization on the torus.

Despite the inherent ambiguity in quantizing $\ca A_{T^2}$ into $\wh{\ca A}_{T^2}$, the dimensionality of the torus Hilbert space is robust, provided that condition \eqref{torusZ2invariance} is satisfied. In fact, the derivation relies only on the properties of $\wh T_i$.
We conclude that the density of states in the effective theory is
\be
\text{\it density of states} = \frac{|N|}{a_1a_2} =  \frac{|B|}{2\pi\hbar}\,,
\ee
matching with \eqref{densitystatesLandau1}.

\section{Alternative quantization for closed isomagnetic contours}
\label{AppClosedBAlt}

In this appendix we consider the most general ($\phi$-dependent) weights to regularize the algebra of $\phi$, $\sin\theta$ and $\cos\theta$ \eqref{closedBcharges0} in Sec.~\ref{subsecPoissonclosed}, and thereby produce an algebra of smooth observables that exponentiates to a group action on the phase space for closed isomagnetic contours. We show that these   weights can either lead to the $\wt{\fr{e}(2)}$ algebra discussed in that section, or to a one-parameter family of $\fr{sl}(2)$ algebras which we explore here.

Consider the set of functions 
\ba
J &:= h(\Phi) \\
X &:= f(\Phi)\cos\theta \\
Y &:= g(\Phi)\sin\theta
\ea
for arbitrary functions $h$, $f$ and $g$ that are smooth away from the origin. From \eqref{closedBcharges0}, we find
\ba
&\{X, Y\} = -f(\Phi)g'(\Phi)\sin^2\theta - f'(\Phi)g(\Phi) \cos^2\theta \label{weightedalg0}\\
&\{X, J\} = - f(\Phi) h'(\Phi) \sin\theta \label{weightedalg1}\\
&\{Y, J\} = g(\Phi) h'(\Phi)\cos\theta\,. \label{weightedalg2}
\ea
In order for this to form a closed algebra, 
it is necessary that the right-hand sides of \eqref{weightedalg1} and \eqref{weightedalg2} give something proportional to $Y$ and $X$, respectively. Thus,
\ba
f(\Phi) h'(\Phi) &= a g(\Phi) \\
g(\Phi) h'(\Phi) &= b f(\Phi)\,.
\ea
Assuming that $f$ and $g$ are almost everywhere nonzero, we conclude that $(h'(\Phi))^2 = ab$, which implies that
$h(\Phi) = c\Phi + d$
for constants $c$ and $d$. Since a scaling and shift of $\Phi$ would not change the algebra fundamentally (as it would correspond to a change of basis plus, possibly, a trivial central extension), we choose for simplicity 
\be
h(\Phi) = \Phi\,.
\ee
Then, $g(\Phi) = b f(\Phi)$, and by absorbing $b$ into $f$ (or, equivalently, rescaling $X$), we can set $b=1$.
The right-hand side of \eqref{weightedalg0} then gives
$- f(\Phi) f'(\Phi)$.
For algebraic closure of $J = \Phi$, $X$ and $Y$, and possibly also including the constant function $1$, this must be a linear function of $\Phi$,
\be
-f(\Phi) f'(\Phi) = -\eta \Phi - \nu
\ee
for constants $\eta$ and $\nu$. The general solution for $f$ is
\be
f(\Phi) = \sqrt{\eta \Phi^2 + 2 \nu \Phi + \gamma}
\ee
for some constant $\gamma$. 
From the argument 
leading to
\eqref{closedBcharges}
we see that $X$ and $Y$ will be smooth provided that $f(\Phi) \sim \sqrt{\Phi}$ near the origin. This requires that $\gamma = 0$ and $\nu>0$. With a further rescaling of $X$ and $Y$, we can set $\nu = 1$, which gives the most general allowed weight
\be
f(\Phi) = g(\Phi) = \sqrt{2\Phi + \eta\Phi^2}\,.
\ee
To ensure that $f$ is real everywhere, we must have $\eta \ge 0$. 

The resulting $\eta$-parametrized family of algebras is then
\ba
&\{X, Y\} = -\eta\Phi - 1 \\
&\{X, \Phi\} = - Y \\
&\{Y, \Phi\} = X \,.
\ea
The choice $\eta=0$ leads to an 
$\wt{\fr{e}(2)}$ algebra, as considered before, while the choice $\eta>0$ leads to a centrally-extended $\fr{sl}(2)$ algebra. This central extension is trivial, as it can be eliminated by redefining 
\ba
J &= \Phi + \frac{1}{\eta} \\
X &:= \sqrt{\frac{2}{\eta}\Phi + \Phi^2}\,\cos\theta \\
Y &:= \sqrt{\frac{2}{\eta}\Phi + \Phi^2}\,\sin\theta
\ea
leading to a standard $\fr{sl}(2)$ algebra for $J$, $X$ and $Y$,
\ba
&\{X, Y\} = -J \\
&\{X, J\} = - Y \\
&\{Y, J\} = X \,.
\ea
This algebra exponentiates to a $\text{PSL}(2,\bb R)$ group of symplectomorphisms on the phase space.\footnote{To identify the group, note that this Poisson algebra is 
isomorphic to a commutator algebra of the 
$2\times2$ matrices,  
with $Y\rightarrow \s_3/2$, $X\rightarrow\s_1/2$, $J \rightarrow i\s_2/2$, 
where $\s_i$ are the Pauli matrices, which form the Lie algebra of $\text{SL}(2,\bb R)$. 
Since $\Phi$ (and $J$) generates $\theta$-rotations on the phase space, $\exp(2\pi\{{}\cdot{},J\})$ is the identity,
whereas $\exp(2\pi i \s_2/2)$ is minus the identity, hence the symplectic group generated by the Poisson algebra is the quotient by $\{\pm I\}$, 
$\text{PSL}(2,\bb R)$.} 

Now we discuss the quantization in terms of this group.
The quantized algebra reads
\ba
&[\wh X, \wh Y] = - i\hbar \wh J \\
&[\wh X, \wh J] = - i\hbar \wh Y \\
&[\wh Y, \wh J] = i\hbar \wh X \,.
\ea
It has a Casimir operator given by
\be
\wh C := \wh J^2 - \wh X^2 - \wh Y^2\,.
\ee
In any irreducible (complex) representation, it must be a multiple of the identity. As standard in the analysis of representations of $\fr{sl}(2)$, the value of this Casimir is typically written as
\be
\wh C = \left(\zeta^2 - \frac{1}{4}\right)\hbar^2
\ee
where $\zeta^2 \in \bb R$, so $\zeta$ is either real or purely imaginary. (We introduced the $\hbar^2$ for dimensional consistency.) 
Classically this Casimir evaluates to $1/\eta^2$, and thus imposing Casimir matching yields
\be\label{zetaCasimirmatch}
\zeta = \pm \sqrt{\frac{1}{4} + \frac{1}{\hbar^2\eta^2}}\,.
\ee
It is convenient to define ladder operators
\be
\wh a := \frac{\wh Y+ i \wh X}{\hbar} \quad\text{and}\quad \wh a^\dag := \frac{\wh Y- i \wh X}{\hbar}\,. 
\ee
The irreducible representations of $\fr{sl}(2)$ are labeled by $\zeta$ and $s \in [0,1)$ and they have the general form \cite{bargmann1947irreducible,gel1947unitary,harish1952plancherel,
Kitaev:2017hnr}
\ba
\wh J |j; \zeta, s\ra &= \hbar(j+s) |j; \zeta, s\ra \\
\wh a^\dag |j; \zeta, s\ra &= \left( j+s+\zeta +\frac{1}{2}\right) |j+1; \zeta, s\ra \label{adagjzetas}\\
\wh a |j; \zeta, s\ra &= \left( j+s-\zeta -\frac{1}{2}\right) |j-1; \zeta, s\ra \label{ajzetas}
\ea
where $j$ is in some subset of $\bb Z$, and the states $|j; \zeta, s\ra$ are not in general
normalized. The representations with $s\ne 0$ exponentiate to projective representations of $\text{PSL}(2,\bb{R})$.
By computing $|\!|\wh a^\dag |j; \zeta, s\ra |\!|^2$ we obtain the relation
\be\label{slnormrel}
\left(j + s - \zeta + \frac{1}{2} \right) \la j; \zeta, s |j; \zeta, s\ra =  \left(j + s + \bar\zeta + \frac{1}{2} \right) \la j+1; \zeta, s |j+1; \zeta, s\ra\,.
\ee
We are interested in the representations where the spectrum of $J$ (and $\Phi$) are bounded from below, since classically $\Phi \ge 0$. Thus, there must exist some state $|j_0; \zeta, s\ra$ that is annihilated by $\wh a$. 
There are two cases:
\begin{enumerate}[label=(\roman*)]
\item $j_0 + s - \zeta -\sfrac{1}{2} = 0$, so the coefficient in \eqref{ajzetas} vanishes when $j=j_0$\,; 

\item $j_0 + s + \bar\zeta -\sfrac{1}{2} = 0$, so the norm of $|j_0-1; \zeta, s\ra$ vanishes\,. 
\end{enumerate}
Case (i) implies $j_0 + s = \zeta + \sfrac{1}{2}$, while case (ii) implies $j_0 + s = -\bar\zeta + \sfrac{1}{2}$.
Applying \eqref{slnormrel} for $j=j_0$, and taking from \eqref{zetaCasimirmatch} that $\zeta$ is real, gives for case (i), 
\be
\la j_0; \zeta, s |j_0; \zeta, s\ra =  \left(2\zeta + 1 \right) \la j_0+1; \zeta, s |j_0+1; \zeta, s\ra
\ee
and for case (ii),
\be
\left(-2\zeta + 1 \right) \la j_0; \zeta, s |j_0; \zeta, s\ra =  \la j_0+1; \zeta, s|j_0+1; \zeta, s\ra\,.
\ee
But notice from \eqref{zetaCasimirmatch} that 
$|\zeta| > \sfrac{1}{2}$, so if positive $\zeta$ is selected only case (i) will correspond to positive norms, and if negative $\zeta$ is selected only case (ii) will correspond to positive norms.
For either sign of $\zeta$, the spectrum of $\wh J$ will contain a minimum value element equal to 
$\hbar(|\zeta| + \sfrac{1}{2})$, with all other values related by $\hbar$ increments.
In fact, both signs of $\zeta$ define the same \emph{discrete series representation} of $\fr{sl}(2)$.\footnote{The action of $\wh a^\dagger$ and $\wh a$ for different signs of $\zeta$ may appear to be different from 
\eqref{adagjzetas} and \eqref{ajzetas}, but that is due to the fact that the basis is not normalized. Using an orthonormal basis we would have $\wh a^\dag |j\ra = \sqrt{(j+s+\sfrac{1}{2})^2 -\zeta^2}|j+1\ra$ and $\wh a |j\ra = \sqrt{(j+s-\sfrac{1}{2})^2 -\zeta^2}|j-1\ra$, with $j+s\ge |\zeta|+\sfrac{1}{2}$, which are clearly independent of the sign of $\zeta$.}
The spectrum of $\wh\Phi$ is then given by
\be
\text{Spec}(\wh \Phi) = \left\{\hbar(\kappa + j), j\in \bb N\cup\{0\}\right\}
\ee
where $j$ was redefined to start at $0$ and
\be
\kappa = \sqrt{\frac{1}{4} + \frac{1}{\hbar^2\eta^2}} + \frac{1}{2} - \frac{1}{\hbar \eta}\,.
\ee
Note that $\kappa \in (\sfrac{1}{2}, 1)$. In particular, the limit $\eta\to 0$ gives $\kappa = \sfrac{1}{2}$, recovering the $\wt{E(2)}$ quantization. (In fact, $\wt{E(2)}$ is the In\"on\"u-Wigner contraction of $\text{PSL}(2)$ via the limit $\eta\to0$.)

\section{Matrix elements between Landau level states}
\label{AppLLmatrix}

In this appendix we derive a formula for computing the matrix elements of functions of position between arbitrary Landau level states in the quantum microscopic theory. The formula reduces the computation to a sum of matrix elements between $n=0$ level states. We present two (almost identical) versions, in the Landau and symmetric gauges.

Let $f(x,y)$ be an operator-valued function of spatial (Cartesian) coordinates, which by a slight abuse of notation can also be expressed in terms of complex coordinates, $z = x+iy$ and $\bar z = x-iy$, as $f(z,\bar z)$.
In the Landau gauge we have, for $a$ defined in \eqref{a-Landau},
\ba
[a, f(z,\bar z)] &= - i \sqrt{2}\ell \frac{\partial f}{\partial \bar z}(z,\bar z) \\
[a^\dag, f(z,\bar z)] &= - i \sqrt{2}\ell \frac{\partial f}{\partial z}(z,\bar z)\,.
\ea
Formal identities, to be understood in terms of their Taylor series, follow,
\ba
e^{ita}f(z,\bar z)e^{-ita} &= f\left(z, \bar z + \sqrt{2}\ell t\right)  \\
e^{-isa^\dag}f(z,\bar z)e^{isa^\dag} &= f\left(z - \sqrt{2}\ell s, \bar z\right)\,,
\ea
and from these we obtain, also using $[a,a^\dag] = 1$, the ``generating functional''
\be
e^{ita}f(z,\bar z) e^{isa^\dag} = e^{-ts} e^{isa^\dag} f\left(z - \sqrt{2}\ell s, \bar z + \sqrt{2}\ell t \right) e^{ita}
    \ee
Contracting with $n=0$ states and looking at the Taylor monomial of order $t^{n'}s^{n}$ gives
\be
\la \Psi_{0,k'}| a^{n'} f(z,\bar z) (a^\dag)^n | \Psi_{0,k}\ra = (-i)^{n+n'} \left(\frac{\partial}{\partial t}\right)^{n'} \left(\frac{\partial}{\partial s}\right)^{n} \la \Psi_{0,k'}| e^{-ts} f\left(z - \sqrt{2}\ell s, \bar z + \sqrt{2}\ell t \right)| \Psi_{0,k}\ra \bigg|_{t,s=0}
\ee
which yields
\be\label{matrixLandaugauge0}
\la \Psi_{n',k'}| f(z,\bar z) | \Psi_{n,k}\ra = \sum_{r=0}^{\text{min}(n,n')} \frac{i^{n-n'}\sqrt{n!n'!}}{r!(n-r)!(n'-r)!} \la \Psi_{0,k'}| \left(\sqrt{2}\ell\frac{\partial}{\partial \bar z}\right)^{n'-r}\! \left(\sqrt{2}\ell\frac{\partial}{\partial z}\right)^{n-r}\!\!f(z,\bar z) | \Psi_{0,k}\ra \,.
\ee
Of particular interest to Sec.~\ref{nonuniformBmatch} is the case of a function $g(x)$ of only $x$, for which this formula gives
\be\label{matrixLandaugauge}
\la \Psi_{n',k'}| g(x) | \Psi_{n,k}\ra = \sum_{r=0}^{\text{min}(n,n')} \frac{i^{n-n'}\sqrt{n!n'!}}{r!(n-r)!(n'-r)!} \la \Psi_{0,k'}| \left(\frac{\ell}{\sqrt{2}}\frac{\partial}{\partial x}\right)^{n+n'-2r}\!\!g(x) | \Psi_{0,k}\ra \,.
\ee

The exact same relations hold in the symmetric gauge, simply replacing $a$ (and $a^\dag$) by $c$ 
 (and $c^\dag$) as defined in \eqref{c-symmetric}. Formula \eqref{matrixLandaugauge0} becomes
\be\label{matrixsymmetricgauge0}
\la \Psi_{n',j'}| f(z,\bar z) | \Psi_{n,j}\ra = \sum_{r=0}^{\text{min}(n,n')} \frac{i^{n-n'}\sqrt{n!n'!}}{r!(n-r)!(n'-r)!} \la \Psi_{0,j'}| \left(\sqrt{2}\ell\frac{\partial}{\partial \bar z}\right)^{n'-r}\! \left(\sqrt{2}\ell\frac{\partial}{\partial z}\right)^{n-r}\!\!f(z,\bar z) | \Psi_{0,j}\ra \,.
\ee

\section{Adiabatic stability of the matching prescription}
\label{AppAdiabIneq}

In this appendix we establish the inequality \eqref{AdiabIneq}, 
\be\label{AdiabIneqApp}
|\!| P(t) - P(0) |\!| \le \varepsilon_1 + \varepsilon_2 t
\ee
used in Sec.~\ref{nonuniformBmatch} to quantify the adiabatic stability of the proposed matching prescription.\footnote{We acknowledge ChatGPT (5.6 Sol) for useful discussions related to this appendix.}
In this inequality, $P(t)$ denotes the Heisenberg-picture projector to the vector subspace $\ca V(t) \subset \ca H$, where $\ca V(0)$ is the vector subspace proposed to be associated with the Hilbert space of the effective theory at fixed magnetic moment $\mu$. As argued in that section, for the prescription to be dynamically consistent, $P(t)$ must remain close to $P(0)$, at least for some macroscopic interval of time.
Thus, our goal is to show that this inequality holds with constants $\varepsilon_1$ and $\varepsilon_2$ that are small in the guiding center regime.
The main result follows directly from a strong version of the adiabatic theorem enunciated in \cite{teufel2003adiabatic}, and attributed to Kato \cite{kato1950adiabatic}. This version applies to unbounded Hamiltonians, as well as cases in which the slowly varying vector subspace is infinite-dimensional or associated with a continuous portion of the spectrum (but still separated by gap from the rest of the spectrum), as occurs in our physical setting.
Due to the technical nature of this appendix, we organize the intermediate derivations into a series of lemmas, and state the main results as theorems (plus one corollary) for ease of reference.
We restate Kato's adiabatic theorem (Theorem~\ref{TheoremAdiabatic}) and its assumptions in a way that is automatically applicable in our context, and derive \eqref{AdiabIneqApp} as a simple consequence (Corollary~\ref{CorollaryAdiabatic}).
Then, we derive exact bounds for $\varepsilon_1$ (Theorem~\ref{theoremvareps1bound}) and $\varepsilon_2$ (Theorem~\ref{theoremvareps2bound}) in terms of expressions directly involving the microscopic Hamiltonian $H$ and the magnetic field $B$. 
To gain further intuition about these bounds, we also give a rougher estimate in the case of a sufficiently regular, quasi-uniform $B$ configuration with homogeneous derivatives, in the sense discussed at the end of Sec~\ref{nonuniformBmatch}.

The main idea for obtaining our result is to express the time-evolution of 
\be
W =  H - \mu B
\ee
in terms of the ``interaction picture'' with respect to the $\mu B$ term. That is, define
\be
S(t) := e^{i\mu B t/\hbar}
\ee
and
\be\label{wtWdef}
\wt W(t) := S(t) W S^\dag(t)\,.
\ee
Unlike the Heisenberg picture, in which the operators evolve with respect to $H$, the $\mu B$ Hamiltonian will evolve operators in a ``slow time scale'' allowing a direct application of the adiabatic theorem. In particular, note that for a constant $B$ field, the interaction picture operators will be constant in time, while their Heisenberg picture counterparts will in general oscillate at frequency $\Omega$.

Since operator norms appear throughout this appendix, it is worth introducing the relevant notation. 
If $\scr O : \ca K \to \ca L$ is a linear operator between normed spaces, its norm is defined by
\be
{|\!|\scr O|\!|}_{\ca K\to\ca L} := \sup_{\psi\in\ca K \backslash\{0\}}\frac{{|\scr O\psi|}_{\ca L}}{\,{|\psi|}_{\ca K}} \,,
\ee
where ${|\cdot|}_{\ca K}$ and ${|\cdot|}_{\ca L}$ are the norms on $\ca K$ and $\ca L$, respectively.
The space of operators $\scr O : \ca K \to \ca L$ with finite norm is denoted by $\ca B(\ca K, \ca L)$. If $\ca K = \ca L$, we write simply $\ca B(\ca K) = \ca B(\ca K, \ca L)$.
If $\ca L$ is the Hilbert space $\ca H$, $\scr O$ is an operator defined on a domain $\text{Dom}(\scr O) \subset \ca H$, and $\ca K = \text{Dom}(\scr O)$ (with its norm inherited from $\ca H$), the subscripts will be omitted, that is,
\be
{|\!|\scr O|\!|} := {|\!|\scr O|\!|}_{\text{Dom}(\scr O)\to\ca H}
\ee
and $|\cdot| := {|\cdot|}_{\ca H}$.
Given an operator $A :\ca D \to \ca H$, where $\ca D \subset \ca H$, define the \emph{graph norm of $A$} as
\be
|\psi|_A := \sqrt{|\psi|^2 + |A\psi|^2} \,.
\ee
Define $\ca D_A$ to be $\ca D$ equipped with the graph norm of $A$. 
For operators $\scr O : \ca D \to \ca H$, we denote
\be
{|\!|\scr O|\!|}_A := {|\!|\scr O|\!|}_{\ca D_A\to\ca H} = \sup_{\psi\in \ca D \backslash\{0\}}\frac{|\scr O\psi|}{\,\,{|\psi|}_A} \,.
\ee
Since the graph norm $|\cdot|_A$ is stronger than the Hilbert space norm $|\cdot|$, $\scr O$ can belong to $\ca B(\ca D_A, \ca H)$ even if it is not bounded with respect to $|\!|\cdot|\!|$.

We consider the following assumptions:
\begin{assumption}
\end{assumption}
\begin{enumerate}[label=(\roman*)]
\item \emph{Suppose that $H$ is a
self-adjoint, bounded from below operator on a dense domain $\ca D \subset \ca H$, and that $B$ is a bounded self-adjoint operator on $\ca H$ satisfying $B\ca D \subset \ca D$.}\footnote{The microscopic Hamiltonian $H = \frac{1}{2m}\! \left(\pi_1^2 + \pi_2^2\right)$, with $\pi_i = p_i -  A_i$, satisfies this assumption. Since $H$ acts as a second-order differential operator, a simple sufficient condition for $B\ca D \subset \ca D$ is that $B$ and all its spatial derivatives up to second order are bounded.\label{assumpiapp}}

\item \emph{Suppose that the spectrum of $W := H - \mu B$ is decomposable into a subset $I$ satisfying
\[
I\subset [r,s]
\]
with $s>r$, and a complement satisfying 
\[
\text{Spec}(W) \backslash I \subset (-\infty, r - \Delta]\cup [s + \Delta, \infty)
\]
for some $\Delta > 0$. (In this manner, $I$ is separated by the gap $\Delta$ from the rest of the spectrum.)
Define $\ca V \subset \ca H$ to be the vector subspace associated with the $I$ part of the spectral decomposition of $W$.}
\end{enumerate}
Now we prove a preliminary result necessary for applying the adiabatic theorem:
\begin{lemma}\label{assumpiii}
Under assumption (i), the ``interaction picture'' operator $\wt W(t)$ defined in \eqref{wtWdef} is self-adjoint on $\ca D$, and uniformly bounded from below, for all $t\in\bb R$. Furthermore, the curve $t \mapsto \wt W(t)$ is a real-analytic map from $\bb R$ into $\ca B(\ca D_W, \ca H)$.\footnote{An operator-valued curve $A :\bb R\to \ca B(\ca D_W,\ca H)$ is real-analytic on $\bb R$ if, for each $t_0\in\bb R$, there exist operators $A_n(t_0) \in \ca B(\ca D_W,\ca H)$, with $n\in \bb N \cup \{0\}$, such that the series $\sum_{n\ge 0} \frac{1}{n!}A_n(t_0) (t-t_0)^n$ converges to $A(t)$ in the operator norm $|\!|\cdot|\!|_{\ca D_W\to\ca H}$, in a neighborhood of $t_0$.}
In particular, $t \mapsto \wt W(t)$ is twice continuously differentiable as a map from $\bb R$ into $\ca B(\ca D_W, \ca H)$,\footnote{An operator-valued curve $t \mapsto A(t)$ is continuous as a map from $\bb R$ into $\ca B(\ca D_W, \ca H)$ if, for each $t_0\in \bb R$, $\lim_{t\to t_0}|\!|A(t) - A(t_0)|\!|_{\ca D_W\to\ca H} = 0$. The curve is first differentiable at $t_0$ if there exists an operator $\dot{A}(t_0) \in \ca B(\ca D_W,\ca H)$ such that
\[
\lim_{t\to t_0}\left|\!\left| \frac{A(t) - A(t_0)}{t - t_0} - \dot{A}(t_0) \right|\!\right|_{\ca D_W\to\ca H} = 0\,.
\]
The curve is first continuously differentiable as a map from $\bb R$ to $\ca B(\ca D_W,\ca H)$ if it is first differentiable at every $t_0$ and the curve $t \mapsto \dot{A}(t)$ is continuous. The curve $t \mapsto A(t)$ is twice continuously differentiable if $t \mapsto \dot{A}(t)$ is first differentiable at every $t$ and $t \mapsto \ddot{A}(t)$ is continuous.}
and given any $T\in\bb R^+$, $\wt W(t)$, $\dot{\wt W}(t)$ and $\ddot{\wt W}(t)$ are uniformly bounded for $0\le t\le T$.
\end{lemma}
\begin{proof}
Since $H$ is self-adjoint on $\ca D$ and $B$ is self-adjoint and bounded on $\ca H$, it follows from \emph{Kato-Rellich theorem} that $W = H - \mu B$ is self-adjoint on $\ca D$.
As $H$ is bounded from below, and $B$ is bounded, $W$ is bounded from below.
Since $\wt W(t) = S(t)WS^\dag(t)$ and $S(t)$ is unitary, $\wt W(t)$ is self-adjoint on the domain $S(t)\ca D$. Moreover, as the spectrum of $\wt W(t)$ is the same as that of $W$, it is bounded from below by a $t$-independent value.
We want to show that the domain of $\wt W(t)$ is actually $\ca D$, so in what follows we prove that $S(t)\ca D = \ca D$. 

As $W$ is self-adjoint on $\ca D$, its graph $\{(\psi, W\psi), \psi \in \ca D\}$ is closed in $\ca H \oplus \ca H$. Equivalently, $\ca D_W$ is closed with respect to the graph norm $|\cdot|_W$, and thus $\ca D_W$ is a Banach space.
Let $\un B$ denote the restriction of $B$ to $\ca D$. Because $B\ca D \subset \ca D$, $\un B$ can be regarded as a map from $\ca D_W$ into $\ca D_W$,
\be
\un B : \ca D_W \to \ca D_W\,.
\ee
We show next that $\un B \in \ca B(\ca D_W)$, i.e., $\un B$ is bounded on $\ca D_W$. Let $\psi_n$ be a sequence in $\ca D_W$ such that 
\ba
\psi_n &\to \psi \\
\un B \psi_n &\to \phi
\ea
converge in $\ca D_W$. Since $|\cdot|_W$ is stronger than the Hilbert space norm $|\cdot|$, both convergences also hold in $\ca H$. In particular, as $B$ is bounded on $\ca H$, $\psi_n \to \psi$ implies that $B\psi_n \to B\psi$. Since $\un B = B$ on $\ca D_W$, it follows that $\un B\psi_n \to \un B\psi$.
By uniqueness of the limit, $\phi = \un B\psi$. Therefore, the graph of $\un B$ is closed and, by the \emph{closed graph theorem}, $\un B$ is bounded on $\ca D_W$.
Now consider the sequence of partial sums
\be
\un S_N(t) = \sum_{n=0}^N \frac{1}{n!}\left(\frac{i\mu t}{\hbar}\right)^n \!\un B^n\,.
\ee
For each $N$, $\un S_N(t) \in \ca B(\ca D_W)$. Moreover, 
\be
\sum_{n=0}^\infty \Big|\!\Big| \frac{1}{n!}\left(\frac{i\mu t}{\hbar}\right)^n \!\un B^n \Big|\!\Big|_W \le e^{\mu|t| |\!|\un B|\!|_W/\hbar} < \infty
\ee
so the sequence $\un S_N(t)$ is Cauchy, and because $\ca B(\ca D_W)$ is Banach, it converges to a bounded operator $\un S(t) \in \ca B(\ca D_W)$.
But since $B$ is also bounded on $\ca H$, the partial sums
\be
S_N(t) = \sum_{n=0}^N \frac{1}{n!}\left(\frac{i\mu t}{\hbar}\right)^n \! B^n
\ee
converge on $\ca B(\ca H)$ to $S(t)$. For any fixed $\psi \in \ca D$, the sequences $\un S_N(t)\psi$ and $S_N(t)\psi$ are equal and thus must converge to the same vector, which implies
\be
\un S(t)\psi = S(t)\psi\,.
\ee
Since $\un S(t)\psi \in \ca D$, we conclude that $S(t)\psi \in \ca D$. As this is true for any $\psi \in \ca D$, we obtain
\be\label{StcaDsubsetcaD}
S(t)\ca D \subset \ca D\,.
\ee
In fact, this is true for any $t$, so we also have $S(-t)\ca D \subset \ca D$. Since $S(t)S(-t) = 1$, we have $\ca D \subset S(t)\ca D$, which in combination with \eqref{StcaDsubsetcaD} gives
\be
S(t)\ca D = \ca D\,.
\ee
Thus, as we wished to show, $\wt W(t)$ is self-adjoint on $\ca D$.

The argument above also reveals that $S(t)$ is real-analytic as a map from $\bb R$ into $\ca B(\ca H)$ and $\un S(t)$ is real-analytic as a map from $\bb R$ into $\ca B(\ca D_W)$. Moreover since $\un S(t)\psi = S(t)\psi$ for every $\psi$ in the domain of $\wt W(t)$, we can write
\be
\wt W(t) = S(t)WS^\dag(t) = S(t)WS(-t) = S(t)W\un S(-t)\,.
\ee
On the right-hand side, we have the product of real-analytic bounded maps 
\be
\ca B(\ca H) \cdot \ca B(\ca D_W,\ca H) \cdot \ca B(\ca D_W) \to \ca B(\ca D_W, \ca H)
\ee
and therefore $\wt W(t)$ is a real-analytic map from $\bb R$ into $\ca B(\ca D_W, \ca H)$.
In particular, $\wt W(t)$ is infinitely continuously differentiable as a map from $\bb R$ into $\ca B(\ca D_W, \ca H)$. Restricting to any compact interval of time, say $0\le t \le T$, $\wt W(t)$ and any finite number of its derivatives are uniformly bounded.
\end{proof}

Let us define a few quantities relevant for stating the adiabatic theorem. 
Let 
\be
U(t) := e^{-iHt/\hbar}
\ee
be the unitary operator implementing the dynamical evolution of $H$, and let 
\be
\wt U(t) = S(t) U(t)
\ee
be the ``interaction picture'' unitary evolution operator, which satisfies the dynamical equation
\be
i\hbar \dot{\wt U}(t) = \wt W(t) \wt U(t)\,.
\ee
Let $P$ be the projector to $\ca V$ and $Q = 1 - P$ the complementary projector. Since $\wt W(t)$ evolves unitarily, its spectrum is time-independent. In particular, the subset $I$ remains part of the spectrum of $\wt W(t)$, and it remains separated from the rest of the spectrum by the gap $\Delta$ at all times. Moreover, the projector to the vector subspace associated with $I$ in the spectral decomposition of $\wt W(t)$ is given by
\be
\wt P(t) = S(t) P S^\dag(t)\,.
\ee
Similarly, $\wt Q(t) = S(t) Q S^\dag(t)$.
Now define the ``deformed evolution operator'' $\wt U_*(t)$ by the equation
\be
i\hbar \dot{\wt U}_*(t) = \left(\wt W(t) + i\hbar [\dot{\wt P}(t), \wt P(t)]\right) \wt U_*(t)
\ee
with initial condition $\wt U_*(0) = 1$. 
Finally, define the time-dependent operator
\be\label{frFdef}
\F(t) := \frac{\hbar}{2\pi i} \oint_{\ca C}\!dz\,  \wt Q(t) R(z;t) \dot R(z;t) + \text{adj.}
\ee
where ``$\text{adj.}$'' denotes the adjoint of the first term, $\ca C$ is a contour in $\bb C \backslash \text{Spec}(W)$ encircling $I$ but no part of $\text{Spec}(W) \backslash I$, and $R(z;t)$ is the resolvent operator associated with $\wt W(t)$,
\be
R(z;t) := \big(\wt W(t) - z\big)^{-1}\,.
\ee
With this, we are ready to restate (a slightly particularized version of) Theorem~2.2 of \cite{teufel2003adiabatic} in a language directly applicable to our context:
\begin{theorem}[Kato's adiabatic inequality]\label{TheoremAdiabatic}
Let $H$, $B$ and $W$ satisfy assumptions (i) and (ii) above, and consequently $\wt W(t)$ have the properties established in Lemma~\ref{assumpiii}. Then, for $0\le t \le T$, with any $T\in \bb R^+$, the deformed evolution operator $\wt U_*(t)$ satisfies the exact intertwining relation
\be\label{intertprop}
\wt U_*(t)P = \wt P(t) \wt U_*(t)
\ee
and moreover we have the following bound
\be\label{Katobound}
\big|\!\big| \big(\wt U(t) - \wt U_*(t)\big) P \big|\!\big| \le \Theta(t)
\ee
with
\be
\Theta(t) := |\!| \F(0) |\!| + |\!| \F(t) |\!| + \int_0^t\!d\tau\left( |\!| \dot \F(\tau) |\!| + |\!| \F(\tau) [\dot{\wt P}(\tau), \wt P(\tau)] |\!| \right)\,.
\ee
\end{theorem}

From this theorem, it immediately follows:
\begin{corollary}\label{CorollaryAdiabatic}
Particularizing further to our context, we have also
\be\label{QUP}
|\!| Q U(t) P |\!| \le \Theta(t)
\ee
and, with $P(t) = U^\dag(t)PU(t)$,
\be\label{PtP0bound}
|\!| P(t) - P(0) |\!| \le \varepsilon_1 + \varepsilon_2 t
\ee
for constants $\varepsilon_1$ and $\varepsilon_2$ given by
\be\label{varepsdef}
\varepsilon_1 := 4 |\!|\F(0)|\!| \quad\text{and}\quad
\varepsilon_2 := 2|\!| \dot \F(0) |\!| + 2|\!| \F(0) [\dot{\wt P}(0), \wt P(0)] |\!|\,.
\ee
\end{corollary}
\begin{proof}
To obtain \eqref{QUP}, first notice using the intertwining property \eqref{intertprop} that
\be
\wt Q(t) \wt U_*(t) P = \wt Q(t) \wt P(t) \wt U_*(t) = 0
\ee
so
\be
|\!|\wt Q(t) \wt U(t) P |\!| = \big|\!\big| \wt Q(t)\left(\wt U(t) - \wt U_*(t)\right) P \big|\!\big| \le |\!|\wt Q(t)|\!| \big|\!\big| \left(\wt U(t) - \wt U_*(t)\right) P \big|\!\big| \le \Theta(t)
\ee
where we used that $|\!|\wt Q(t)|\!| \le 1$ and the bound \eqref{Katobound}. 
Finally, we have
\be
|\!|Q U(t) P |\!| = |\!|S(t)QS^\dag(t) S(t)U(t) P |\!| = |\!|\wt Q(t) \wt U(t) P |\!| \le \Theta(t)
\ee
where we used that multiplication from the right or left by a unitary does not affect the operator norm.

Before showing \eqref{PtP0bound}, let us prove that the following auxiliary inequality holds
\be\label{wtUPwtUi}
|\!|\wt U(t) P \wt U^\dag(t) - \wt P (t) |\!| \le 2 \Theta(t)\,.
\ee
Notice that
\be
\wt U(t) P \wt U^\dag(t) - \wt U_*(t) P \wt U_*^\dag(t) = \left(\wt U(t) - \wt U_*(t)\right) P \wt U^\dag(t) + \wt U_*(t) P \left(\wt U^\dag(t) -\wt U_*^\dag(t)\right)
\ee
and consequently
\be
|\!|\wt U(t) P \wt U^\dag(t) - \wt U_*(t) P \wt U_*^\dag(t) |\!| \le \big|\!\big| \left(\wt U(t) - \wt U_*(t)\right) P \big|\!\big| + \big|\!\big| P \left(\wt U^\dag(t) -\wt U_*^\dag(t)\right) \big|\!\big| \le 2\Theta(t)
\ee
where we have used the triangle inequality and the self-adjointness of $P$. Then, using the intertwining property \eqref{intertprop}, we see that $\wt U_*(t) P \wt U_*^\dag(t) = \wt P(t) \wt U_*(t) \wt U_*^\dag(t) = \wt P(t)$, establishing \eqref{wtUPwtUi}. 

Finally, to obtain \eqref{PtP0bound}, notice that
\ba
|\!| P(t) - P(0) |\!| &= |\!| U^\dag(t)P U(t) - P |\!| \no
&= |\!| U(t)P U^\dag(t) - P |\!| \no
&= \big|\!\big| S^\dag(t) \left( \wt U(t) P \wt U^\dag(t) - S(t) P S^\dag(t) \right) S(t) \big|\!\big| \no
&= \big|\!\big| \wt U(t) P \wt U^\dag(t) - \wt P(t) \big|\!\big| \no
&\le 2 \Theta(t)\,.
\label{PtP02Theta}
\ea
The last thing to show is that $\Theta(t)$ is a first-order polynomial in time. This follows from the fact that all time evolutions involved are unitary, and therefore all norms appearing in $\Theta$ are preserved in the evolution. In particular, $\F(t) = S(t) \F(0) S^\dag(t)$, which implies that 
\be
|\!|\F(t)|\!| + |\!|\F(0)|\!| = 2|\!|\F(0)|\!|\,.
\ee
Also,
\be
\int_0^t\!d\tau\left( |\!| \dot \F(\tau) |\!| + |\!| \F(\tau) [\dot{\wt P}(\tau), \wt P(\tau)] |\!| \right) = \left(|\!| \dot \F(0) |\!| + |\!| \F(0) [\dot{\wt P}(0), \wt P(0)] |\!|\right) t\,,
\ee
which gives the linear-in-time contribution.
Combining \eqref{PtP02Theta} with this form of $\Theta$ gives \eqref{PtP0bound}, with $\varepsilon_1$ and $\varepsilon_2$ defined in \eqref{varepsdef}.
\end{proof}

Our next goal is to find simpler bounds for $\varepsilon_1$ and $\varepsilon_2$, in terms of quantities more directly related to the $B$ field, so we can estimate their values. 
We begin by deriving an inequality for a Sylvester equation, which will be useful for all subsequent results. (The first part of the lemma below follows trivially from Lemma~4.4 of \cite{albeverio2011operator}, while the second part establishes an inequality that is known in the case of bounded operators \cite{bhatia1997and} but we generalize to the unbounded case.) 

\begin{lemma}\label{lemmaSylvester}
Let $A$ be a densely defined, closed (possibly unbounded) self-adjoint operator on a Hilbert space $\ca L$, and $B$ be a bounded 
self-adjoint operator on a Hilbert space $\ca K$. Assume that their spectra are contained in 
\be
\text{\emph{Spec}}(A)\subset (-\infty, r - \Delta]\cup [s+\Delta, \infty) \quad\text{and}\quad \text{\emph{Spec}}(B)\subset [r,s]
\ee
with $s>r$ and $\Delta > 0$.
If $Y$ is any bounded operator from $\ca K$ to $\ca L$, then the Sylvester equation
\be\label{sylvestereq}
AX - XB = -Y
\ee
has a unique operator solution $X \in \ca B(\ca K, \ca L)$ given by 
\be\label{Sylvestersol}
X = \frac{1}{2\pi i} \oint_{\ca C}\!dz\, (A - z)^{-1}Y(B - z)^{-1}
\ee
where $\ca C$ is a contour in $\bb C$, disjoint from $\text{\emph{Spec}}(A) \cup \text{\emph{Spec}}(B)$, encircling  $\text{\emph{Spec}}(B)$ but no part of $\text{\emph{Spec}}(A)$.
Moreover,
\be\label{Sylvesterineq}
{|\!| X |\!|}_{\ca K \to \ca L} \le \frac{1}{\Delta} {|\!| Y |\!|}_{\ca K \to \ca L}
\ee
\end{lemma}
\begin{proof}
Taking the adjoint of
\eqref{sylvestereq} gives
\be
X^\dag A^\dag - B^\dag X^\dag = -Y^\dag\,.
\ee
It immediately follows from Lemma~4.4 of \cite{albeverio2011operator} that this equation has a unique operator solution given by
\be
X^\dag = \frac{1}{2\pi i} \oint_{\ca C}\!dz\, (B^\dag - z)^{-1}(-Y^\dag)(A^\dag - z)^{-1}
\ee
where $X^\dag \in \ca B(\ca L, \ca K)$.
Taking the adjoint again, we obtain \eqref{Sylvestersol} as the unique operator solution of \eqref{sylvestereq}.\footnote{It is contained in this result that $\text{Ran}(X) \subset \text{Dom}(A)$.}
To prove the inequality \eqref{Sylvesterineq}, it is useful to shift the operators so that $[r,s]$ becomes symmetrical about the origin. That is, define
\be
A_c := A - c \quad\text{and}\quad B_c := B - c
\ee
with $c := (r+s)/2$. Notice that  
\be
\text{Spec}(A_c)\subset (-\infty, -a]\cup [a, \infty) \quad\text{and}\quad \text{Spec}(B_c)\subset [-b,b]
\ee
where $b := (s-r)/2$ and $a := b + \Delta$. Moreover, $X$ is solution of \eqref{sylvestereq} if and only if it is also solution of 
\be
A_c X - X B_c = - Y\,.
\ee
Note that for any $\psi \in \text{Dom}(A_c) =\text{Dom}(A) \subset \ca L$, $|A_c\psi |_{\ca L} \ge a|\psi|_{\ca L}$.
If $\phi \in \ca K$,
\be
|Y\phi|_{\ca L} = |A_cX\phi - XB_c\phi|_{\ca L} \ge |A_cX\phi|_{\ca L} - |XB_c\phi|_{\ca L} \ge a|X\phi|_{\ca L} - |XB_c\phi|_{\ca L}\,.
\ee
Restricting $\phi$ to be a unit vector, we have
\be
|XB_c\phi|_{\ca L} \le {|\!|X|\!|}_{\ca K\to\ca L} |B_c\phi|_{\ca K} \le b {|\!|X|\!|}_{\ca K\to\ca L}
\ee
so $|Y\phi|_{\ca L}  \ge a|X\phi|_{\ca L} - b {|\!|X|\!|}_{\ca K\to\ca L}$. 
Considering a sequence of unit vectors $\phi_n \in \ca K$ such that $|X\phi_n|_{\ca L}\to {|\!|X|\!|}_{\ca K\to\ca L}$, we obtain
\be
|\!|Y|\!|_{\ca K\to\ca L} \ge \sup_n |Y\phi_n|_{\ca L} \ge a {|\!|X|\!|}_{\ca K\to\ca L} - b {|\!|X|\!|}_{\ca K\to\ca L} = \Delta {|\!|X|\!|}_{\ca K\to\ca L} \,, 
\ee
proving the inequality.
\end{proof}

We now turn to the estimation of $\varepsilon_1$. 
First we establish that $\F(t)$ is block off-diagonal with respect to $\wt P(t)$ and $\wt Q(t)$.
Define,
\be
\F_\text{QP}(t) := \frac{\hbar}{2\pi i} \oint_{\ca C}\!dz\,  \wt Q(t) R(z;t) \dot R(z;t) \wt P(t)\,.
\ee
This operator only connects the $\wt P(t)$ subspace to the $\wt Q(t)$ subspace, i.e., $\F_\text{QP}(t) = \wt Q(t) \F_\text{QP}(t) \wt P(t)$.
Notice that
\ba
\F_\text{QP}(t) &= \frac{\hbar}{2\pi i} \oint_{\ca C}\!dz\,  \wt Q(t) R(z;t) \dot R(z;t) - \frac{\hbar}{2\pi i} \oint_{\ca C}\!dz\,  \wt Q(t) R(z;t) \dot R(z;t) \wt Q(t) \no
&= \frac{\hbar}{2\pi i} \oint_{\ca C}\!dz\,  \wt Q(t) R(z;t) \dot R(z;t) + \frac{\hbar}{2\pi i} \oint_{\ca C}\!dz\,  \left(\wt Q(t) R(z;t)^2 \wt Q(t)\right)\dot{\wt W}(t) \left( \wt Q(t) R(z;t) \wt Q(t)\right) \no
&= \frac{\hbar}{2\pi i} \oint_{\ca C}\!dz\,  \wt Q(t) R(z;t) \dot R(z;t)
\ea
where in the first line we used $\wt P(t) = 1 - \wt Q(t)$; in second line we used that
\be
\dot R(z;t) = - R(z;t) \dot{\wt W}(t) R(z;t)
\ee
and that $\wt Q(t) = \wt Q(t)^2$ commutes with $R(z;t)$; in the third line we used that the integrand of the second integral is analytic inside $\ca C$.
As $\F_\text{QP}^\dag(t) = \wt P(t) \F^\dag(t) \wt Q(t)$, we confirm that
\be\label{Fblockoff}
\F(t) = \F_\text{QP}(t) + \F_\text{QP}^\dag(t)
\ee
is block off-diagonal with respect to $\wt P(t)$ and $\wt Q(t)$.
Next, we show:
\begin{lemma}\label{lemmaFQPsylvester}
The operator $\F_\text{QP}(t)$, understood as an operator from $\ca H_P(t) := \wt P(t)\ca H$ to  $\ca H_Q(t) := \wt Q(t)\ca H$, satisfies the Sylvester equation
\be\label{FQPsylvester}
\wt W_Q(t) \F_\text{QP}(t) -  \F_\text{QP}(t)\wt W_P(t) = - \hbar \wt Q(t) \dot{\wt P}(t) \wt P(t)
\ee
where
\be\label{WQWQdef}
\wt W_Q(t) := \wt Q(t) \wt W(t) \wt Q(t) \quad\text{and}\quad \wt W_P(t) := \wt P(t) \wt W(t) \wt P(t)\,.
\ee
are understood as operators on $\ca H_Q(t)$ and $\ca H_P(t)$, respectively.

Moreover, $\F(t)$ satisfies the inequality
\be\label{FdotPineq}
|\!|\F(t)|\!|  \le \frac{\hbar}{\Delta} |\!|\dot{\wt P}(t)|\!|\,.
\ee
\end{lemma}
\begin{proof}
To show that $\F_\text{QP}(t)$ satisfies the Sylvester equation above, we consider the two terms on the left-hand side. (In this proof, we omit the time dependence in all functions.) First, note the following identity
\ba
\wt W_Q \wt Q R(z) &= \wt Q \wt W \wt Q R(z) \no
&= \wt Q \big(R(z)^{-1} + z\big) \wt Q R(z) \no
&= \wt Q + z \wt Q R(z)
\ea
where in the first line we used $\wt Q^2 = \wt Q$; in the second line we used $\wt W = R(z)^{-1} + z$; and in the third line we used that $\wt Q$ commutes with $R(z)$. Then, the first term in the Sylvester in the left-hand side of the Sylvester equation gives
\ba
\wt W_Q \F_\text{QP} &= \frac{\hbar}{2\pi i} \oint_{\ca C}\!dz\,  \wt W_Q \wt Q R(z) \dot R(z) \wt P \no
&= \frac{\hbar}{2\pi i} \oint_{\ca C}\!dz\,  \wt Q \dot R(z) \wt P + \frac{\hbar}{2\pi i} \oint_{\ca C}\!dz\,  z\wt Q R \dot R(z) \wt P\,.
\ea
Next, note that
\be\label{wtPwtWPrel}
\wt P\wt W_P = R(z)^{-1}\wt P + z \wt P
\ee
so the second term
in the left-hand side of the Sylvester equation gives
\ba
\wt F_{QP} \wt W_P &= \frac{\hbar}{2\pi i} \oint_{\ca C}\!dz\,  \wt Q R(z) \dot R(z) \wt P \wt W_P \no
&= \frac{\hbar}{2\pi i} \oint_{\ca C}\!dz\,  \wt Q R(z) \dot R(z) \wt P R(z)^{-1} \wt P + \frac{\hbar}{2\pi i} \oint_{\ca C}\!dz\,  z\wt Q R(z) \dot R(z) \wt P \no
&= -\frac{\hbar}{2\pi i} \oint_{\ca C}\!dz\,  \wt Q R(z)^2 \dot{\wt W} R(z)\wt P R(z)^{-1} \wt P + \frac{\hbar}{2\pi i} \oint_{\ca C}\!dz\,  z\wt Q R(z) \dot R(z) \wt P \no
&= -\frac{\hbar}{2\pi i} \oint_{\ca C}\!dz\,  \big(\wt Q R(z)^2 \wt Q\big) \dot{\wt W} \wt P + \frac{\hbar}{2\pi i} \oint_{\ca C}\!dz\,  z\wt Q R(z) \dot R(z) \wt P \no
&= \frac{\hbar}{2\pi i} \oint_{\ca C}\!dz\,  z\wt Q R \dot R(z) \wt P
\ea
where in the second line we used relation \eqref{wtPwtWPrel}; in the third line $\dot R(z) = - R(z) \dot{\wt W} R(z)$; in the fourth line the commutativity of $\wt P$ and $\wt Q$ with $R(z)$; and in the last line the fact that the integrand $\big(\wt Q R(z)^2 \wt Q\big) \dot{\wt W} \wt P$ is analytic in the interior of the contour.
We thus see that
\ba
\wt W_Q \F_\text{QP} -  \F_\text{QP}\wt W_P &= -\hbar \wt Q \frac{d}{dt}\left[-\frac{1}{2\pi i} \oint_{\ca C}\!dz\, R(z)\right] \wt P \no
&= - \hbar \wt Q \dot{\wt P} \wt P
\ea
where the quantity inside the brackets gives the Riesz realization of $\wt P$ (see footnote~\ref{Rieszfn}).
This establishes \eqref{FQPsylvester}.

Lemma~\ref{lemmaSylvester} can be automatically applied by identifying $\wt W_Q$ with $A$ (with $\ca L = \ca H_Q$), $\wt W_P$ with $B$ (with $\ca K=\ca H_P$) and $\hbar \wt Q \dot{\wt P} \wt P$ with $Y$, and observing that the spectra of $A$ and $B$ coincide respectively with $\text{Spec}(W) \backslash I$ and $I$. We obtain
\be
{|\!| \F_{QP} |\!|}_{\ca H_P \to \ca H_Q} \le \frac{1}{\Delta}{|\!| \hbar \wt Q \dot{\wt P} \wt P |\!|}_{\ca H_P \to \ca H_Q} = \frac{\hbar}{\Delta}{|\!| \wt Q \dot{\wt P} \wt P |\!|}_{\ca H_P \to \ca H_Q}\,.
\ee
Since $\F$ is block off-diagonal, according to \eqref{Fblockoff}, with the off-diagonal blocks equal to $\F_{QP}$ and $\F_{QP}^\dag$, we have
\be
|\!|\F|\!| = {|\!|\F_{QP}|\!|}_{\ca H_P \to \ca H_Q}\,.
\ee
Since $\dot{\wt P}$ is also block off-diagonal with respect to $\wt P$ and $\wt Q$,\footnote{Since $\wt P \wt Q = 0$, we have $\dot{\wt P} \wt Q = - \wt P \dot{\wt Q}$, so $\wt Q \dot{\wt P} \wt Q = - \wt Q \wt P \dot{\wt Q} = 0$. Since $\wt P + \wt Q = 1$, we have $\dot{\wt P} = - \dot{\wt Q}$, so $\wt P \dot{\wt P} \wt P = - \wt P \dot{\wt Q} \wt P = \dot{\wt P} \wt Q \wt P = 0$.\label{footdotPblock}} we have
\be\label{dotPQPP}
|\!|\dot{\wt P}|\!| = {|\!|\wt Q \dot{\wt P} \wt P|\!|}_{\ca H_P \to \ca H_Q}\,.
\ee
The inequality \eqref{FdotPineq} then follows.
\end{proof}

The next step is to derive a bound for $\dot{\wt P}$ directly in terms of $B$ and $H$. To that end, we establish one more intermediate result:
\begin{lemma}\label{lemmaQdotPPsylvester}
The operator $\wt Q(t) \dot{\wt P}(t) \wt P(t)$, from $\ca H_P(t)$ to $\ca H_Q(t)$, satisfies the Sylvester equation
\be
\wt W_Q(t) \big(\wt Q(t) \dot{\wt P}(t) \wt P(t)\big) -  \big(\wt Q(t) \dot{\wt P}(t) \wt P(t)\big)\wt W_P(t) = - \wt Q(t) \dot{\wt W}(t) \wt P(t)
\ee
with $\wt W_Q$ and $\wt W_P$ defined in \eqref{WQWQdef}.
Consequently, $\dot{\wt P}(t)$ satisfies the inequality
\be\label{dotPQdotWPineq}
|\!|\dot{\wt P}(t)|\!|  \le \frac{1}{\Delta} |\!|\wt Q(t) \dot{\wt W}(t) \wt P(t)|\!|\,.
\ee
\end{lemma}
\begin{proof}
The proof is very similar to that of Lemma \ref{lemmaFQPsylvester}. (Again, we omit the time dependence in all functions.) 
First, note that from the Riesz realization of $\wt P$,
\be
\wt P = - \frac{1}{2\pi i} \oint_{\ca C}\!dz\, R(z)\,,
\ee
and the identity $\dot R(z) = - R(z) \dot{\wt W} R(z)$, we have
\be
\dot{\wt P} = \frac{1}{2\pi i} \oint_{\ca C}\!dz\, R(z) \dot{\wt W} R(z)\,.
\ee
Now consider
\ba
\wt W_Q \big(\wt Q \dot{\wt P} \wt P\big) &= \frac{1}{2\pi i}  \oint_{\ca C}\!dz\, \wt Q \wt W R(z) \dot{\wt W} R(z) \wt P \no
&= \frac{1}{2\pi i}  \oint_{\ca C}\!dz\, \wt Q \big(R(z)^{-1} + z \big) R(z) \dot{\wt W} R(z) \wt P \no
&= \frac{1}{2\pi i}  \oint_{\ca C}\!dz\, \wt Q \dot{\wt W} R(z) \wt P + \frac{1}{2\pi i}  \oint_{\ca C}\!dz\, z\wt Q R(z) \dot{\wt W} R(z) \wt P
\ea
and
\ba
\big(\wt Q \dot{\wt P} \wt P\big) \wt W_P &= \frac{1}{2\pi i}  \oint_{\ca C}\!dz\, \wt Q R(z) \dot{\wt W} R(z) \wt W \wt P \no
&= \frac{1}{2\pi i}  \oint_{\ca C}\!dz\, \wt Q R(z) \dot{\wt W} R(z) \big(R(z)^{-1} + z \big) \wt P \no
&= \frac{1}{2\pi i}  \oint_{\ca C}\!dz\, \wt Q R(z) \dot{\wt W} \wt P + \frac{1}{2\pi i}  \oint_{\ca C}\!dz\, z\wt Q R(z) \dot{\wt W} R(z) \wt P\,.
\ea
Therefore, using again the Riesz realization of $\wt P$,
\ba
\wt W_Q \big(\wt Q \dot{\wt P} \wt P\big) - \big(\wt Q \dot{\wt P} \wt P\big) \wt W_P &= \frac{1}{2\pi i}  \oint_{\ca C}\!dz\, \wt Q \dot{\wt W} R(z) \wt P - \frac{1}{2\pi i}  \oint_{\ca C}\!dz\, \wt Q R(z) \dot{\wt W} \wt P \no
&= \wt Q \dot{\wt W} ( - \wt P) \wt P - \wt Q ( - \wt P) \dot{\wt W} \wt P \no
&= -  \wt Q \dot{\wt W} \wt P\,.
\ea
Applying Lemma~\ref{lemmaSylvester}
then gives
\be
{|\!|\wt Q\dot{\wt P} \wt P|\!|}_{\ca H_P \to \ca H_Q}  \le \frac{1}{\Delta} {|\!|\wt Q \dot{\wt W} \wt P|\!|}_{\ca H_P \to \ca H_Q}\,.
\ee
Using that ${|\!|\wt Q\dot{\wt P} \wt P|\!|}_{\ca H_P \to \ca H_Q} = |\!|\dot{\wt P}|\!|$, as established in \eqref{dotPQPP}, and ${|\!|\wt Q \dot{\wt W} \wt P|\!|}_{\ca H_P \to \ca H_Q} = |\!|\wt Q \dot{\wt W} \wt P|\!|$, yields the inequality \eqref{dotPQdotWPineq}.
\end{proof}
With this, we can finally state the main, useful bound for $\varepsilon_1$,
\begin{theorem}\label{theoremvareps1bound}
Under assumptions (i) and (ii), the constant $\varepsilon_1$ is bounded by
\be
\varepsilon_1 \le  \frac{4\mu}{\Delta^2} |\!| Q[B,H] P|\!|\,.
\ee
\end{theorem}
\begin{proof}
Using that $|\!|\F(t)|\!|$ is time independent,  Lemma~\ref{lemmaFQPsylvester} and Lemma~\ref{lemmaQdotPPsylvester}, we obtain
\be\label{vareps1boundinter}
\varepsilon_1 := 4|\!|\F(0)|\!| = 4|\!|\F(t)|\!| \le \frac{4\hbar}{\Delta} |\!|\dot{\wt P}(t)|\!| \le \frac{4\hbar}{\Delta^2} |\!|\wt Q(t)\dot{\wt W}(t) \wt P(t)|\!|\,.
\ee
From the definition of $\wt W(t)$, we have
\be
\dot{\wt W}(t) = i\frac{\mu}{\hbar}S(t)[B, W] S^\dag(t) = i\frac{\mu}{\hbar}S(t)[B, H] S^\dag(t)\,.
\ee
In particular, note that $B\ca D \subset \ca D$ from assumption (i), together with $S(t)\ca D = \ca D$ from Lemma~\ref{assumpiii}, ensure that the right-hand side is well-defined on $\ca D$.
Finally, from the unitarity of $S(t)$, 
\be\label{QdotWPQBHP}
|\!|\wt Q(t)\dot{\wt W}(t) \wt P(t)|\!| =  \frac{\mu}{\hbar} |\!|S(t) Q [B, H] P S^\dag(t)|\!| =  \frac{\mu}{\hbar} |\!| Q [B, H] P |\!|
\ee
Note that the left-hand side is bounded on $\ca H$, since $\wt P(t)$ is bounded as a map from $\ca H$ to $\ca D_W$ and $\dot{\wt W}(t) \in \ca B(\ca D_W, \ca H)$ from Lemma~\ref{assumpiii}, so the right-hand side is also bounded on $\ca H$.
In conjunction with \eqref{vareps1boundinter}, the result is proven.
\end{proof}

We now turn to the estimation of $\varepsilon_2$. 
It is convenient to define the scale
\be
\varepsilon := \frac{2\mu}{\Delta^2} |\!| Q[B,H] P|\!|
\ee
for purposes of comparing with the bound for $\varepsilon_1$ (i.e., $\varepsilon_1 \le 2 \varepsilon$).
We start with:
\begin{lemma}\label{lemmaFdotPP}
The operator $\F(t) [\dot{\wt P}(t), \wt P(t)]$ satisfies the bound
\be
|\!|\F(t) [\dot{\wt P}(t), \wt P(t)] |\!| \le \frac{\Delta}{4\hbar} \varepsilon^2\,.
\ee
\end{lemma}
\begin{proof}
Using the fact that $\dot{\wt P}$ is block off-diagonal with respect to $\wt P$ and $\wt Q$, as shown in footnote~\ref{footdotPblock}, we have
\be
|\!|[\dot{\wt P}, \wt P] |\!| = |\!|\dot{\wt P} \wt P - \wt P \dot{\wt P} |\!| = |\!|\wt Q\dot{\wt P} \wt P - \wt P \dot{\wt P} \wt Q |\!| = |\!|\wt Q\dot{\wt P} \wt P |\!| = |\!|\dot{\wt P} |\!| \,.
\ee
Thus, from Lemma~\ref{lemmaFQPsylvester},
\be
|\!|\F [\dot{\wt P}, \wt P] |\!| \le |\!|\F |\!| |\!|[\dot{\wt P}, \wt P] |\!| \le \frac{\hbar}{\Delta} |\!|\dot{\wt P} |\!|^2\,.
\ee
Using Lemma~\ref{lemmaQdotPPsylvester}, together with formula \eqref{QdotWPQBHP}, gives
\be
|\!|\F [\dot{\wt P}, \wt P] |\!| \le \frac{\hbar}{\Delta} \left(\frac{\mu}{\hbar \Delta} |\!|Q[B,H]P|\!| \right)^2 = \frac{\mu^2}{\hbar \Delta^3} |\!|Q[B,H]P|\!|^2
\ee
concluding the derivation.
\end{proof}
We now estimate the term $\dot\F$. This term is more complicated to analyze than the ones before, and for this reason we will only derive a relatively crude bound, which nonetheless is still sufficient for our purposes.
The result is as follows:
\begin{lemma}\label{lemmadotF}
The operator $\dot \F(0)$ satisfies, for any allowed contour $\ca C$, the bound
\be
|\!| \dot\F(0) |\!| \le \frac{\hbar|\ca C|}{\pi} \kappa^2 \left( |\!|\dot{\wt P}(0)|\!| \alpha_1 + 3\alpha_1^2 + \alpha_2 \right)
\ee
where
\ba
\kappa &:= \sup_{z\in \ca C} |\!|R(z;0)|\!| \\
\alpha_1 &:= \sup_{z\in \ca C} |\!|\dot{\wt W}(0)R(z;0)|\!| \\
\alpha_2 &:= \sup_{z\in \ca C} |\!|\ddot{\wt W}(0)R(z;0)|\!|
\ea
and $|\ca C|$ is the length of the contour $\ca C$.
\end{lemma}
\begin{proof}
Assumptions (i) and (ii), together with Lemma~\ref{assumpiii}, ensure that $\kappa$, $\alpha_1$ and $\alpha_2$ are well-defined (i.e., finite).\footnote{Since $\wt W(t) \in \ca B(\ca D_W, \ca H)$, then $\wt W(t) - z$ is a bijective bounded map between the Banach spaces $\ca D_W$ and $\ca H$, for $z\notin \text{Spec}(W)$. The \emph{bounded inverse theorem} implies that $R(z;t) \in \ca B(\ca H, \ca D_W)$. In particular, $R(z;t)$ is bounded on $\ca H$. Since $\ca C$ is compact and separated by a non-zero gap from $\text{Spec}(W)$, and $R(z;t)$ is continuous in $z$, $\sup_{z\in\ca C}|\!|R(z,t)|\!| < \infty$. Moreover, as $\dot{\wt W}(t) \in \ca B(\ca D_W, \ca H)$, the composition $\dot{\wt W}(t)R(z,t)$ is a bounded map from $\ca H$ to $\ca H$. Again, from the compactness and separation of $\ca C$, and continuity in $z$, $\sup_{z\in\ca C}|\!|\dot{\wt W}(t)R(z,t)|\!| < \infty$. The same applies for $\ddot{\wt W}(t)R(z,t)$.}
Differentiating the relation $\dot R(z) = - R(z) \dot{\wt W} R(z)$ once more (in time) gives
\be
\ddot R(z) = - R(z) \ddot{\wt W} R(z) + 2 R(z) \dot{\wt W} R(z)\dot{\wt W} R(z)\,.
\ee
We thus have (omitting the $z$ argument of $R$)
\be\label{dotFexpl}
\dot\F = \frac{\hbar}{2\pi i} \oint_{\ca C}\!dz \left( - \dot{\wt Q} R^2 \dot{\wt W} R + \wt Q \big(R \dot{\wt W} R \big)^2 + 2 \wt Q R^2 \dot{\wt W} R\dot{\wt W} R - \wt Q R^2 \ddot{\wt W} R \right) + \dot{\text{adj.}}
\ee
and consequently
\be
|\!| \dot\F |\!| \le 2\frac{\hbar}{2\pi} |\ca C| \left(|\!| \dot{\wt P} |\!| \kappa^2\alpha_1 + \kappa^2\alpha_1^2 + 2 \kappa^2\alpha_1^2 + \kappa^2\alpha_2 \right)
\ee
where the first $2$ appears because $|\!|\dot{\text{adj.}}|\!|$ is equal to the norm of the first term displayed in \eqref{dotFexpl}. 
We have also used that $\dot{\wt P} = - \dot{\wt Q}$.
\end{proof}
In fact $\alpha_2$ can be bounded by a simple quantity. In particular, specializing to the physical, microscopic Hamiltonian
\be\label{microHapp}
H = \frac{\pi_1^2 + \pi_2^2}{2m}\,
\ee
with $\pi_i = p_i - A_i(x)$, and using the canonical commutation relations, we find
\be
\ddot{\wt W}(0) = -\frac{\mu^2}{\hbar^2}[B,[B,H]] = \frac{\mu^2}{m} |\nabla B|^2\,.
\ee
Consequently, if we assume that $|\nabla B|^2$ is bounded (which is quite natural from assumption (i)---see footnote~\ref{assumpiapp}),
\be\label{alpha2bound}
\alpha_2 \le |\!|\ddot{\wt W}(0)|\!| \sup_{z\in\ca C} |\!|R(z;0)|\!| = \frac{\kappa\mu^2}{m} \, |\!| |\nabla B|^2 |\!| \,.
\ee
Putting it all together, we get
\begin{theorem}\label{theoremvareps2bound}
Under assumptions (i) and (ii), the specialization to the microscopic Hamiltonian \eqref{microHapp}, plus the assumption that $|\nabla B|^2$ is bounded, the constant $\varepsilon_2$ is bounded by
\be\label{vareps2exactbound}
\varepsilon_2 \le \frac{2\hbar|\ca C|}{\pi} \kappa^2 \left( \frac{\varepsilon\Delta}{2\hbar} \alpha_1 + 3\alpha_1^2 + \frac{\kappa\mu^2}{m} \, |\!| |\nabla B|^2 |\!| \right) + \frac{\Delta}{2\hbar} \varepsilon^2
\ee
for any allowed contour $\ca C$.
\end{theorem}
\begin{proof}
The result follows immediately from Lemma~\ref{lemmaFdotPP} and Lemma~\ref{lemmadotF}, the bound on $|\!| \dot{\wt P} |\!|$ resulting from \eqref{dotPQdotWPineq} and \eqref{QdotWPQBHP}, and the bound on $\alpha_2$ given in \eqref{alpha2bound}.
\end{proof}

Finally, let us consider the simple class of $B$ fields considered at the end of Sec.~\ref{nonuniformBmatch}, to gain some intuition on the bound for $\varepsilon_2$. Recall that it was assumed that $B = B(x_1)$ was quasi-uniform, behaving homogeneously under differentiation (i.e., if $B$ varies over a length scale $L$, then $\partial^sB/\partial x_1^s \sim B/L^s$), $\mu \sim \hbar n/m$ and $\Delta =: \hbar \Omega\,\, (\sim \hbar B/m)$. 
Also, the eigenstates of $W$ were assumed to be approximately $\Psi_{n,f}$, the Landau-gauge energy eigenstates, so $P \sim P_n$, where $P_n$ is the projector to the Landau level $n$.
As estimated in that section, the small parameter controlling the validity of the adiabatic bound is
\be
\varepsilon \sim n\rho_n \Big|\!\Big| \frac{\partial \log B}{\partial x_1} \Big|\!\Big| \sim \frac{n\rho_n}{L}\,.
\ee
In particular, $\varepsilon_1 \le 2 \varepsilon$.
Notice that if $\ca C$ is chosen as a circle centered at the origin, with radius $\sim \Delta/2$, we would have $|\ca C| \sim \pi\Delta$. Also, as this $\ca C$ would pass as close as $\Delta/2$ from the eigenvalues of $W$ (assumed to be concentrated in narrow clusters centered at integer multiples of $\Delta$),
\be
\kappa = \sup_{z\in \ca C} |\!|(W-z)^{-1}|\!| \sim (\Delta/2)^{-1}\,.
\ee
To roughly estimate $\alpha_1$, notice that for $B$ depending only on $x_1$,
\be
\dot{\wt W}(0) = -\frac{\mu}{m} \partial_1B \pi_1 + \frac{i\hbar\mu}{2m} \partial_1^2B\,.
\ee
so we have
\be\label{alpha1inter}
\alpha_1 = \sup_{z\in \bb C}|\!|\dot{\wt W}R(z)|\!| \le \frac{\mu}{m}|\!|\partial_1B|\!| \sup_{z\in \bb C}|\!|\pi_1 R(z) |\!| +  \frac{\hbar\mu\kappa}{2m} |\!|\partial_1^2B|\!|\,.
\ee
Since $|\!|\scr O|\!|^2 = |\!|\scr O^\dag \scr O|\!|$ for any bounded operator $\scr O$, we have 
\be
|\!|\pi_1 R(z) |\!| = \sqrt{|\!|(\pi_1 R(z))^\dag\pi_1 R(z)|\!|}
\ee
so we can instead estimate $|\!|(\pi_1 R(z))^\dag\pi_1 R(z)|\!|$.
Notice that for any vector $\psi \in \ca H$, 
\be
\la \psi|R(z)^\dag\pi_1^2 R(z)|\psi\ra \le  \la \psi|R(z)^\dag(2mH) R(z)|\psi\ra
\ee
and consequently $|\!|(\pi_1 R(z))^\dag\pi_1 R(z)|\!| \le |\!|R(z)^\dag(2mH) R(z)|\!|$.
Since $R(z)$ and $H$ are (approximately) diagonal in the Landau basis,
\be
|\!|R(z)^\dag(2mH) R(z)|\!| \sim \sup_{n'} \frac{2m E_{n'}}{|E_{n'} - \mu B - z|^2}\,. 
\ee
Notice that the numerator grows with $n'$ while the denominator grows with $(n')^2$ for large $n'$. Thus, since $z \sim \Delta/2$, we expect that the supremum should occur when $E_{n'} - \mu B \sim 0$, that is, $n' \sim n$. Now taking the supremum over $z\in \ca C$, we should have roughly
\be
\sup_{z\in \ca C}|\!|R(z)^\dag(2mH) R(z)|\!| \sim \frac{2m E_n}{(\Delta/2)^2}\,. 
\ee
If $n$ is not too small, so we can approximate $E_n \approx \hbar\Omega(n+1)$, we get
\be
\sup_{z\in \ca C}|\!|\pi_1 R(z) |\!| \lesssim \frac{2m\Omega\rho_n}{\Delta}
\ee
and thus, from \eqref{alpha1inter},
\be
\alpha_1 \lesssim 2\varepsilon \Omega
\ee
where the term coming from $|\!|\partial_1^2B|\!|$ was neglected for being of order $(\varepsilon^2/n^2)\Omega$.
Putting it all together in \eqref{vareps2exactbound}, we obtain
\be
\varepsilon_2 \lesssim 100 \varepsilon^2 \Omega\,.
\ee
Therefore, the linear-in-time contribution to the bound remains small, $\varepsilon_2t \ll 1$, 
as long as
\be
\Omega t \ll \frac{0.01}{\varepsilon^2}\,.
\ee
If $\varepsilon$ is very small, this allows $P(t)$ to be adiabatically conserved for many gyro periods.

\newpage
\bibliographystyle{JHEPs}
\bibliography{biblio}
\end{document}